\documentclass[thm-restate, english]{lipics-v2021}

\usepackage[T1]{fontenc}

\usepackage{amsmath, amsthm}
\usepackage{amssymb}
\usepackage{mathtools}
\usepackage{thm-restate, apptools}

\AtAppendix{\counterwithin{lemma}{section}}
\usepackage{enumerate}
\usepackage[shortlabels]{enumitem}
\usepackage{xcolor}
\usepackage{array}
\usepackage{xspace}
\usepackage{algorithm}
\usepackage[noend]{algpseudocode}
\usepackage{xparse}
\usepackage{microtype}
\usepackage{hyperref}
\usepackage{tikz,pgfplots}
\usetikzlibrary{automata, arrows.meta, positioning,decorations.markings,shapes.misc}
\usepackage{xfrac}
\usepackage{wrapfig}
\usepackage{caption,subcaption}
\usepackage{tgadventor}
\usepackage{mathrsfs}
\setlist[description]{leftmargin=1cm,labelindent=0.5cm}

\newcommand{\Surv}{\text{Pres}\xspace}
\newcommand{\Killer}{\text{Cons}\xspace}
\newcommand{\Pres}{\text{Pres}\xspace}
\newcommand{\Cons}{\text{Cons}\xspace}
\newcommand{\Min}{\text{Min}\xspace}
\newcommand{\Max}{\text{Max}\xspace}
\newcommand{\prefM}{\pi^{\leq m}}
\newcommand{\pref}[1]{\pi^{\leq #1}}
\newcommand{\NatInf}{\Nat^\infty}
\renewcommand{\path}{\textsf{path}}

\NewDocumentCommand{\Cand}{g}{\textsf{Cand}^{\IfNoValueTF{#1}{}{#1}}}
\NewDocumentCommand{\Allowd}{g}{\textsf{A}^{\IfNoValueTF{#1}{}{#1}}}
\NewDocumentCommand{\relativebud}{g}{\textsf{Trim}\IfNoValueTF{#1}{}{(#1)}}
\newcommand{\finprefix}{\textsf{fin}}

\newcommand{\stam}[1]{}

\newtheorem{problem}{Problem}

\newcommand{\tup}[1]{\ensuremath{\left(#1\right)}}
\newcommand{\set}[1]{\ensuremath{\left\lbrace #1 \right\rbrace}}

\newcommand{\zug}[1]{\langle #1 \rangle}
\NewDocumentCommand{\play}{g}{\textsf{play}\IfNoValueTF{#1}{}{(#1)}}
\NewDocumentCommand{\spare}{gg}{\triangledown_{\IfNoValueTF{#2}{\pi}{#2}}\IfNoValueTF{#1}{}{(#1)}}

\newcommand{\Nat}{\mathbb{N}}
\newcommand{\Natstr}{\mathbb{N}^*}
\newcommand{\Natstro}{\mathbb{N}^* \setminus \mathbb{N}}

\newcommand{\upperboundforloop}{\absolut{V}k\maxweight}
\NewDocumentCommand{\pathacyclic}{g}{\IfNoValueTF{#1}{\pi}{#1}_{\textsf{acycl}}}

\newcommand{\G}{\ensuremath{\mathcal{G}}}
\newcommand{\R}{\mathcal{R}}
\newcommand{\C}{\mathcal{C}}

\newcommand{\calO}{\mathcal{O}}

\newcommand{\A}{\ensuremath{\mathcal{A}}}
\renewcommand{\P}{\ensuremath{\mathit{Player}}}
\NewDocumentCommand{\lattice}{g}{\mathcal{F}_{\IfNoValueTF{#1}{\mathcal{G}}{#1}}}
\NewDocumentCommand{\auxop}{gg}{\mathbb{O}\IfNoValueTF{#1}{}{(#1)}\IfNoValueTF{#2}{}{(#2)}}
\NewDocumentCommand{\actop}{gg}{\tilde{\mathbb{O}}\IfNoValueTF{#1}{}{(#1)}\IfNoValueTF{#2}{}{(#2)}}

\NewDocumentCommand{\weightf}{g}{\mathsf{w}\IfNoValueTF{#1}{}{(#1)}}
\newcommand{\maxweight}{\mathsf{W}}

\NewDocumentCommand{\weightforpath}{g}{\rho_{\IfNoValueTF{#1}{}{#1}}}
\NewDocumentCommand{\sumofweights}{g}{\mathsf{sum}\IfNoValueTF{#1}{(\pi)}{(#1)}}
\NewDocumentCommand{\suff}{g}{\textsf{Suff}\IfNoValueTF{#1}{}{(#1)}}

\newcommand{\budgetset}{[k] \cup \{k+1\}}
\newcommand{\PO}{\ensuremath{\mathit{\P~1}}\xspace}
\newcommand{\PT}{\ensuremath{\mathit{\P~2}}\xspace}

\NewDocumentCommand{\PLi}{g}{\IfNoValueTF{#1}{\ensuremath{\mathit{\P~i}}\xspace}{\ensuremath{\mathit{\P~#1}}\xspace}}

\NewDocumentCommand{\allowedthresh}{g}{\textit{A}_{\thresh}\IfNoValueTF{#1}{}{(#1)}}

\NewDocumentCommand{\thresh}{gg}{\mathit{Th}_{\IfNoValueTF{#2}{}{#2}}\IfNoValueTF{#1}{}{(#1)}}

\NewDocumentCommand{\val}{gg}{\mathit{val}_{\IfNoValueTF{#2}{}{#2}}\IfNoValueTF{#1}{}{(#1)}}

\NewDocumentCommand{\config}{gg}{\tup{\ensuremath{\IfNoValueTF{#1}{v}{#1}, \IfNoValueTF{#2}{B}{#2}}}}

\NewDocumentCommand{\tbid}{gg}{b^{\IfNoValueTF{#2}{T}{#2}}_{\IfNoValueTF{#1}{}{#1}}}
\renewcommand{\set}[1]{\{ #1  \}}
\newcommand{\succb}[1]{#1 \oplus 0^*\xspace}
\newcommand{\succInt}[1]{\left(#1\right)^*}
\newcommand{\succbInt}[1]{#1^*}
\newcommand{\predb}[1]{#1 \ominus 0^*\xspace}
\newcommand{\predbInt}[1]{\left(#1-1\right)^*}
\newcommand{\eps}{\varepsilon}
\newcommand{\floor}[1]{\left\lfloor #1 \right\rfloor}
\newcommand{\mustar}{\mu^*}

\NewDocumentCommand{\sumT}{g}{\IfNoValueTF{#1}{\absolut{\thresh{\vplus}} + \absolut{\thresh{\vminus}}}{\absolut{#1(\vplus)} + \absolut{#1(\vminus)}}}
\NewDocumentCommand{\sumTf}{g}{\IfNoValueTF{#1}{\absolut{\thresh{\vplus}} + \absolut{\thresh{\vminus}}}{\absolut{#1(\vplus{#1})} + \absolut{#1(\vminus{#1})}}}
\NewDocumentCommand{\diffT}{g}{\IfNoValueTF{#1}{\absolut{\thresh{\vplus}} - \absolut{\thresh{\vminus}}}{\absolut{#1(\vplus)} - \absolut{#1(\vminus)}}}
\NewDocumentCommand{\diffTf}{g}{\IfNoValueTF{#1}{\absolut{\thresh{\vplus}} - \absolut{\thresh{\vminus}}}{\absolut{#1(\vplus{#1})} - \absolut{#1(\vminus{#1})}}}
\NewDocumentCommand{\vmin}{g}{\IfNoValueTF{#1}{v^-}{#1^-}}
\NewDocumentCommand{\vmax}{g}{\IfNoValueTF{#1}{v^+}{#1^+}}
\newcommand{\absolut}[1]{|#1|}
\newcommand{\neighbor}[1]{N(#1)}
\NewDocumentCommand{\Tmax}{g}{T^+\IfNoValueTF{#1}{}{(#1)}}
\NewDocumentCommand{\Tmin}{g}{T^-\IfNoValueTF{#1}{}{(#1)}}
\NewDocumentCommand{\optbidB}{gg}{\textit{Bid}\IfNoValueTF{#1}{}{(#1\IfNoValueTF{#2}{}{, #2})}}
\NewDocumentCommand{\muoptbid}{ggg}{\textit{Bid}_{\IfNoValueTF{#3}{\buildmu}{#3}}\IfNoValueTF{#1}{}{(#1\IfNoValueTF{#2}{}{, #2})}}
\NewDocumentCommand{\optbid}{gg}{\textsf{bid}^{\IfNoValueTF{#1}{\thresh}{#1}}_{\IfNoValueTF{#2}{}{#2}}}
\NewDocumentCommand{\optbidstar}{g}{\textit{Bid}_{\succb{T}}\IfNoValueTF{#1}{}{(#1)}}
\NewDocumentCommand{\diffbudget}{gg}{\triangledown_{\IfNoValueTF{#1}{\rho}{#1}}\IfNoValueTF{#2}{}{(#2)}}

\NewDocumentCommand{\advS}{g}{\textit{Adv}\IfNoValueTF{#1}{}{(#1)}}
\NewDocumentCommand{\prefplay}{g}{\IfNoValueTF{#1}{\rho_{\textit{pref}}}{{#1}_{\textit{pref}}}}
\NewDocumentCommand{\primethresh}{gg}{\textit{Th}'_{\IfNoValueTF{#2}{}{#2}}\IfNoValueTF{#1}{}{(#1)}}
\NewDocumentCommand{\muvertexmin}{g}{\textsf{V}_{\textsf{min}}\IfNoValueTF{#1}{}{(#1)}}
\NewDocumentCommand{\muvertexmax}{g}{\textsf{V}_{\textsf{max}}\IfNoValueTF{#1}{}{(#1)}}

\NewDocumentCommand{\vertexmin}{g}{\textsf{V}_{\textsf{min}}^{\IfNoValueTF{#1}{\thresh}{#1}}}

\NewDocumentCommand{\conThresh}{gg}{\mathit{Th}'}

\NewDocumentCommand{\preThresh}{gg}{\mathit{Th}^{\IfNoValueTF{#1}{}{#1}}}
\NewDocumentCommand{\Th}{gg}{\mathit{Th}}

\NewDocumentCommand{\reach}{g}{\ensuremath{\mathit{Reach}\IfNoValueTF{#1}{}{(#1)}}}

\NewDocumentCommand{\payoff}{g}{\textsf{Payoff}\IfNoValueTF{#1}{}{(#1)}}
\NewDocumentCommand{\meanpayoff}{g}{\textsf{Mean-Payoff}\IfNoValueTF{#1}{}{(#1)}}
\NewDocumentCommand{\energy}{g}{\textsf{Energy}\IfNoValueTF{#1}{}{(#1)}}
\NewDocumentCommand{\energyn}{gg}{\textsf{Energy}_{\IfNoValueTF{#1}{n}{#1}}\IfNoValueTF{#2}{}{(#2)}}
\NewDocumentCommand{\partialenergy}{gg}{\textsf{Energy}_{\IfNoValueTF{#1}{i}{#1}}\IfNoValueTF{#2}{}{(#2)}}

\NewDocumentCommand{\maxorzero}{gg}{\delta\IfNoValueTF{#1}{}{\left(#1, #2\right)}}
\NewDocumentCommand{\buildmu}{g}{\mu\IfNoValueTF{#1}{}{\left(#1\right)}}
\NewDocumentCommand{\buildmun}{gg}{\mu_{\IfNoValueTF{#1}{}{#1}}\IfNoValueTF{#2}{}{(#2)}}

\NewDocumentCommand{\triumph}{g}{\textit{trump}\IfNoValueTF{#1}{}{(#1)}}
\NewDocumentCommand{\trump}{g}{\textit{trump}\IfNoValueTF{#1}{}{(#1)}}

\NewDocumentCommand{\minenergy}{gggg}{\min_{v' \in N(\IfNoValueTF{#2}{v}{#2})}\maxorzero{\buildmun{\IfNoValueTF{#1}{n}{#1}}{\zug{v', \IfNoValueTF{#3}{B}{#3} - \IfNoValueTF{#4}{b}{#4}}}}{\weightf{\IfNoValueTF{#2}{v}{#2}, v'}}}
\NewDocumentCommand{\maxenergy}{gggg}{\max_{v' \in N(\IfNoValueTF{#2}{v}{#2})}\maxorzero{\buildmun{\IfNoValueTF{#1}{n}{#1}}{\zug{v', \IfNoValueTF{#3}{B}{#3} + \triumph{\IfNoValueTF{#3}{B}{#3}, \IfNoValueTF{#4}{b}{#4}}}}}{\weightf{\IfNoValueTF{#2}{v}{#2}, v'}}}

\NewDocumentCommand{\minenergymu}{ggg}{\min_{v' \in N(\IfNoValueTF{#1}{v}{#1})}\maxorzero{\buildmu{\zug{v', \IfNoValueTF{#2}{B}{#2} - \IfNoValueTF{#3}{b}{#3}}}}{\weightf{\IfNoValueTF{#1}{v}{#1}, v'}}}
\NewDocumentCommand{\maxenergymu}{ggg}{\max_{v' \in N(\IfNoValueTF{#1}{v}{#1})}\maxorzero{\buildmu{\zug{v', \IfNoValueTF{#2}{B}{#2} + \triumph{\IfNoValueTF{#2}{B}{#2}, \IfNoValueTF{#3}{b}{#3}} }}}{\weightf{\IfNoValueTF{#1}{v}{#1}, v'}}}

\NewDocumentCommand{\nobidminenergymu}{gg}{\min_{v' \in N(\IfNoValueTF{#1}{v}{#1})}\maxorzero{\buildmu{v', \IfNoValueTF{#2}{B}{#2}}}{\weightf{\IfNoValueTF{#1}{v}{#1}, v'}}}
\NewDocumentCommand{\nobidmaxenergymu}{gg}{\max_{v' \in N(\IfNoValueTF{#1}{v}{#1})}\maxorzero{\buildmu{v', \IfNoValueTF{#2}{B}{#2}}}{\weightf{\IfNoValueTF{#1}{v}{#1}, v'}}}

\NewDocumentCommand{\fminenergy}{gggg}{\textsf{f}_{\textsf{min}}(\IfNoValueTF{#1}{n}{#1}, \IfNoValueTF{#2}{v}{#2}, \IfNoValueTF{#3}{B}{#3}, \IfNoValueTF{#4}{b}{#4})}
\NewDocumentCommand{\fmaxenergy}{gggg}{\textsf{f}_{\textsf{max}}(\IfNoValueTF{#1}{n}{#1}, \IfNoValueTF{#2}{v}{#2}, \IfNoValueTF{#3}{B}{#3}, \IfNoValueTF{#4}{b}{#4})}

\NewDocumentCommand{\gminenergymu}{ggg}{\textsf{g}_{\textsf{min}}\left(\IfNoValueTF{#1}{v}{#1}, \IfNoValueTF{#2}{B}{#2}, \IfNoValueTF{#3}{b}{#3}\right)}
\NewDocumentCommand{\gmaxenergymu}{ggg}{\textsf{g}_{\textsf{max}}\left(\IfNoValueTF{#1}{v}{#1}, \IfNoValueTF{#2}{B}{#2}, \IfNoValueTF{#3}{b}{#3}\right)}

\NewDocumentCommand{\T}{gg}{\textit{T}_{\IfNoValueTF{#2}{}{#2}}\IfNoValueTF{#1}{}{(#1)}}
\NewDocumentCommand{\tminus}{ggg}{{\IfNoValueTF{#2}{\thresh}{#2}}^-_{\IfNoValueTF{#3}{}{#3}}\IfNoValueTF{#1}{}{(#1)}}
\NewDocumentCommand{\tplus}{ggg}{{\IfNoValueTF{#2}{\thresh}{#2}}^+_{\IfNoValueTF{#3}{}{#3}}\IfNoValueTF{#1}{}{(#1)}}
\NewDocumentCommand{\vminus}{g}{v^{-}_{\IfNoValueTF{#1}{}{#1}}}
\NewDocumentCommand{\vplus}{g}{v^{+}_{\IfNoValueTF{#1}{}{#1}}}

\NewDocumentCommand{\sumvT}{g}{\absolut{\IfNoValueTF{#1}{\thresh{\vplus}}{#1(\vplus)}} + \absolut{\IfNoValueTF{#1}{\thresh{\vminus}}{#1(\vminus)}}}
\NewDocumentCommand{\diffvT}{g}{\absolut{\IfNoValueTF{#1}{\thresh{\vplus}}{#1(\vplus)}} - \absolut{\IfNoValueTF{#1}{\thresh{\vminus}}{#1(\vminus)}}}
\NewDocumentCommand{\fracvT}{m}{\frac{#1}{2}}
\NewDocumentCommand{\floorvT}{m}{\floor{\fracvT{#1}}}
\NewDocumentCommand{\minvT}{g}{\IfNoValueTF{#1}{\thresh{\vminus}}{#1(\vminus)}}

\NewDocumentCommand{\optbidthresh}{g}{\textsf{cbid}^{\IfNoValueTF{#1}{}{#1}}}
\NewDocumentCommand{\optbidf}{gg}{\textsf{bid}^{\IfNoValueTF{#2}{\thresh}{#2}}\IfNoValueTF{#1}{(v)}{(#1)}}

\NewDocumentCommand{\presagn}{gg}{{\IfNoValueTF{#2}{\sigma}{#2}}_{\textsf{agn}}\IfNoValueTF{#1}{}{(#1)}}

\NewDocumentCommand{\consagn}{gg}{{\IfNoValueTF{#2}{\tau}{#2}}_{\textsf{agn}}\IfNoValueTF{#1}{}{(#1)}}

\NewDocumentCommand{\agnosticstr}{gg}{{\IfNoValueTF{#2}{\sigma}{#2}}_{\textsf{agn}}\IfNoValueTF{#1}{}{(#1)}}
\NewDocumentCommand{\sagn}{g}{{\IfNoValueTF{#1}{\sigma}{#1}}_{\textsf{agn}}}
\NewDocumentCommand{\tagn}{g}{{\IfNoValueTF{#1}{\tau}{#1}}_{\textsf{agn}}}
\NewDocumentCommand{\sVI}{g}{{\IfNoValueTF{#1}{\sigma}{#1}}_{\textsf{VI}}}
\NewDocumentCommand{\sVIn}{gg}{\IfNoValueTF{#1}{\sigma}{#1}_{\textsf{VI}, #2}}

\NewDocumentCommand{\posstr}{gg}{{\IfNoValueTF{#2}{\sigma}{#2}}_{\textsf{pos}}\IfNoValueTF{#1}{}{(#1)}}
\newcommand{\NP}{\textsf{NP}}

\newcommand{\coNP}{\textsf{coNP}\xspace}

\newcounter{cSS}
\newcounter{cGA}

\def\SS{\@ifstar\SSi}

\def\GA{\@ifstar\GAi}

\newcommand{\SSn}[1]{}
\newcommand{\GAn}[1]{}

\newcommand{\SSi}[1]{}
\newcommand{\GAi}[1]{}

\NewDocumentCommand{\ewin}{g}{e_{\text{win}}^{\IfNoValueTF{#1}{}{#1}}}
\NewDocumentCommand{\elose}{g}{e_{\text{lose}}^{\IfNoValueTF{#1}{}{#1}}}

\NewDocumentCommand{\enxt}{g}{e_{\text{next}}^{\IfNoValueTF{#1}{}{#1}}}

\NewDocumentCommand{\ebid}{g}{e_{\text{bid}}^{\IfNoValueTF{#1}{}{#1}}}
\NewDocumentCommand{\sqglymu}{g}{\tilde{\mu}_{\IfNoValueTF{#1}{n}{#1}}}

\tikzstyle{rond1}=[draw,circle,minimum size=1.5mm,inner sep=1pt]
\tikzstyle{rond2}=[draw,circle,minimum size=2.5mm,inner sep=1pt]
\tikzstyle{rond5}=[draw,circle,minimum size=5mm,inner sep=1pt]
\tikzstyle{rond6}=[draw,circle,minimum size=6mm,inner sep=1pt]
\tikzstyle{rond7}=[draw,circle,minimum size=6mm,inner sep=1pt]
\tikzstyle{rond}=[draw,circle,minimum height=7mm]
\tikzstyle{rect} = [draw, rectangle, minimum size = 5mm, inner sep = 1pt]

\AtAppendix{\counterwithin{claim}{section}}

\makeatletter
\renewcommand\subsubsection{\@startsection{subsubsection}{3}{\z@}%
	{-18\p@ \@plus -4\p@ \@minus -4\p@}%
	{0.5em \@plus 0.22em \@minus 0.1em}%
	{\normalfont\normalsize\bfseries\boldmath}}
\makeatother
\title{Mean-payoff and Energy Discrete-Bidding Games}

\author{Guy Avni}{Department of Computer Science, University of Haifa, Israel}{gavni@cs.haifa.ac.il}{}{}
\author{Suman Sadhukhan}{Department of Computer Science, University of Haifa, Israel}{suman.sadhukhan00@gmail.com}{}{}
\authorrunning{G. Avni and S. Sadhukhan}
\Copyright{Guy Avni and Suman Sadhukhan}



\begin{document}
	\nolinenumbers
	\maketitle
	
	\begin{abstract}
A \emph{bidding} game is played on a graph as follows. A token is placed on an initial vertex and both players are allocated budgets. 
In each turn, the players simultaneously submit bids that do not exceed their available budgets, the higher bidder moves the token, and pays the bid to the lower bidder. 
We focus on \emph{discrete}-bidding, which are motivated by practical applications and restrict the granularity of the players' bids, e.g, bids must be given in cents. 
We study, for the first time, discrete-bidding games with {\em mean-payoff} and {\em energy} objectives. In contrast, mean-payoff {\em continuous}-bidding games (i.e., no granularity restrictions) are understood and exhibit a rich mathematical structure. The {\em threshold} budget is a necessary and sufficient initial budget for winning an energy game or guaranteeing a target payoff in a mean-payoff game. We first establish existence of threshold budgets; a non-trivial property due to the concurrent moves of the players. Moreover, we identify the structure of the thresholds, which is key in obtaining compact strategies, and in turn, showing that finding threshold is in \NP~and \coNP  even in succinctly-represented games.
	\end{abstract}
	
\section{Introduction}
Two-player {\em graph games} constitute a fundamental model with applications in {\em reactive synthesis}~\cite{PR89} and multi-agent systems~\cite{AHK02}, and a deep connection to foundations of logic \cite{Rab69}.
A game is played on a graph as follows. A token is placed on a vertex and the players move the token throughout the graph to generate an infinite path~({\em play}).
Two orthogonal characterizations for graph games are (1)~the {\em mode} by which the players move the token, e.g., in {\em turn-based} games, the players alternate turns in moving the token, and (2)~the players' objectives, which determine the winner or utilities in a play. 

We study {\em bidding games}~\cite{LLPSU99,LLPU96} in which an auction (bidding) determines which player acts in each turn: both players are allocated initial budgets, and in each turn, they simultaneously submit bids that do not exceed their budgets, the higher bidder moves the token, and pays their bid to the opponent. 
{\em Discrete bidding}~\cite{DP10}, which is the focus of this paper, impose granularity restrictions: budgets are given in ``cents'' and the smallest positive bid is a ``cent''. In contrast, {\em continuous bidding} allows arbitrarily small bids. 
We study, for the first time, discrete-bidding games with {\em mean-payoff} and {\em energy} objectives (formally defined in Sec.~\ref{sec:def-objectives}).


The motivation for discrete bidding is practical; every practical application requires some granularity restriction. 
We describe examples of applications of bidding games.

{\bf Auction-based scheduling}~\cite{AMS24} applies bidding games in a ``decoupled'' synthesis procedure: given two objectives $\psi_1$ and $\psi_2$, the idea is to independently construct two policies $f_1$ and $f_2$, where policy $f_i$ only aims to satisfy $\psi_i$, for $i \in \set{1,2}$, and to compose $f_1$ with $f_2$ at runtime using a bidding for who chooses the action in each turn. For example, consider the task of finding a plan for a robot waiter, where $\psi_1$ specifies delivering food and $\psi_2$ specifies collecting dishes. The challenge in~\cite{AMS24} is to ensure that the composition of $f_1$ and $f_2$ satisfies $\psi_1 \wedge \psi_2$ even though they are constructed independently. Our work enables a combination of discrete bidding, which, again, is necessary in practice, with quantitative specifications. For example, consider the task of finding a plan for a patrolling robot, where $\psi_i$ specifies maximizing the time spent at location $t_i$, for $i \in \set{1,2}$.

{\bf Fair allocation of resources} is timely (e.g.,~\cite{ABFV22,ALMW22}). The goal is to allocate a collection of items 
to agents in a {\em fair} manner. A mechanism based on bidding games is both natural and useful~\cite{MKT18,BEF21}: each agent is allocated an initial budget, and the items are auctioned sequentially. 
Repeated applications of this mechanism are used for ongoing allocation of resources~\cite{GBI21}, e.g., daily allocation of GPU time to users.
Mean-payoff objectives naturally specify a strong notion of fairness. 
For example, ``in the long-run, the users are scheduled for the same duration of time''. As another example, online advertisement platforms hold auctions for allocation of ad slots~\cite{Mut09}, then an advertiser might aim to maximize the long-run average daily exposure, again,  a mean-payoff objective. 

\paragraph*{Previous results}
\noindent{\bf Continuous-bidding games.}
We briefly survey relevant literature. Continuous-bidding games with {\em reachability} objectives were studied in~\cite{LLPU96,LLPSU99} and {\em parity} objectives in~\cite{AHC19}.
The central quantity in these games is the {\em threshold budget}, which is roughly a necessary and sufficient initial budget for winning the game. Thresholds satisfy the {\em average property}: the threshold in a vertex is the average of two of its neighbors. This leads to an equivalence between bidding games and a class of {\em stochastic games}~\cite{Con92} called {\em random-turn games}~\cite{PSSW09}. 

Mean-payoff continuous-bidding games have been extensively studied. A generalized equivalence with random-turn games was shown in~\cite{AHC19}. It implies that, somewhat surprisingly, in a strongly-connected game, the optimal payoff depends only on the structure of the game; that is, a player {\em cannot} guarantee a higher payoff, given a higher initial budget. Moreover, intricate equivalences between mean-payoff continuous-bidding games and random-turn games were shown for various bidding mechanisms~\cite{AHI18,AHZ21,AJZ21}; including mechanisms for which equivalences in finite-duration games are not known and unlikely to exist. 

\noindent{\bf Discrete-bidding games}
 are far less understood than their continuous-bidding counterparts. So far, only qualitative objectives have been studied. 
Reachability discrete-bidding games were studied in~\cite{DP10}. It was shown that threshold budgets exist and satisfy a discrete version of the average property. 
Existence of thresholds in infinite-duration games was established in~\cite{AAH21}, but the question of whether thresholds satisfy the average property was left open. 
Moreover, in both papers, the algorithms for finding thresholds are exponential when the budgets are given succinctly\footnote{More formally, the total budget, later denoted $k$, is part of the representation of a discrete-bidding game, and we assume throughout the paper that $k$ is represented in binary.};
in practice, a succinct representation is appealing since large budgets imply reduced granularity constraints and high precision. 
Recently, both problems were solved in~\cite{AS25}: thresholds in parity discrete-bidding games were shown to satisfy the average property and based on this,  finding thresholds was shown to be in NP and coNP. 

\smallskip
Previous results leave a gap in our understanding of mean-payoff bidding games:
under continuous-bidding, the literature is a rich, whereas under discrete-bidding, even the basic properties were not known.

\paragraph*{Our results}
The central quantity that we study in mean-payoff games is threshold budgets, which we define as follows: for a target payoff $c$ for \Max, the threshold budget is necessary and sufficient for guaranteeing payoff $c$. Before elaborating on our results, we point to an inherent distinction between mean-payoff discrete- and continuous-bidding games: in strongly-connected games, under continuous-bidding, \Max can guarantee the same payoff for every initial budget, whereas the following example shows that this is not the case under discrete bidding.
\begin{example}
\label{ex:example1}
Consider the game that is depicted in Fig.~\ref{fig:multoptpayoff}. 
A {\em configuration}\footnote{Later we will specify only one of the players' budgets in a configuration and the other budget is implicit.} $\zug{v, B_\Max, B^*_\Min}$ means that the token is placed on $v$, \Max and \Min's budgets are respectively $B_\Max,B_\Min \in \Nat$. Tie breaking is resolved as follows, based on~\cite{DP10}. One of the players (in this case \Min) holds the {\em tie-breaking advantage}, marked with $*$. \Min chooses a bid $b \leq B_\Min$ and chooses whether she uses the advantage: (i)~she uses it by bidding $b^*$, then if \Max bids $b' \leq b$, she wins and pays \Max $b^*$, and (ii)~she does not use it by bidding $b$, then \Max wins if he bids $b' \geq b$.

We describe {\em optimal plays} that arise from optimal play of both players. Note that upon winning a bidding, it is optimal for \Min and \Max to proceed left and right, respectively. We write $c \xrightarrow{b_\Max, b_\Min} c'$ to indicate that $c'$ results from configuration $c$ when the players respectively bid $b_\Max$ and $b_\Min$.
First, the optimal play from $\zug{v_0, 1^*, 0}$ is $\zug{v_0, 1^*, 0} \xrightarrow{0^*, 0} \zug{v_2, 1, 0^*} \xrightarrow{0, 0^*} \zug{v_0, 1^*, 0} \xrightarrow{0^*, 0} \zug{v_2, 1, 0^*} \ldots$ with corresponding path $(v_0, v_2)^\omega$ whose payoff is~$\frac{3}{2}$. Second, consider the configuration $\zug{v_0, 0^*, 1}$, in which \Max has less budget. The optimal play is $\zug{v_0, 0^*, 1} \xrightarrow{0^*, 1} \zug{v_1, 1^*, 0} \xrightarrow{0^*, 0} \zug{v_0, 1, 0^*} \xrightarrow{1, 0^*} \zug{v_2, 0, 1^*} \xrightarrow{0, 0^*} \zug{v_0, 0^*, 1} \ldots$ with corresponding path $(v_0, v_1, v_0, v_2)^\omega$ with payoff~$\frac{1}{4}$. 
This is a key distinction from continuous-bidding. There, the optimal payoff that \Max can guarantee does not depend on the initial budget; it is roughly $1$ in this game, for {\em every} positive initial budget.
\end{example}

\begin{figure}
	\centering
	\begin{tikzpicture}
		\draw (0,0) node[rond5] (0) {$v_0$};
		\draw (4,0) node[rond5] (1) {$v_2$};
		\draw (-4,0) node[rond5] (2) {$v_1$};

		\draw (0) edge[->, bend left = 10] node[sloped, above, pos= 0.5] {$ 3 $} (1);
		\draw (1) edge[->, bend left = 10] node[sloped, below, pos = 0.5] {$ 0 $}(0);
		\draw (2) edge[->, bend right = 10] node[sloped, below, pos= 0.5] {$ 0 $} (0);
		\draw (0) edge[->, bend right = 10] node[sloped, above, pos = 0.5] {$ -2 $}(2);
		
		\draw (1) edge[->, loop right] node[above] {$ 5 $} (0);
		\draw (2) edge[->, loop left] node[above] {$ 0 $} (0);
	\end{tikzpicture}
	\caption{A mean-payoff discrete bidding game where optimal payoff depends on the initial budget}\label{fig:multoptpayoff}
\end{figure}
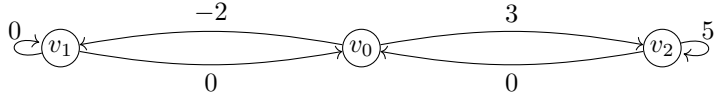

Technically, we study energy games, and the results directly apply to mean-payoff bidding games, similar to turn-based games (e.g.,~\cite{BFLMS08}).
An energy bidding game is played between {\em Consumer} (\Cons) and {\em Preserver} (\Pres) on a weighted graph. A play $\pi$ corresponds to an infinite sequence of weights. Fix an initial energy $M \in \Nat$. \Cons's goal is to ``consume'' the energy; formally, \Cons wins iff there is a prefix of length $m$ such that $M + \sumofweights{\pref{m}} < 0$. 

We study two types of thresholds in energy games. First, we show existence of {\em energy thresholds}; for every initial vertex $v$, for every initial budget $B$, we show that there exists an initial energy, denoted $\energy(v,B)$ that is both necessary and sufficient for \Pres to guarantee winning. We point out that existence of energy thresholds implies {\em determinacy}, namely from each initial configuration, one of the players has a (pure) winning strategy. Bidding games are formally a subclass of {\em concurrent games}~\cite{AHK02}, the latter are not determined; e.g., neither player has a winning strategy in ``matching pennies''. Still, we show that energy bidding games are a determined subclass of concurrent games (see also~\cite{AAH21,BBR21}). 
Second, we define threshold budgets in energy bidding games. Before describing the definition, note that at $v$, there could be a budget $B$ for which \Pres loses with every initial energy. A simple example is a sink with a negative self loop. In such cases, $\energy(v, B) = \infty$. Further note that $\energy(v, B)$ increases as $B$ decreases. We define the threshold in a vertex $v$, denoted $\Th(v)$, as the minimal budget $B$ such that $\energy(v, B)$ is finite. 


A key result in the paper shows that $\Th$ satisfies the average property. 
This is an important ingredient in constructing concise {\em budget agnostic} strategies, which ignore ``excess'' budget; more formally, at vertex $v$, for budget $B \geq \Th(v)$, a budget agnostic strategy acts in  $B$ and $B^*$ as if the budget is $\Th(v)$ and $\Th(v)^*$, respectively. 
Existence of winning budget agnostic strategies is key in proving that finding threshold budgets\footnote{Formally, given a game, a vertex $v$, and a value $t$, decide whether $\Th(v) \geq t$.} is in NP and coNP. 

\noindent{\bf Comparison with previous works.}  
We establish existence of thresholds via a value-iteration algorithm, similar to the approach in previous works. 
However, previously, establishing the average property and constructing budget-agnostic strategies was a simple byproduct~\cite{DP10,AS25}, and in our case, it is significantly more challenging. In fact, we show in Ex.~\ref{ex:example2}, that our algorithm produces strategies that are {\em not} budget agnostic. 
We circumvent this challenge by developing a novel proof structure.  We first establish the average property directly (Thm.~\ref{thm:preThreshAverage}), from which we obtain ``budget agnostic bids''. For both \Cons and \Pres separately (Sections~\ref{sec:Pres-bud-agn} and~\ref{sec:Cons-bud-agn}), we identify certain scenarios in which the value-iteration algorithm's winning strategies match the budget agnostic bids. This is particularly challenging for \Cons, for which we do not have an explicit strategy. 
These observations are used to show that eventually, our strategies maintain energy invariants like the value-iteration strategies. 

Second, interestingly, our constructions are conceptually very different from strategies in mean-payoff continuous-bidding games. There, strategies are {\em not} budget agnostic: \Max maintains an invariant between accumulated energy and budget so that when the energy increases, both \Max's budget and his bids decrease, thus the strategy is quite the opposite of being budget agnostic. Our strategies are conceptually closer to constructions in turn-based games (e.g.,~\cite{BFLMS08}) in that they guarantee that eventually the play avoids ``bad cycles'', namely cycles with average weight lower than the target payoff.

		\section{Preliminaries}
		\label{sec:prelim}
	We denote \(\Nat\) as the set of natural numbers including \(0\) and $\NatInf = \Nat \cup \{\infty\}$. 
	
	\paragraph*{Concurrent games}
	A bidding game is formally, a succinctly represented {\em concurrent game}. We define concurrent games, and then describe the concurrent game that a bidding game corresponds to. 
	
	Intuitively, a concurrent game is a two-player game that is played on a graph, where each vertex is associated with a set of {\em allowed actions} for each player. The game proceeds as follows. A token is initially placed on a vertex. In each turn, the players simultaneously choose an allowed action, and their joint actions determine the next vertex the token moves to. This generates an infinite path, which determines the winner of the game. 

Formally a concurrent game is played on an arena \(\zug{A, Q, \lambda, \delta}\), where \(A\) is a set of actions, \(Q\) is a set of states, \(\lambda: Q \times \{1, 2\} \rightarrow 2^A \setminus \emptyset\) specifies the allowed actions for each player at a state, and the transition function is \(\delta: Q \times A \times A \rightarrow Q\).  The {\em neighbors} of $q \in Q$ are $N(q) = \set{q' \in Q: \exists a_1, a_2 \in A \text{ such that } q' \in \delta(q, a_1, a_2)}$. For an infinite path $\pi = q_0,q_1,\ldots$, we denote the prefix of length $m \in \Nat$ by $\prefM = q_0,\ldots, q_m$. 
	
	A {\em strategy} is intuitively a \emph{recipe} for playing the game. For $i \in \set{1,2}$, a strategy for \PLi is a function \(\sigma_i: Q^* \rightarrow A\), which prescribes which action to take given a {\em history} of the game. We restrict to strategies that choose only allowed actions, that is for a history \(h\in  Q^*\), a \PLi strategy chooses an action $a_i \in \lambda(q, i)$. The {\em play} that two strategies \(\sigma_1\) and \(\sigma_2\) and an initial vertex \(q_0\) give rise to, denoted  \(\play(q_0, \sigma_1, \sigma_2)\), is defined inductively as follows.
	The play starts from \(q_0\). Suppose that the prefix $\prefM$ of length \(m \geq 1\) of \(\play(q_0, \sigma_1, \sigma_2)\) is defined, \PLi takes action \(a_i^j = \sigma_i(\prefM \cdot q_j)\), for $i \in \set{1,2}$, then the next state is \(q_{j+1} = \delta(q_j, a_1^j, a_2^j)\).  
	We say that a play $\pi$ is {\em consistent} with $\sigma_1$ from $q_0$ if there is $\sigma_2$ such that $\pi = \play(q_0, \sigma_1, \sigma_2)$, and similarly for $\sigma_2$. 
	
	For \(i \in \{1, 2\}\), we say \PLi \emph{controls} a state \(q \in Q\), if intuitively, the next state is determined solely based on their choice of action. 
	Formally, \PO controls state \(q\) if for any \(a_1 \in \lambda(q, 1)\) and \(a_2, a_2' \in \lambda(q, 2)\), we have \(\lambda(q, a_1, a_2) = \lambda(q, a_1, a_2')\). 
	The definition is dual for \PT. 
	\emph{Turn-based} games are a special case of concurrent games, where each state \(q\) is controlled by one of the players. 
	Note that a concurrent game which is not turn-based may still have some states which are controlled by one of the players.

		\paragraph*{Bidding games}
	A discrete bidding game is played on an arena \(\zug{V, E, k}\), where \(V\) is the set of vertices, \(E\) is the set of edges, and \(k \in \Nat\) is the total budget in the game.  The neighbors of a vertex \(v\), denoted \(\neighbor{v}\), are \(\neighbor{v} = \{u: (v, u) \in E\}\). 
	
	We introduce notation to formalize the tie-breaking mechanism, called {\em advantage-based} tie-breaking~\cite{DP10}.
	We denote the advantage with \(*\). Thus, when a player's budget is \(B^*\), this means that the player has a budget of \(B \in \Nat\) and holds the advantage. 
	Similarly, when we say that a player bids \(b^*\), we mean that they bid \(b \in \Nat\), and in case a tie occurs, they will use the advantage. Denote  \(\Natstr = \{0, 0^*, 1, 1^*, \ldots \}\) and \([k]=\{0, 0^*, 1, 1^*, \ldots k, k^*\}\). The integral part of \(B \in \Natstr\) is denoted \(|B|\). We define two operators \(\oplus\) and \(\ominus\) over \(\Natstr\). We describe how the operators are used. Suppose that \PO's budget is \(m^*\) and the players bid $b_1$ and $b_2$, respectively. Recall that the higher bidder pays the lower bidder. Thus, when $b_1 > b_2$, \PO's budget is updated to \(m^* \ominus b_1\), and when $b_2 > b_1$, \PO's budget is updated to \(m^* \oplus b_2\). Note that \(x^* \oplus y^*\) and \(x \ominus y^*\), for $x,y \in \Nat$, will not occur in the setting above, still it is useful to define both for convenience and completion. Formally,

	\begin{definition}{\bf (\(\oplus\) and \(\ominus\) operators).}
		For \(x, y \in  \Nat\), we define \(x^* \oplus y = x \oplus y^* = (x+y)^*\), \(x \oplus y = x + y\) and \(x^* \oplus y^* = x+y+1\). 
		For \(x, y \in \Nat\), we define \(x^* \ominus y = (x -y)^*\), \(x^* \ominus y^* = x - y\) and \(x \ominus y = x - y\). 
		Finally, \(x \ominus y^* = (x - y - 1)^*\). 
	\end{definition}
	
Consider the natural order $\prec$ over $\Natstr$ as $0 \prec 0^* \prec 1 \prec 1^* \prec \ldots$. We will frequently use the {\em successor} and {\em predecessor} according to this order: for $B \in \Natstr$, the successor of $B$ is $\succb{B}$ and its predecessor is $\predb{B}$.

\paragraph{Bidding games as concurrent games}
Consider an arena $\A = \zug{V, E, k}$ of a bidding game. 
	The {\em configurations} of $\A$ are \(\C = \{\zug{v, B}: v \in V, B \in \budgetset\}\), where \(\zug{v, B} \in \C\) means that the token is placed at vertex \(v\) and \PO's current budget is \(B\). Implicitly, \PT's budget is \(k^* \ominus B\). 
	The arena of the corresponding concurrent game is \(\zug{[k] \times V, \C, \lambda, \delta}\), where we define the allowed actions $\lambda$ and transitions $\delta$ next. 
Choosing an action $\zug{b, v} \in [k] \times V$ corresponds to bidding $b$ and moving to $v$ upon winning the bidding. Consider a configuration $\zug{v, B}$.
Define \(\lambda(\zug{v, B}, 1) = \{0, \ldots, B\} \times \neighbor{v}\), that is \PO must choose a bid within his budget and must move to a neighbor of $v$ upon winning the bidding. Similarly, \(\lambda(\zug{v, B}, 2) = \{0, \ldots, k^* \ominus B\} \times \neighbor{v}\). We define \(\delta\) next.
	Suppose that the token is placed on \(c = \zug{v, B}\), and \PLi chooses \(\zug{b_i, v_i}\), for \(i = \{1, 2\}\). 
	If \(b_1 > b_2\), then the token moves to \(\zug{v_1, B \ominus b_1}\); that is, \PO wins the bidding, pays \PT, and moves the token. 
	Dually, if $b_2 > b_1$, then the token moves to $\zug{v_2, B\oplus b_2}$. The remaining case is \(b_1 = b_2\). Note that this occurs only when the player with the advantage does not use it. In this case, the other player wins the bidding, and the token moves as in the above. 	
	
	As above, two strategies $\sigma_1$ and $\sigma_2$ and an initial configuration $c_0 = \zug{v_0, B_0}$ give rise to an infinite play $\play(c_0, \sigma_1, \sigma_2)= c_0, c_1,\ldots \in \C^\omega$. 
We use $\path(c_0, \sigma_1, \sigma_2) = v_0, v_1,\ldots$ to refer to the path in $\zug{V, E}$ that corresponds to the play, assuming $c_j = \zug{v_j, B_j}$, for $j \geq 0$.
	
\begin{remark}
[Representation size] Consider an arena \(\A = \zug{V, E, k}\) of a bidding game. We assume that $k$ is encoded in binary. Thus, the size of $\A$ is \(O(\absolut{V} + \absolut{E} + \log{k})\). Note the size (number of configurations) of the explicit concurrent game that corresponds to $\A$ is $k \cdot |V|$, thus exponentially larger than $\A$. 
\end{remark}	

\subsection{Mean-payoff and energy bidding games }
\label{sec:def-objectives}
Both mean-payoff and an energy bidding games are played on an arena  \(\zug{V, E,k, \weightf}\), where \(\zug{V, E, k}\) is as above, and $\weightf: E \rightarrow \set{-W,\ldots,W}$ is a function that assigns integer weights to edges, $W$ being the largest absolute weight. 
Consider an infinite path $\pi = v_0, v_1,\ldots$. We call the sum of weights traversed by the prefix $\prefM$ as its {\em energy}, $\sumofweights{\prefM} = \sum_{0 \leq j < m} \weightf(\zug{v_j, v_{j+1}})$. 
We define below which player wins in $\pi$ under the two objectives, and later, in Thm.~\ref{thm:MP-energy}, we will show an equivalence between the two objectives.  
	
	\noindent{\bf Energy objective.}
We call the players in an energy game  \emph{preserver} (\Pres) and \emph{consumer} (\Cons). Intuitively, the game starts with an initial energy $M \in \Nat$, \Cons's objective is to drop the energy below $0$, and \Pres wins otherwise, namely if the energy stays non-negative throughout the whole play. Formally, for an initial energy $M \in \Nat$, \Cons~{\em $M$-wins} $\pi$ if there is $m \in \Nat$ such that $M + \sumofweights{\prefM} < 0$ and \Pres~{\em $M$-wins} $\pi$ if for all $m \in \Nat$, we have $M + \sumofweights{\prefM} \geq 0$.
	We say that \Surv $M$-wins from a configuration $\zug{v,B}$ if she has a strategy $\sigma$ such that for every \Killer strategy $\tau$, \Surv $M$-wins $\path(\zug{v, B}, \sigma, \tau)$, and the definition for \Cons is dual.

	\noindent{\bf Mean-payoff objective.}
We call the players in a mean-payoff game {\em maximizer} (\Max) and {\em minimizer} (\Min). The {\em payoff} of an infinite path, which is Max's reward and Min's cost, is the long-run average of the weights traversed. Formally, we define $\meanpayoff{\pi} = \liminf_{m \rightarrow \infty} \frac{1}{m} \sumofweights{\prefM}$. Max wins in $\pi$ if $\meanpayoff{\pi} \geq 0$ and Min wins $\pi$ if $\meanpayoff{\pi} < 0$. 
	\Max wins from a configuration \(\zug{v, B}\) if he has a strategy $\sigma$ such that for every \Min strategy $\tau$, \Max wins $\path(\zug{v, B}, \sigma, \tau)$, and the definition for \Min is dual.

\section{Existence of Energy Thresholds in Energy Bidding Games}
\label{sec:VI}
In this section, we show existence of an {\em energy threshold}, which is a necessary and sufficient energy required for \Pres to win from an initial configuration. 
Importantly, when \Pres loses with every initial energy, we call the energy threshold $\infty$. Formally,

\begin{definition}{\bf (Energy threshold).}\label{def:energy-thresh}
Consider an energy bidding game $\G = \zug{V, E, k, \weightf}$. The energy threshold is $\energy: V \times [k] \rightarrow \NatInf$ such that for configuration $\zug{v,B} \in \C$:
\begin{itemize}
\item If $\energy(v, B) = M \in \Nat$, then (1)~\Surv $M$-wins from $\zug{v, B}$ and \\(2)~\Killer $(M-1)$-wins from $\zug{v, B}$.
\item If $\energy(v, B) = \infty$, for every $M \in \Nat$, \Killer $M$-wins from $\zug{v, B}$.
\end{itemize}
\end{definition}	

\begin{remark}
[Energy thresholds and determinacy]
We point out that existence of $\energy$ is not trivial. Indeed, as seen Sec.~\ref{sec:prelim}, bidding games are succinctly represented concurrent games, and even simple concurrent games are not {\em determined}, namely neither player can guarantee winning.\footnote{Note that we restrict to {\em pure} strategies as opposed to {\em mixed} strategies that allow choosing a probability distribution over actions.} Existence of $\energy$ implies determinacy of energy bidding games. Indeed, assume that $\energy$ exists,  consider an initial configuration $c = \zug{v, B}$, and an initial energy level $M$. Then, if $M \geq \energy(v, B)$, \Pres has a winning strategy and if $M < \energy(v, B)$, \Cons has a winning strategy. The game is thus determined. 
\end{remark}

\subsection{Energy thresholds exist in finite-duration games}
Fix an energy game $\G = \zug{V, E, k, \weightf}$ for the remainder of this section. 
Let $n \in \Nat$. The {\em truncated game} $\G_n$ intuitively favors \Pres: she needs to keep the energy non-negative {\em only} in the first $n$ turns. We will show existence of energy thresholds in every truncated game, which is still not trivial since $\G_n$ is a concurrent game.  In the next section we extend to $\G$. 

Formally, for $M \in \Nat$, \Surv $M$-wins a path $\pi = v_0,v_1\ldots$ in $\G_n$ if for every $m \leq n$, we have $M + \sumofweights{\prefM} \geq 0$ and \Killer wins otherwise. We define threshold energies in $\G_n$, denoted $\energy_n: V \times [k] \rightarrow \Nat$, by plugging in the definition of $M$-wins in $\G_n$ in Def.~\ref{def:energy-thresh}. Recall that the minimal possible weight is $-W$, thus an initial energy of $nW$ suffices for \Pres to win in $\G_n$. It follows that energy thresholds in $\G_n$ are finite.
	
	\stam{
	We show existence of $\energy_n$ via an algorithm to compute it. 
	We recursively define a function \(\mu_{i}: V \times [k] \rightarrow Y\) and show that $\energy_n \equiv \mu_n$. 
	The base case is easy: since paths of length $0$ do not accumulate (negative) energy, \Surv always wins in $\G_0$, thus $\mu_0 \equiv \energy_0 \equiv 0$. 
	For the inductive step, consider a configuration $\zug{v, B}$. Intuitively, we define the initial energy $\mu_n(v, B)$ to suffice for winning even when \Pres reveals her bid first, \Cons responds adversarially, and the game proceeds to a configuration $\zug{v', B'}$, which requires an energy of \(\mu_{n-1}(v', B')\).
	}

	We show existence of $\energy_n$ via an algorithm to compute it. 
	We recursively define \(\mu_{i}: V \times [k] \rightarrow Y\) and show that $\energy_n \equiv \mu_n$. 
	For the base case, \Surv always wins in $\G_0$, thus $\mu_0 \equiv \energy_0 \equiv 0$. 
	For the inductive step, consider a configuration $\zug{v, B}$. Intuitively, we define the initial energy $\mu_n(v, B)$ to suffice for winning even when \Pres reveals her bid first, \Cons responds adversarially, and the game proceeds to a configuration $\zug{v', B'}$, which requires an energy of \(\mu_{n-1}(v', B')\).
	
We define $\mu_n$ formally. We first define \(\trump: [k] \times [k] \rightarrow [k]\) as the minimal bid that lets \Killer ``trump'' a \Surv bid and win the bidding. The definition depends on the tie-breaking status: if $B \in \Natstro$ and $b \in \Nat$ (\Pres has and does not use the advantage), then $\trump(B, b) = b$ and otherwise $\trump(B, b) = \succb{b}$. 
Consider a configuration $\zug{v, B}$ and a \Pres bid of $b \in [0, B]$. We consider two bidding outcomes: (1)~\Surv wins the bidding: the next configuration is \(\zug{v_{\text{win}}, B \ominus b}\), where \Pres chooses \(v_{\text{win}} \in \neighbor{v}\), and (2)~\Cons wins the bidding: the next configuration is $\zug{v_{\text{lose}}, B\oplus \trump{B,b})}$, where \Killer chooses  \(v_{\text{lose}} \in \neighbor{v}\). Note that, \Cons can win the bidding only if \(B \oplus \trump{B, b} \leq k^*\). 
The path accumulates energy \(\weightf(v, v_{\text{win}})\) or \(\weightf(v, v_{\text{lose}})\), respectively. The minimum required energy for \Pres to win in the respective cases is:
	\begin{align}
		\ewin{n}(v, B, b)&= \min_{v' \in N(v)} \max\set{\mu_{n-1}{(v', B \ominus b)} - \weightf{v, v'}, 0} \label{def: ewin}\\
		\elose{n}(v, B, b) &= \max_{v' \ \in N(v)} \max\set{\mu_{n-1}(v', B \oplus \trump{B, b}) - \weightf{v, v'}, 0} \label{def: elose}
	\end{align}
We stress that both $\elose{n}(v, B, b) \geq 0$ and $\ewin{n}(v, B, b) \geq 0$. Moreover, \(\elose{n}(v, B, b)\) is defined only when \(B \oplus \trump{B, b} \leq k^*\). Define $\enxt{n}(v, B, b) =\max\set{\ewin{n}(v, B, b), \elose{n}(v, B, b)}$ if $B \oplus \trump(B, b) \leq k^*$, and $\enxt{n}(v, B, b) =\ewin{}(v, B, b)$ otherwise. Then, 
	\begin{align}\label{def:O}
		\mu_{n}(v, B) = \min_{b \leq B} \enxt{n}(v, B, b)
	\end{align}


A \Pres strategy that maintains the energy above $\mu_n$ follows from the construction above, thus we obtain the following. 

\stam{
	\begin{align*}
		\enxt{n}(v, B, b) = 
		\begin{cases}
			\max\set{\ewin{n}(v, B, b), \elose{n}(v, B, b)} &\text{~if~} B \oplus \trump(B, b) \leq k^* \\
			\ewin{}(v, B, b) &\text{~otherwise}
		\end{cases} 
	\end{align*}
	\begin{align}\label{def:O}
		\mu_{n}(v, B) = \min_{b \leq B} \enxt{n}(v, B, b)
	\end{align}
}

%

	\begin{lemma}
	\label{lem:finite-Pres}
		For every $\zug{v, B}$ and $n \in \Nat$, \Surv $\mu_n(v,B)$-wins from $\zug{v, B}$ in $\G_n$. 
	\end{lemma}
	\stam{
	\begin{proof}
		The proof is by induction on $n$. The base case is covered above. For the inductive step, let $\zug{v, B}$ and an initial energy of $e = \mu_n(v,B)$. A winning \Surv strategy chooses $\zug{b, u}$ such that $b$ attains the minimum in Eq.~\ref{def:O} and $u$ attains the minimum in Eq.~\ref{def:req-energy}. 
		This ensures that no matter what the next configuration $\zug{v', B'}$ is, the updated energy $e' = e + \weightf(v, v')$ satisfies $e' \geq \mu_{n-1}(v', B')$, and the claim follows from the induction hypothesis. 
	\end{proof}
	}
	
The following lemma, whose proof can be found in App.~\ref{app:conswinsfinite}, shows that an energy of $\mu_n(v,B)$ is necessary for \Pres to win. 
The proof proceeds by showing that for every \Pres strategy, \Cons has a winning response. Existence of a winning strategy for \Cons follows from determinacy of reachability discrete-bidding games~\cite{DP10,AAH21}.
	
	\begin{restatable}{lemma}{conswinsfinite}\label{lem:conswinsfinite}
		For every $\zug{v, B}$ and $n \in \Nat$, \Killer $(\mu_n(v,B)-1)$-wins from $\zug{v, B}$ in $\G_n$. 
	\end{restatable}

Combining Lem.~\ref{lem:finite-Pres} and~\ref{lem:conswinsfinite}, we obtain the following.

	\begin{restatable}{theorem}{energycoincideswithmufinite}\label{thm: energycoincideswithmufinite}
		For every \(n \geq 0\), \(\energy_n\) exists. Moreover, we have \(\energy_n \equiv \mu_n\).
	\end{restatable}

\subsection{Extending to un-bounded energy games}
In this section, we show existence of energy thresholds in unbounded energy games. We define $\mu$ as a fixed point of $\mu_n$ and show that it coincides with $\energy$. The fixed-point computation needs to be done with care. A first attempt would be to simply consider the limit of the sequence $\set{\mu_n: n \geq 0}$. However, the sequence might not converge; indeed, even when $\energy(v, B) = \infty$, every $\mu_n(v, B)$ is finite. Instead, we define a sequence of trimmed functions \(\sqglymu: V \times [k] \rightarrow \NatInf\) (see details in App.~\ref{app:VI-alg}).
		\begin{align}\label{eq:trunc}
			\sqglymu{n}(v, B) &\coloneqq
			\begin{cases}
				\mu_n(v, B) &\text{~if~} \mu_n(v, B) \leq \upperboundforloop\\
				\infty &\text{~otherwise}
			\end{cases}
		\end{align}
		
In App.~\ref{app:monotonicity}, we established monotonicity. Since there are finitely many $\sqglymu{n}$ functions, monotonicity means that the sequence $\set{\sqglymu{n}: n \geq 0}$ reaches a fixed point.


		\begin{restatable}{lemma}{mumonotone}{\bf (Monotonicity)}
			\label{lemma: mumonotone}
			For all \(n \geq 0\), \(\sqglymu{n} \leq \sqglymu{n+1}\). 
			Moreover, for any vertex \(v\), and two budgets \(B_1, B_2 \in [k]\) with \(B_1 \geq B_2\), \(\sqglymu(v, B_1) \leq \sqglymu(v, B_2)\). 
		\end{restatable}

We denote the fixed point by $\mu$ and define the strategy that it gives rise to as follows. 

\begin{definition}
\label{def:sVI}
We define a strategy \(\sVI: V \times [k] \rightarrow [k] \times V\). At configuration $\zug{v, B}$, we define \(\zug{b', u'} = \sVI(v, B)\), where the bid $b'$ and $v'$ minimizes  Eq.~\ref{def:O} at the fixed point.
Namely, let $ \ewin{}(v, B, b) $ and $ \elose{}(v, B, b) $ is defined like Eq.~\ref{def: ewin} and Eq.~\ref{def: elose}, except $ \mu_{n-1} $ is replaced with $ \mu $, and $ \enxt{}(v, B, b) =  \max \set{\ewin(v, B, b), \elose{}(v, B, b)}$ if $ B \oplus \trump(B, b) < k+1 $, otherwise $ \enxt{}(v, B, b) = \ewin{}(v, B, b) $. Then, $ b' =  \arg\min_{b \leq B} \enxt{}(v, B, b)$, and $ u' = \arg\min_{u\in \neighbor{v}} \max\set{\mu(u, B \ominus b') - \weightf(v, u), 0} $.  

\stam{
			\begin{align*}
				b' &= \arg\min_{b} \max\{\ewin{}(v, B, b), \elose{}(v, B, b)\}\\
				v' &= \arg\min_{u \in N(v)} \max \set{\mu(u, B \ominus b') - \weightf(v, u), 0}\\
				\ewin{}(v, B, b) &= \min_{u \in \neighbor{v}} \max \set{\mu(u, B \ominus b) - \weightf(v, v'), 0}\\
				\elose{}(v, B, b) &= \max_{v' \in \neighbor{v}} \max \set{\mu(v', B \oplus \trump(B, b)) - \weightf(v, v'), 0}
			\end{align*}
}
\end{definition}

We show how \Pres wins by maintaining an {\em energy invariant}. This is reused in Sec.~\ref{sec:Pres-bud-agn}.


\begin{restatable}{lemma}{preswinsinfinite}\label{lemma: preswinsinfinite}
			If \(\mu(v, B) <\infty\), then $\sVI$ is a \Pres' strategy by which she \(\mu(v, B)\)-wins from configuration \(\zug{v, B}\)  in \(\G\).
\end{restatable}

\begin{proof}
			When \(\mu(v, B) < \infty\), we have $\mu(v, B) = \min_{b \leq B}\max\set{\ewin{}(v, B, b), \elose{}(v, B, b)}$. 

Consider an initial configuration $\zug{v_0, B_0}$ and an initial energy $e_0 \geq \mu(v, B)$. 
We describe a \Pres winning strategy inductively. Suppose that the game reaches  $\zug{v, B}$ with an accumulated energy of $e$. 
\Pres maintains that $e_0 + e \geq \mu(v, B)$. 
Since $\mu \geq 0$, it follows that a non-negative energy is preserved throughout the play, and \Pres wins. \Pres chooses \(\zug{b', u'}\) such that \(b'\) attains the minimum in the definition of $\mu$ 
and \(v'\) attains the minimum in the definition of \(\ewin\). It is not hard to verify that the invariant is maintained: no matter what the next configuration \(\zug{v'', B''}\) is, the accumulated energy \(e' = e + \weightf(v, v'')\) satisfies \(e' \geq \mu(v'', B'')\). 
		\end{proof}

We show that \Cons wins $\G$ by simulating a winning strategy in a finite game $\G_n$, for a large enough $n$. This idea is reused in Sec.~\ref{sec:Cons-bud-agn}. 
		
\begin{restatable}{lemma}{conswinsinfiniteenergy}\label{lemma: conswinsinfiniteenergy}
If \(\mu(v, B) < \infty\), then \Cons \((\mu(v, B) - 1)\)-wins from \(\zug{v, B}\) in \(\G\). 
If \(\mu(v, B) = \infty\), then for every \(e \in \Nat\), \Cons \(e\)-wins from \(\zug{v, B}\) in \(\G\). 
\end{restatable}

\begin{proof}[Proof sketch]
	We describe the proof idea and the details can be found in App.~\ref{appn: conswinsinfiniteenergy}. 
	Assume \(\mu(v, B) = \infty\), thus \(\mu_n(v, B) > \upperboundforloop\), for some \(n\).  
	Consider an initial energy \(e\). 
	We construct a \Cons strategy \(\tau\) as follows. 
	Let $\tau_n$ be a \((\mu_n(v, B) -1)\)-winning strategy in \(\G_n\). 
	Intuitively, $\tau$ ``simulates'' $\tau_n$ and follows its actions. Consider a play $\pi_n$ that is consistent with $\tau_n$. Note that $\pi_n$ coincides with a prefix of a of a play $\pi$  in $\G$ that is consistent with $\tau$. Since $\tau_n$ is winning, $\pi_n$ consumes at least $\upperboundforloop$ units of energy. Recall that the minimal weight in $\G$ is $-W$, thus the length of $\pi_n$ is at least $|V| \cdot k$, which implies that $\pi_n$ must contain a negative configuration cycle. Formally, there is an index $m$, such that $\pi_n^{\leq m} = \pi'_n \cdot \chi_n$, where $\chi_n$ is cycle of configurations with $\sumofweights{\chi_n} < 0$. We define $\tau$ to intuitively omit $\chi_n$ and restart the simulation of $\tau_n$ at $\pi'_n$. That is, the next action it chooses is $\tau_n(\pi_n')$. 
	By repeating this process, we obtain that a play $\pi$ that is consistent with $\tau$ consists of only negative cycles and at most \(n\) additional configuration, thus \(e\) is eventually gets consumed and $\tau$ is winning.
\end{proof}

Combining Lem.~\ref{lemma: preswinsinfinite} and Lem.~\ref{lemma: conswinsinfiniteenergy}, we obtain the following. 
		\begin{restatable}{theorem}{energyconicideswithmulimit}\label{them: energyconicideswithmulimit}
Consider an energy bidding game \(\G= \zug{V, E, k, \weightf}\). The energy threshold function \(\energy: V \times [k] \rightarrow \NatInf\) exists and satisfies 
$\mu \equiv \energy$. 
		\end{restatable}
	
	We observe the following about \(\energy\):
	
	\begin{corollary}\label{corr: energyinherits}
		\(\energy\) inherits the monotonicity of \(\sqglymu\) for budgets: for every vertex \(v\), and two budgets \(B_1 \geq B_2\), we have \(\energy(v, B_1) \leq \energy(v, B_2)\). 
	\end{corollary}

We close this section by an equivalence between energy and mean-payoff games, proof of which can be found in App.~\ref{app: MP-energy}.

\begin{restatable}{theorem}{mpenergy}
\label{thm:MP-energy}
Consider a game $\G = \zug{V, E, k, \weightf}$ and a configuration $\zug{v, B}$. If $\energy(v, B) < \infty$, \Max can guarantee non-negative payoff from $\zug{v,B}$. If $\energy(v, B) = \infty$, \Min can guarantee negative payoff from $\zug{v, B}$. 
\end{restatable}

\section{On Threshold Budgets}
Our goal is to develop succinct winning strategies. Towards this goal, in this section, we define threshold budgets, identify their mathematical structure, and deduce succinct {\em budget agnostic} strategies from this structure.

Recall that $\energy(v, B)$ is the necessary and sufficient initial energy for \Pres to win from a configuration $\zug{v, B}$. 
Further recall that $\energy(v, B)$ is (weakly) monotonically increasing as $B$ decreases (a lower initial budget, requires a higher initial energy). 
The {\em threshold budget} is the lowest budget $B$ such that $\energy(v, B)$ is finite.


\begin{definition}[Threshold budgets]\label{def: presthresholdbudget}
Define \(\preThresh: V \rightarrow \budgetset\) such that (1)~if $\energy(v, B) = \infty$ for all $B \in [k]$, then $\preThresh(v) = k+1$ and (2)~$\preThresh(v) = \min\set{B < k+1: \mu(v, B) < \infty}$ otherwise. 
\end{definition}

\begin{remark}[Thresholds in mean-payoff games]
Thm.~\ref{thm:MP-energy} implies that thresholds directly apply to mean-payoff games; $\Th(v)$ is a necessary and sufficient budget for \Max to guarantee a non-negative payoff. 
\end{remark}



In reachability continuous-bidding games~\cite{LLPU96}, the threshold of a vertex $v$ is the average of two of its neighbors, $v^+$ and $v^-$, which respectively denote the neighbor with the maximal and minimal threshold. Below, we describe a discrete version of this {\em average property}~\cite{DP10}.

\begin{definition}[Average property]\label{def: averageprop}
		Consider a graph $\zug{V, E}$ and $k \in \Nat$. A function \(T: V \rightarrow \budgetset\) satisfies the average property if 
		\begin{align*}
			T(v) = \floorvT{\sumTf{T}} + \eps 
		\end{align*}
		where \(\vplus{T} = \arg\max_{v' \in \neighbor{v}} T(v')\), \(\vminus{T} = \arg\min_{v' \in \neighbor{v}} T(v')\), and (1)~if $\sumTf{T}$  is even and $T(\vminus{T}) \in \Nat$, then $\eps = 0$, (2)~if $\sumTf{T}$  is odd and $T(\vminus{T}) \in \Natstr$, then $\eps = 1$, and (3)~otherwise $\eps = *$. 
		We often drop $ T $ from the notation of $ \vplus{T} $ and $ \vminus{T} $.
\stam{
		\begin{align*}
			\eps = 
			\begin{cases}
				0 &\text{~if~} \sumTf{T} \text{~is even and~} T(\vminus{T}) \in \Nat\\
				1 &\text{~if~} \sumTf{T} \text{~is odd and~} T(\vminus{T}) \in \Natstr\\
				* &\text{~otherwise}
			\end{cases}
		\end{align*}
		}
	\end{definition}

The main result in this section states that thresholds in energy games satisfy the average property. 
Our proof technique is very different from previous works; both for reachability~\cite{DP10} and parity games~\cite{AS25}, the proof that thresholds satisfy the average property is a byproduct of a value-iteration algorithm. Our value-iteration algorithm (Sec.~\ref{sec:VI}) focuses on the energy threshold and does not immediately imply the average property for the budget thresholds. Instead, in App.~\ref{app:preThreshAverage} we proceed as follows. For each $v \in V$, we define $f(v) \coloneqq \floorvT{\sumTf{\preThresh}} + \eps$ as in Def.~\ref{def: averageprop}. We show that $f \equiv \preThresh$ by showing that for every vertex $v$, we have $\energy(v, f(v)) < \infty$ and $\energy(v, \predb{f(v)}) = \infty$. The proof proceeds by a careful case-by-case analysis.		
		
		\begin{restatable}{theorem}{preThreshAverage}\label{thm:preThreshAverage}
			Consider an energy game \(\G = \zug{V, E, k, \weightf, \text{energy}}\). The threshold budget \(\preThresh: V \rightarrow \budgetset\) satisfies the average property. 
		\end{restatable}

\paragraph*{A budget-agnostic partial strategy}
We seek succinct winning strategies, which intuitively choose bids ignoring excess budget. 

\begin{definition}[Budget agnostic strategy]\label{def: budget-agnostic}
Define $ \relativebud: V \times [k] \rightarrow [k] $ that ``trims'' excess budget: for $\zug{v, B}$ with $B \geq \Th(v)$, define $\relativebud(v, B)$ to be whichever of $\preThresh(v)$ or $\succb{\preThresh(v)}$ agrees with $B$ on the tie-breaking advantage.
	A winning strategy $f$ is budget agnostic if for every $v \in V$ and $B \geq \Th(v)$, and every two histories $h_1,h_2$ that end in $\zug{v, B}$, we have $\zug{b, u_1} = f(h_1)$ and $\zug{b, u_2} = f(h_2)$.
\end{definition}

In fact, in Rem.~\ref{rem: trubudgetagnostic}, we will show existence of winning {\em positional} budget agnostic strategies.

We proceed as follows. A function $T$ that satisfies the average property gives rise to a {\em partial} budget-agnostic strategy $f_T$; namely it assigns to each configuration $\zug{v, B}$ a pair $\zug{b, S}$, where $b$ is a bid and $S \subseteq V$ is a set of {\em allowed vertices} to move to upon winning the bidding. 
Intuitively, a strategy that {\em agrees} with $f_T$ maintains a {\em budget invariant} (see Lem.~\ref{lem:Talgebra} below).
In subsequent sections, we will construct winning budget-agnostic strategies $\sagn$ and $\tagn$ for \Pres and \Cons, respectively, that agree with a partial strategy $f_T$, namely the strategies choose the same bid as $f_T$ and choose one of the allowed vertex, thus they maintain a budget invariant. We define the bids and allowed vertices below. 

\begin{definition}[Bid choice]\label{def: optbid}
Consider a function $T: V \rightarrow [k]$ that satisfies the average property. We define a bid \(\optbid{T}(v, B)\) in two steps. First, let 
\begin{align*}
b^T(v) = 
\begin{cases}
	T(v) \ominus T(\vminus) &\text{~if~} T(\vminus) \in \Nat\\
	T(v) \ominus \left(\absolut{T(\vminus)} +1 \right) &\text{~otherwise}
\end{cases}
\end{align*}
Second, we define $\optbid{T}(v, B)= b^T(v)$ when both $b^T(v)$ and $B$ belong to either $\Nat$ or $\Natstro$, and $\succb{b^T(v)}$ otherwise. 
\end{definition}
Intuitively, \Pres ``attempts'' to bid $b^T(v)$ at $\zug{v, B}$. If $b^T(v)$ requires the advantage and $B$ does not have it, \Pres bids $\absolut{b^T(v)} + 1 \in \Nat$, which does not require the advantage. 

Next, we define the allowed vertices as the neighbors of $v$ that minimize $T$.


		\begin{definition}[Allowed vertices]\label{def: allowedvertices}
For a function $T$ that satisfies the average property, the set of \emph{allowed vertices} at vertex $v$ are: if $T(\vminus) \in \Nat$, then $\Allowd{T}(v) = \{u \in \neighbor{v}: T(u) = T(\vminus)\}$, and otherwise $\Allowd{T}(v) = \{u \in \neighbor{v}: T(u) \leq \succb{T(\vminus)}\}$.
\stam{
			\begin{align*}
				\Allowd{T}(v) = 
				\begin{cases}
					\{u \in \neighbor{v}: T(u) = T(\vminus)\} &\text{~if~} T(\vminus) \in \Nat\\
					\{u \in \neighbor{v}: T(u) \leq \succb{T(\vminus)}\} &\text{~otherwise}
				\end{cases}
			\end{align*}
			}
		\end{definition}

Consider a configuration $\zug{v, B}$ with $B \geq T(v)$. The following lemma shows that choosing an action $\zug{b, u}$ with $b = \optbid{T}(v, B)$ and $u \in \Allowd{T}(v)$, maintains a budget invariant: no matter how the opponent acts, in the next configuration $\zug{v', B'}$, we have $B' \geq T(v')$. In particular, this implies that a $\optbid{T}(v, B)$ is a legal bid, i.e., $\optbid{T}(v, B) \leq B$. 


\begin{lemma}\label{lem:Talgebra}\cite{AS25}
Let \(T\) be a function that satisfies the average property, and \(v \in V\). Then, \(T(v) \ominus b^T(v) \geq T(\vminus)\) and \(T(v) \oplus \left(\succb{b^T(v)}\right) \geq T(\vplus)\).
\end{lemma}

\section{Constructing a Budget Agnostic Winning Strategy for \Pres}
\label{sec:Pres-bud-agn}
In this section we construct a budget-agnostic strategy $\sagn$ for \Pres. Recall that $\sVI$ is the strategy that is constructed by the value-iteration algorithm (see Def.~\ref{def:sVI}). In qualitative games~\cite{DP10,AS25}, the value-iteration algorithm outputs a budget-agnostic strategy. However, the following example shows that $\sVI$ is not budget agnostic.

\begin{example}
\label{ex:example2}
Consider the energy discrete-bidding game depicted in Fig.~\ref{fig:example2}. Set \(k = 5\). 
Suppose that the game starts from $\zug{v_1, 1}$ with initial energy $2$, and we describe a play that is consistent with $\sVI$ when \Cons responds optimally:
$\zug{v_1, 1} \xrightarrow{0,0} \zug{v_2,1} \xrightarrow{1, 1^*} \zug{v_1, 2^*}  \xrightarrow{0,0} \zug{v_2, 2^*} \xrightarrow{0^*,1} \zug{v_1, 3^*} \xrightarrow{0,0} \zug{v_2,3^*}\xrightarrow{3^*,2} \zug{t,0}$. 
Intuitively, observe that traversing the cycle $v_1,v_2,v_1$ causes both a decrease in energy and an increase to \Pres's budget. 
Note that at $v_2$, the bids are $0^*$ until the last visit in which $\sVI$ bids $3^*$. Since \Cons budget is $2$, this forces the game to $t$, where \Pres wins. 

The budget agnostic strategy $\sagn$ that we construct will always bid $0^*$ at $v_2$. It too is winning, but requires traversing the cycle twice more. Indeed, after two more traversals, we reach configuration $\zug{v_2, 5^*}$ in which \Cons's budget is $0$, and \Pres's bid of $0^*$ is winning. Finally, we note that while an initial energy of $2$ suffices for $\sVI$ to win from $\zug{v_1, 1}$, $\sagn$ requires an initial energy of $5$.
\stam{
	Consider the mean-payoff discrete-bidding game that is depicted in Fig.~\ref{fig:example2}, and a total budget \(k = 5\). 
	We claim that $ \preThresh(v_0) = 1 $, and $ \energy(u, 1) = 2 $. 
	Note that, if the initial energy is less than \(2\), then \Pres cannot afford to take the cycle $ c = v_0 \rightarrow v_1 \rightarrow v_2  \rightarrow v_0$ more than once. 
	On the other hand, \Cons with budget $ 4^* $ can, in fact, force the cycle at least twice, thus showing $ \energy(u, 1) \geq 2 $. 
	However, a \Pres strategy $ \sigma $ exists which ensures that from $ \zug{u, 1} $, if $ c $ is taken twice then when the token reaches again at $ v_2 $, she has a budget of $ 3 $.
	From $ \zug{v_2, 3} $, she plays $ \zug{3, t} $, and wins the game. 
	This strategy requires the initial energy $ 2 $ to keep the energy level always non-negative. 
	}
\end{example}
%
%

\begin{figure}
	\centering
	\begin{tikzpicture}
		\draw (4,0) node[rond5] (1) {$v_1$};
		\draw (8,0) node[rond5] (2) {$v_2$};
		\draw (12, 0) node[rect] (3) {$ t $};
		
		\draw (1) edge[->] node[sloped, below, pos= 0.5] {$ 2 $} (2);
		\draw (2) edge[->] node[sloped, above, pos= 0.5] {$ 2 $} (3);
		\draw (2) edge[->, bend right = 18.5] node[sloped, above, pos= 0.5] {$ -3 $} (1);
		
		\draw (3) edge[->, loop right] node[above] {$ 0 $} (3);
	\end{tikzpicture}
	\caption{A mean-payoff discrete bidding game where $ \sVI{\sigma} $ and $ \sagn{\sigma}$ sometimes act differently.}\label{fig:example2}
\end{figure}

In App.~\ref{app: acts-coincide}, we prove the following key lemma.
It identifies configurations in which the bids chosen by $\sVI$ coincide with the bids that are deduced from the average property (Def.~\ref{def: optbid}).

\begin{restatable}{lemma}{actscoincide}
\label{lem:acts-coincide}
Let \(\zug{v, B}\) with \(B \in \set{\preThresh(v), \succb{\preThresh(v)}}\). 
Then, $ \sVI(v,B) = \zug{\optbid{\preThresh}(v, B), u}$ for some $ u \in \Allowd{\preThresh}(v) $. 
\end{restatable}
\stam{
\begin{proof}[Proof sketch]
	To show that $ \sVI $ bids $ \optbid{\preThresh}(v, B) $ at these configurations, we show contradiction to all other possible bids with the properties of $ \sVI $, described in Def.~\ref{def:sVI}. 
	Similar is the argument to show $ \sVI $ must choose the next vertex from $ \Allowd{\preThresh}(v) $. (See proof in App.~\ref{app: acts-coincide}).
\end{proof}
}
We define $\sagn$ as follows and note that it is budget agnostic by construction. Recall that $\relativebud(v, B)$ is whichever of $\preThresh(v)$ or $\succb{\preThresh(v)}$ agrees with $B$ on the tie-breaking advantage. 

\begin{definition}\label{def:sagn}
For a configuration $ \zug{v, B} $ with $ B \geq \preThresh(v) $,	 define $ \sagn{\sigma}(v, B) = \sVI(v, \relativebud(v, B)) $. 
\end{definition}

Intuitively, $\sagn$ enjoys two features. First, since it follows $\sVI$, it maintains an energy invariant as in the proof of Lem.~\ref{lemma: preswinsinfinite}. Second, Lem.~\ref{lem:acts-coincide} means that it maintains a budget invariant as in Lem.~\ref{lem:Talgebra}. But Lem.~\ref{lem:acts-coincide} does not apply in all configurations. 
In the proof of the following theorem, we will show that both features are guaranteed to {\em eventually} hold.


\begin{restatable}{theorem}{sagniswinning}
\label{thm:sagn-is-winning}
Consider an energy game \(\G= \zug{V, E, k, \weightf}\) and a configuration \(\zug{v, B}\) with $ B \geq \preThresh(v) $. Then, there exists a finite energy \(M = (k+1) \cdot (\maxweight + \max_{v \in V} \energy(v, \preThresh(v)))\) such that $\sagn$ \(M\)-wins from $\zug{v, B}$.
\end{restatable}

\begin{proof}
Consider a play \(\pi = \zug{v_0, B_0} \zug{v_1, B_1}, \ldots \) consistent with \(\sagn\) from \(\zug{v, B}\). 
We define a function $\spare: \Nat \rightarrow [k]$ that intuitively assigns to each turn $i \in \Nat$, \Pres's {\em spare change}, which is roughly the difference between $B_i$ and the required threshold budget $\Th(v_i)$. Formally, $\spare{i} = \absolut{B_i \ominus \preThresh(v_i)}$. We first observe that $\spare{i}$ is monotonically increasing. Indeed, $\sagn$ bids at turn $i$ as if the budget is $B_i$, which suffices to ensure that $B_{i+1} \ominus \spare{i+1} \geq \Th(v_{i+1})$, thus the spare change is unused and can only increase. Clearly, $\spare$ is bounded by $k$, thus it eventually stabilizes; let \(N \geq 0\) and \(r \in \Nat\) such that \(\spare{m} = r\), for every \(m \geq N\). 

In App.~\ref{app:singletransitioninvariant}, we show that when $\spare$ is stable, the energy invariant is maintained. Formally, 


		\begin{restatable}{claim}{singletransitioninvariant}\label{clm: singletransitioninvariant}
			If \(\spare{m} = \spare{m+1}\) for some \(m \geq 0\), then $\energy(v_m, \relativebud{v_m, B_m}) + \weightf(v_m, v_{m+1}) \geq \energy(v_{m+1}, \relativebud{v_{m+1}, B_{m+1}})$.
			\stam{
			\begin{align*}
				\energy(v_m, \relativebud{v_m, B_m}) + \weightf(v_m, v_{m+1}) \geq \energy(v_{m+1}, \relativebud{v_{m+1}, B_{m+1}})
			\end{align*}
			}
		\end{restatable}

Suppose that $\spare$ stabilizes at turn $N$. 
Using Claim~\ref{clm: singletransitioninvariant}, in App.~\ref{app:thm:sagn-is-winning} we show: 
(1)  \Pres wins the suffix from turn $N$, if the energy is at least \(\energy(v_N, \relativebud(v_N, B_N))\), and 
(2) we bound the initial energy that $\sagn$ requires to guarantee that at turn $N$, the energy is at least \(\energy(v_N, \relativebud(v_N, B_N))\). 
For (2), intuitively, we partition $\pref{N}$ into ``patches'' such that $\spare$ is stable within each patch and changes between patches. We show that within each patch, the energy is preserved, and between patches, the energy increases by at most $W$. Since there are at most $k$ patches, $\sagn$ requires factor $k \cdot W$ more initial energy than $\sVI$. 
\end{proof}

	\section{Constructing a Budget Agnostic Winning Strategy for \Cons}
	\label{sec:Cons-bud-agn}
In this section, we show existence of a \Cons budget-agnostic winning strategy. There are two challenges w.r.t. the construction for \Pres. First, the proof that \Cons has a winning strategy in a finite energy game (Lem.~\ref{lem:conswinsfinite}) is existential and, unlike the case of \Pres, does not provide an explicit construction. Second,  while \Pres needs to maintain an energy invariant, \Cons needs to ``make progress'' and consume energy. 

Throughout this section, in order to avoid clutter, we use primed notation to refer to \Cons's perspective: 
we use $B'$ to refer to \Cons budget, thus configuration $\zug{v, B'}$ (from \Cons perspective) refers to $\zug{v, k^* \ominus B'}$ (from \Pres perspective),
we denote by $\Th'$ the threshold from \Cons perspective, formally, $\Th'(v) = (k+1)\ominus \Th(v)$,  and so on. 

Consider a function $T: V \rightarrow [k]$ that satisfies the average property. The {\em complement} of $T$, denoted by \(T'\), intuitively represents \Cons's budget when \Pres's budget is $\predb{T(v)}$. 

	\begin{restatable}{lemma}{complementaverage}\label{lemma: complementaverage}\cite{AS25}
	Consider a function \(T: V \rightarrow \budgetset\) that satisfies the average property. The complement of \(T\), denoted \(T' : V \rightarrow \budgetset\), is \(T'(v) = (k+1) \ominus T(v)\). Then, $T'$ satisfies the average property. 
	\end{restatable} 

Since $\Th$ satisfies the average property (Thm.~\ref{thm:preThreshAverage}), its complement $\Th'$ also satisfies it.

\paragraph*{On winning strategies in finite-duration energy bidding games}
Intuitively, our budget-agnostic strategy has two features. First, its bids match the bids derived from $\Th'$ (Def.~\ref{def: optbid}), thus it maintains a {\em budget invariant} as in Lem.~\ref{lem:Talgebra}. Second, its actions follow some winning strategy $\sVIn{\tau}{n}$ in a truncated game $\G_n$, thus it maintains an energy invariant that implies energy consumption. In this section, we establish properties of $\sVIn{\tau}{n}$ that enable these features. 

We denote by $\energy'_n$ the energy threshold in $\G_n$ from \Cons perspective, namely from $\zug{v, B'}$, \Cons can consume $\energy'_n(v, B')$ units of energy within $n$ turns, but cannot consume $\energy'_n(v, B')+1$ units. Formally, $ \energy'_{n}(v, B') = \energy_{n}(v, B) - 1$ where $ B = k^* \ominus B' $. 

Recall that $\Th'$ is defined such that $\energy'(v, \Th'(v)) = \infty$ but $\energy'_n(v, \Th'(v))$ is finite, for every $n \in \Nat$. Intuitively, the following lemma (see the proof in App.~\ref{app:sufficientnumber}) states that for every energy $e$, there is a truncated game $G_N$ in which \Cons $e$-wins from $\zug{v, \Th'(v)}$.

\begin{restatable}{lemma}{sufficientnumber}\label{lemma: sufficientnumber}
	For every $ e < \infty$ and $ v \in V$, there exists $ N $ such that $\energy'_{N}(v, \conThresh(v)) \leq e$.
\end{restatable}
\stam{
\begin{proof}
	Since $ \energy(v, \conThresh(v)) = \infty $, by Lem.~\ref{lemma: conswinsinfiniteenergy}, \Cons has a strategy $ \tau $ that is $ M $-winning in $ \G $, for every $ M \in \Nat $. 
	Let $ n_M $ be the maximal number of turns that \Pres can avoid against $ \tau $ (see App.~\ref{app:sufficientnumber} for the proof that such $ n_M $ exists). 
	Thus, \Cons $ M $-wins in $ G_{n_M} $. 
	The lemma follows from choosing $ N = n_M $, for $ M = e $. 
\end{proof}
}

	Let $n \in \Nat$. We denote by $\sVIn{\tau}{n}$, a \Cons strategy that $\energy_{n}'(v, \conThresh(v))$-wins from $\zug{v, \conThresh(v)}$ in $\G$. 
	In Lem.~\ref{lem:conswinsfinite}, we show that $\sVIn{\tau}{n}$ exists, moreover, it operates on an arena in which each vertex is \(\tup{v, e, m}\), where $v$ is a vertex, the current energy is $e$, and $m \leq n$ is a counter that marks the remaining turns. In App.~\ref{app: properties}, we establish the following properties. 
	
	\begin{restatable}{lemma}{properties}\label{lemma: properties}
		Consider $x = \tup{v, e, m}$, a \Cons budget $ B' \geq \conThresh(v)$ such that, $ e \leq \energy_m'(v, B') $, and let $ \zug{b, u} = \sVIn{\tau}{n}(x, B') $. The following hold:
		\begin{itemize}[noitemsep,topsep=0pt,parsep=0pt,partopsep=0pt,leftmargin=*]
			\item There is $\sVIn{\tau}{n}$ such that $ \sVIn{\tau}{n}(x, B') = \sVIn{\tau}{n}(\tup{v, \energy_m'(v, B'), m}, B') $
			
			\item $ \sVIn{\tau}{n} $ maintains an energy invariant: (i)~\Cons wins the bidding: $ e + \weightf(v, u) \leq \energy_{m-1}'(u, B' \ominus b) $, and 
			(ii)~\Cons loses the bidding: assuming $ B' \oplus (\succb{b}) < k+1$, then $ e + \weightf(v, v') \leq \energy_{m-1}'(v', B' \oplus (\succb{b})) $, for every $ v' \in \neighbor{v} $.
			\item If $ B' \in \set{\conThresh(v), \succb{\conThresh(v)}}$ and $ \energy_m'(v, B') > |V|k\maxweight $, then $b = \optbid{\conThresh}(v, B')$ and $ u \in \Allowd{\thresh'}(v) $. 
		\end{itemize}
	\end{restatable}

Intuitively, 	in the first item, note that $\zug{v, \energy'_{m-1}(v, B'), m}$ is a worse configuration than $x$ for \Cons, thus \Cons wins in the former by following a winning strategy of the latter.
The third item is a key property that is analogous to Lem.~\ref{lem:acts-coincide} for \Pres. It identifies inputs such that $\sVIn{\tau}{n}$ matches the bids derived from $\Th'$.

	

%
%

\paragraph*{A budget-agnostic \Cons winning strategy}	
We proceed as follows. We start by defining a budget agnostic strategy $\tagn'$, which drops energy from $ 2\upperboundforloop+1 $ to $ \upperboundforloop +1$.
The definition of $\tagn'$ is based on $\sVIn{\tau}{n}$, and since it operates at high energy levels, we obtain the properties in Lem.~\ref{lemma: properties}. Then, we define $\tagn$ to simulate $\tagn'$ and repeatedly omit negative configuration cycles, as in Lem.~\ref{lemma: conswinsinfiniteenergy}.
	
We define $ \relativebud': V \times [k] \rightarrow [k] $: for $ \zug{v, B'} $ with $ B' \geq \conThresh(v) $, define $ \relativebud'(v, B') $ to be $ \conThresh(v) $ or $ \succb{\conThresh(v)} $ that agrees with $ B' $ on the tie-breaking advantage. 

\begin{definition}\label{def: consbudgetagnostic}
	Let $ \zug{v, B'} $ with $B' \geq \conThresh(v)$, energy $ e > \upperboundforloop $, and  $ P(e) \in \Nat$ the minimal integer s.t. $ \energy'_{P(e)}(v, \relativebud'(v, B)) \geq e $. 
We define $ \tagn'(v, e, B') = \sVIn{\tau}{P(e)}(\tup{v, e, P(e)}, \relativebud'(v, B'))$. 
\end{definition}
	
Note that the third item in Lem.~\ref{lemma: properties} implies that $\tagn'$ is budget agnostic by construction. Indeed, the bid chosen by $\sVIn{\tau}{P(e)}$ is $\optbid{\conThresh}(v, B')$. 

	\begin{restatable}{lemma}{tauwinningaux}\label{lemma: tau-winning-aux}
		From a configuration $ \zug{v, B'} $ with $ B \geq \conThresh(v) $, and initial energy level $ e_0 \geq 2\upperboundforloop $, $ \tagn' $ ensures that the energy drops to $ \upperboundforloop$.
	\end{restatable}
	
	\begin{proof}[Proof Sketch]
Consider a play $ \pi'$ consistent with $\tagn'$. 
		Similar to the proof of Thm.~\ref{thm:sagn-is-winning}, we define the {\em spare change} of \Cons $\triangledown_{\pi'}$ and show that it eventually stabilizes due to the budget invariant that $\tagn'$ maintains. When $\triangledown_{\pi'}$ is stable, we show that $\tagn'$ maintains an energy invariant, which implies that energy is consumed. See details in App.~\ref{app:  tau-winning-aux}. 
	\end{proof}
	
	We  proceed to prove the main result in this section.

	\begin{restatable}{theorem}{tauagniswinning}\label{thm: tau-agn-winning}
		Consider an energy game $ G = \zug{V, E, k, \weightf} $. There exists a \Cons budget-agnostic strategy $ \tagn $ that $M$-wins from every configuration $ \zug{v, B'} $ with $ B' \geq \conThresh(v)$, for every energy $ M \in \Nat$. 
	\end{restatable}

	\begin{proof}
We proceed similar to the proof of Lem.~\ref{lemma: conswinsinfiniteenergy}. We observe that any play $ \pi' $ consistent with $ \tagn' $ from $ \zug{v, B'} $ and initial energy $ e_0 = 2\upperboundforloop $ must contain a negative configuration cycle before reaching energy $ \upperboundforloop$. 
		We define $ \tagn$ to simulate $ \tagn' $ from $ \tup{v, e_0, B'} $ until such a negative cycle is closed, omit it, and restart the simulation. Since $ \tagn' $ is budget agnostic, so is $\tagn$. By repeating this process, we obtain that a play $ \hat{\pi} $ that is consistent with $ \tagn $, consists of only negative cycles and a finitely many additional configurations (see App.~\ref{app:maximumnumber} for an upper bound). Thus, $ \tagn$ is $ M $-winning from $ \zug{v, B'} $, for every energy $ M $. 
	\end{proof}


	\section{Finding Threshold Budgets is in \NP~and \coNP~}\label{sec: NPandcoNP}	
	We formalize the problem of finding threshold budgets as a decision problem: 
	\begin{problem}[Finding Threshold Budgets]\label{prob: decisionprobthresh}
		Given an energy bidding game \(\G = \zug{V, E, k, \weightf}\), a vertex \(v\), and a budget \(\ell \in [k]\), decide whether \(\preThresh(v) \geq \ell\).
	\end{problem}

We show that Prob.~\ref{prob: decisionprobthresh} is in \NP~and \coNP. Our approach is similar to~\cite{AS25}. The core of the algorithm is to decide, given a function \(T: V \rightarrow \budgetset\) that satisfies the average property whether $T \equiv \Th$. Note that $T$ can be guessed since its size is \(\calO(|V|\cdot\log(k))\).

\noindent{\bf From bidding games to turn-based games.}
We show how to decide $T \geq \Th$ and showing that $T \leq \Th$ is dual.
Given an energy bidding game $\G$ and $T$ that satisfies the average property, we construct and solve a turn-based energy game $G_{T, \G}$. 
We describe the idea and the details can be found in App.~\ref{app:turn-based-construction}. Intuitively, a \Pres winning strategy in $G_{T,\G}$ corresponds to a winning budget-agnostic strategy in $\G$. 
For each vertex $v \in V$, there are two \Pres vertices $\zug{v, B}$ for $B \in \set{T(v), \succb{T(v)}}$, which are configurations in $\G$, and a third copy $\zug{v, \top}$, which is a winning sink for \Pres. 
Vertex $\zug{v, B}$ in $G_{T,\G}$ simulates configuration $\zug{v, \tilde{B}}$ in $\G$ with $B = \relativebud(v, \tilde{B})$. \Pres can choose any $v' \in \Allowd{T}(v)$. This corresponds to choosing action $\zug{\optbid{T}(v, B), v'}$ in $\G$. Then, \Cons responds by either: (1)~letting \Pres win the bidding and proceeding to $\zug{v', B \ominus \optbid{T}(v, B)}$ or (2)~win the bidding and choosing the successor vertex $u$, then the next vertex is $\zug{u, \tilde{B}}$, where $\tilde{B}= B \oplus (\succb{\optbid{T}(v, B)})$ if $\tilde{B} \in \set{T(u), \succb{T(u)}}$, otherwise the budget is trimmed to $\top$. 
The weights in $G_{T,\G}$ are derived from $\G$ such that a play in $G_{T, \G}$ that does not end in a sink, corresponds to a play in $\G$ that traverses the same sequence of weights. 

The following lemma shows soundness of the procedure. $\energy$ in turn-based games is similar to Def.~\ref{def:energy-thresh}; a necessary and sufficient initial energy for \Pres to win, see also~\cite{BCD+11}. 
	
	\begin{restatable}{lemma}{soundness}\label{lemma: soundness}
If $\energy(\zug{v, B}) < \infty$ in \(G_{T, \G}\), for every $v$ with $T(v) < k+1$, then \(T \geq \preThresh\). 
	\end{restatable}
	\begin{proof}[Proof sketch]
We construct a \Pres budget-agnostic strategy $\sigma$ in $\G$ that simulates the operation of a winning strategy $\sigma'$ in $G_{T,\G}$ as follows. See details in App.~\ref{app:soundness}. 
Suppose that $\G$ is in $\zug{v, \tilde{B}}$ with $\tilde{B} \geq T(v)$, then the simulation of $G_{T,\G}$ is in $\zug{v, \relativebud(v, \tilde{B})}$. We define $\zug{\optbid{T}(v,B), v'} = \sigma(v, \tilde{B})$ such that $\sigma$ chooses $v'$ in $G_{T,\G}$. 
\Pres simulates in $G_{T, \G}$, \Cons's response in $\G$, which we assume wlog, is either $\zug{0, u}$ or $\zug{\succb{\optbid{T}(v,B)}, u}$, for some $u \in N(v)$. If the next vertex in $G_{T,\G}$ is not a sink, we repeat. A sink is reached only when \Cons wins the bidding and \Pres's updated budget is $\tilde{B}' = B \oplus (\succb{\optbid{T}(v,B)})$ with $\tilde{B} > \succb{T(u)}$. We restart the simulation of $G_{T, \G}$ in $\zug{u, \relativebud(u, B')}$. Note that \Pres's spare change has increased, thus a play can end in a sink only finitely often. 
	\end{proof}

We proceed to prove completeness. 
	
	\begin{restatable}{lemma}{completeness}\label{lemma: completeness}
		If \(T \equiv \preThresh\), then \(\energy(\zug{v,B}) < \infty\) in \(G_{T, \G}\), for every $v$
		 with $T(v) < k+1$. 
	\end{restatable}
	\begin{proof}[Proof sketch]
Suppose towards contradiction that \(T \equiv \preThresh\) and for $\zug{v, B}$, we have $\energy(v,B) = \infty$ in $G_{T,\G}$ but $\energy(v,B) < \infty$ in $\G$. Let $\sagn$ be a $M$-winning from $\zug{v,B}$ in $\G$ (see Thm.~\ref{thm:sagn-is-winning}) and let $\tau'$ be a \Cons $M$-winning strategy in $G_{T,\G}$. We simulate $\sagn$ against $\tau'$ in both games. Crucially, the simulation in $G_{T,\G}$ is possible only since the actions that $\sagn$ chooses imply that the configuration updates in both games are the same. We obtain two plays, one in $\G$ and the other in $G_{T,\G}$ with the same sequence of weights, which is a contradiction since \Cons $M$-wins in $G_{T,\G}$ while \Pres $M$-wins in $\G$. See details in App.~\ref{app: completeness}.
	\end{proof}
	
	Finally, we verify \(T' \geq \conThresh\). 
	We proceed as before to construct \(G_{T', \G}\), except that \Pres responds to \Cons and sink vertices are winning for \Cons, i.e, $\zug{v,\top}$ has a \((-1)\)-valued self-loop. Dually it follows:
		\begin{restatable}{lemma}{conscorrectness}\label{lemma: conscorrectness}
		If \(\energy(\zug{v, B}) = \infty\) in $G_{T', \G}$, for  every $v$ with $\conThresh(v) < k+1$, then \(T' \geq \conThresh\). If \(T' \equiv \conThresh\), then \(\energy(\zug{v, B}) = \infty\) in $G_{T,\G}$, for every $v$ with $T'(v) < k+1$. 
	\end{restatable}
	
Since solving mean-payoff turn-based games is in NP and coNP~\cite{ZP96}, we obtain:
	
		\begin{restatable}{theorem}{complexity}\label{thm: complexity}
Finding threshold budgets in energy bidding game is in \NP~and \coNP.	
	\end{restatable}

		\begin{remark}\label{rem: trubudgetagnostic}
			Since there exist optimal strategies in turn-based energy games that are memoryless, the strategies in $\G$ constructed in Lem.~\ref{lemma: soundness} and~Lem.~\ref{lemma: conscorrectness} strengthen Thm.~\ref{thm:sagn-is-winning} and Thm.~\ref{thm: tau-agn-winning}: 
there exists a winning {\em positional} budget agnostic strategy $ \sagn{\hat{\sigma}}: V \times [k] \rightarrow [k] \times V$ such that $\sagn{\hat{\sigma}}(v, B) = \sagn{\hat{\sigma}}(v, \relativebud(v, B))$, for every $B \geq \Th(v)$, and similarly for \Cons. 
	\end{remark}

\section{Discussion}
We study, for the first time, a combination of discrete-bidding with mean-payoff and energy objectives. We define threshold budgets, establish existence, and construct concise budget-agnostic winning strategies, which serve as the basis for showing that finding thresholds is in NP and coNP, even when the budgets in the game are represented in binary. 

We believe that our work opens the door to extensions and generalizations, which are technically challenging to study in combination with continuous bidding. One example is bidding games in which the players are partially-informed of their opponent's budgets~\cite{AJZ23}, which is common in practice, but results under continuous bidding are limited due to the demanding technicalities. Other examples that have not yet been consider for mean-payoff objectives include bidding games with charging~\cite{AGHM24}, stochastic transitions~\cite{AK+25}, and non-zero sum games, which require refined solution concepts~\cite{AMS24}. 

Finally, we point out that our result positions mean-payoff discrete-bidding games in a peculiar state of affairs: on the one hand, solving turn-based mean-payoff games (in NP and coNP but not known to be in P) easily reduces to solving mean-payoff discrete-bidding games with total budget \(0\), and on the other hand, the result applies to a seemingly exponentially harder problem of finding thresholds in bidding games with budgets given in binary.

	\bibliography{ga-2}	
	\bibliographystyle{plainurl}

	\newpage
	\appendix

\section{Proof of Lem.~\ref{lem:conswinsfinite}}
\label{app:conswinsfinite}
We describe $\G_n$ as a {\em reachability} bidding game $\R_n = \zug{V_n, E_n, k, T_n}$, where the vertices are $V_n = V \times \set{-1, \ldots, |V| \cdot W} \times \set{0,\ldots, n}$, for $v \in V$, $e \geq 0$, and $m \geq 0$, there is an edge $E_n(\zug{v, e, m}, \zug{v', e', m-1})$ if $m \geq 0$, $v' \in N(v)$, and $e' = \max\set{e+w(v, v'), -1}$, and vertices of the form $\zug{v, e, 0}$ are sinks, the total budget is $k$, and $T_n = \set{\zug{v, -1, m}: v \in V \text{ and } m > 0}$. \Cons wins a play in $\R_n$ iff it visits $T_n$. 

The following claim relates $\R_n$ with $\G_n$. It can be proven using a simple induction.

\begin{restatable}{claim}{equivalent}\label{claim: equivalent}
	A player wins in $\R_n$ from configuration $\zug{\tup{v, e, m}, B}$ iff they $e$-win from $\zug{v, B}$ in $\G_m$.
\end{restatable}

\begin{restatable}{lemma}{conshasresponse}\label{lemma: conshasresponse}
	 Let $v \in V$ and $B \in [k]$. We claim that \Pres does not have a winning strategy from configuration \(c_0 = \zug{\tup{v, \mu_{n}(v, B) - 1, n}, B}\). That is, for every \Pres strategy $\tau$, \Cons has a response $\sigma$, such that $\play(c_0, \sigma, \tau)$ visits $T$. 
\end{restatable}

\begin{proof}
We first note that {\em determinacy} of reachability discrete-bidding games~\cite{DP10,AAH21} means that if \Pres does not win, then \Cons wins; formally, \Cons has a strategy $\tau$ such that for every strategy $\sigma$ of \Pres, $\play(c_0, \sigma, \tau)$ visits $T$.

		The proof of the claim is by induction on \(n\). 
		The base case is trivial: since \(\mu_0(v, B) = 0\), we have \(\zug{\tup{v, -1, 0}, B}\) which is winning for \Cons. 
For the induction step, we assume that \Cons $(\mu_{n-1}(v', B')-1)$-wins in $\G_{n-1}$ from configuration $\zug{v', B'}$, for $v' \in V$ and $B' \in [k]$, and we prove that \Cons $(\mu_n(v, B)-1)$-wins in $G_n$ from $\zug{v, B}$, for $v \in V$ and $B \in [k]$. 
Assume without loss of generality that \(\mu_n(v, B) > 0\), otherwise we apply the argument in the base case. 
		We will show that for every \Pres choice of \(\zug{b, u}\) at \(\zug{v, B}\) in \(\G\), \Cons has a response \(\zug{b', u'}\) which ensures that the next configuration \(\zug{v', B'}\) satisfies:
		\begin{align}\label{ineq: inductiveconsumer}
			\buildmun{n}(v, B) - 1 + \weightf{v, v'} \leq \buildmun{n-1}{v', B'} -1
		\end{align} 
By the induction hypothesis, \Cons $(\buildmun{n-1}{v', B'} -1)$-wins from $\zug{v',B'}$, thus we will obtain that \Pres does not \((\mu_{n}(v, B) -1)\)-win from \(\zug{v, B}\), which using determinacy implies that \Cons \((\mu_{n}(v, B) -1)\)-wins from \(\zug{v, B}\), as required. 
		
		\smallskip
		
		We proceed to show \eqref{ineq: inductiveconsumer}.
		Suppose \Pres chooses \(\zug{b, u}\) at \(\zug{v, B}\). Recall the energy \Pres requires if she wins and loses the bidding with $b$ is respectively
		\begin{align*}
			\ewin{n}(v, B, b) &= \max\set{\mu_{n-1}(u, B \ominus b) - \weightf(v, u), 0}, \text{~and}\\
			\elose{n}(v, B, b) &= \max\set{\mu_{n-1}(u', B \oplus b') - \weightf(v, u'), 0}
		\end{align*}
		We define \Cons's response \(\zug{b', u'}\) by
		\begin{align*}
			b' \coloneqq 
			\begin{cases}
				0 &\text{~if~} \ewin{n}(v, B, b) \geq \elose{n}(v, B, b)\\
				\triumph{B, b} &\text{~otherwise}
			\end{cases}
		\end{align*}
That is, \Cons lets \Pres win if $\ewin{n}(v, B, b) \geq \elose{n}(v, B, b)$, and otherwise he chooses the smallest possible winning bid. 
If $b'=0$, \Cons proceeds to \(u' = u\). If $b'>0$, we define \(u' \in \arg\max_{v' \in N(v)} \max\set{\mu_{n-1}(v', B + \trump{B, b}) - \weightf(v, v'), 0}\) 
		
		We analyze two exhaustive cases:
		\begin{itemize}
			\item \textbf{Case I:} \(\ewin{n}(v, B, b) \geq \elose{n}(v, B, b)\).
			
			First note that, \(\ewin{n}(v, B, b) > 0\), otherwise we would have \(\mu_{n}(v, B) = 0\). 
			In this case, either \Pres outright wins the bid at \(\zug{v, B}\) (when \(b > 0\) or \(B \in \Nat\)), or \Cons wins the bid (when  \(b = 0\) and \(B \in \Natstro\)). 
			In the latter scenario, \Cons chooses the same vertex as \Pres after winning. 
			Therefore, the subsequent configuration is \(\zug{u, B \ominus b}\). 
			Now, 
			\begin{align*}
				\mu_n(v, B)  - 1 + \weightf{v, u}
				&= \min_{b \leq B} \max\set{\ewin{n}(v, B, b), \elose{n}(v, B, b)} -1 + \weightf{v, u}\\
				&\leq \max\set{\ewin{n}(v, B, b), \elose{n}(v, B, b)} - 1 + \weightf{v, u}\\
				&= \ewin{n}(v, B, b) - 1 + \weightf{v, u}\\
				&= \mu_{n-1}(u, B \ominus b) - \weightf{v, u} - 1 + \weightf{v, u}\\
				&= \mu_{n-1}(u, B \ominus b) - 1
			\end{align*}
			
			\item \textbf{Case II:} \(\elose{n}(v, B, b) > \ewin{n}(v, B, b)\)
			
			Like the previous case, here as well, we have \(\elose{n}(v, B, b) > 0\). 
			\Cons wins the bid at \(\zug{v, B}\) since he bids \(b' = \triumph{B, b}\).
			Therefore, the subsequent configuration becomes \(\zug{u', B \oplus b'}\). 
			We proceed as follows:
			\begin{align*}
				 \mu_n(v, B) - 1 + \weightf{v, u'}
				&= \min_{b \leq B} \max\set{\ewin{n}(v, B, b), \elose{n}(v, B, b)} -1 + \weightf{v, u}\\
				&\leq \max\set{\ewin{n}(v, B, b), \elose{n}(v, B, b)} - 1 + \weightf{v, u}\\
				&= \elose{n}(v, B, b) - 1 + \weightf{v, u}\\
				&= \mu_{n-1}(u', B \oplus b') - \weightf{v, u'} - 1 + \weightf{v, u'}\\
				&= \mu_{n-1}(u', B \oplus b') - 1
			\end{align*}
		\end{itemize}
		
		This concludes the proof of Eq.~\eqref{ineq: inductiveconsumer}, consequently the proof of lemma as well.
\end{proof}

\section{The fixed-point algorithm to find $\sqglymu$}
\label{app:VI-alg}
		
		We develop an algorithm using this recursive definition of \(\sqglymu\), in which we initiate the function \(\sqglymu\) with \(\sqglymu{0}\), and until it reaches a fixed point, we upgrade the function from \(\sqglymu{i}\) to \(\sqglymu{i+1}\). (See the pseudocode in Algorithm \ref{algo: valiteration})

\algdef{SE}[DOWHILE]{Do}{doWhile}{\algorithmicdo}[1]{\algorithmicwhile\ #1}%


		\begin{algorithm}
			\caption{A fixed-point algorithm for computing energy thresholds in un-bounded energy games}\label{algo: valiteration}
			\begin{algorithmic}[1]
				\Require An Energy Game \(\G = \zug{V, E, \weightf, k}\)
				\Ensure $\mu^*: V \times [k] \rightarrow \NatInf$
				\GAn{Change Require -> Input, Ensure -> Output}
				\State \(i \gets 0\)
				\State Init $\sqglymu{0} \equiv 0$
				\Do
				\State \(i \gets i+1\)
				\For{\(\zug{v, B} \in V \times [k]\)}
				\For{\(b \leq B\)}
				\State \(\ewin{}(v, B, b) = \min_{v' \in N(v)}\max \set{\sqglymu{i-1}(v', B \ominus b) - \weightf(v, v'), 0}\)
				\If{\(\trump(B, b) < k+1\)}
					\State \(\elose{}(v, B, b) = \max_{v' \in N(v)}\max \set{\sqglymu{i-1}(v', B \oplus \trump(B, b)) - \weightf(v, v'), 0}\)
					\State $ \enxt{}(v, B, b) = \max\set{\ewin{}(v, B, b), \elose{}(v, B, b)} $
				\Else
					\State $ \enxt{}(v, B, b) = \ewin{}(v, B, b) $
				\EndIf
				\EndFor
				\State \(\sqglymu{i}(v, B) = \min_{b \leq B}\enxt{}(v, B, b)\)
				\If{\(\sqglymu{i}(v, B) > \upperboundforloop + 1\)}
				\State \(\sqglymu{i}(v, B) = \infty\)
				\EndIf
				\EndFor
				\doWhile{$\sqglymu{i} \not \equiv \sqglymu{i-1}$}
				\State \Return $\mu^* = \sqglymu{i}$
			\end{algorithmic}
		\end{algorithm}

\section{Proof of Lemma~\ref{lemma: mumonotone}}
\label{app:monotonicity}
\mumonotone*
	
\begin{proof}
	We proceed by induction on \(n\). 
	The base case is trivial since by definition \(\sqglymu{0}(v, B)= 0\), for all \(\zug{v, B}\). Thus, \(\sqglymu{1}(v, B) \geq \sqglymu{0}(v, B)\), for all configuration \(\zug{v, B}\). 
	
	For the induction hypothesis, we assume \(\sqglymu{n-1} \leq \sqglymu{n}\) and establish the statement for \(n\). 
	Note that, the induction hypothesis necessarily implies \(\mu_{n-1} \leq \mu_{n}\) as well. 
	
	Suppose towards contradiction that there exists a configuration \(\zug{v, B}\) for which \(\sqglymu{n}(v, B) > \sqglymu{n+1}(v, B)\). 
	There are two possibilities : (1)~both are finite, or (2)~\(\sqglymu{n}(v, B) = \infty\) but \(\sqglymu{n+1}(v, B) < \infty\). 
	We will show that both cases imply that \(\mu_n(v, B) > \mu_{n+1}(v, B)\). 
	For (1), the assertion is immediate since \(\sqglymu{n}(v, B)\) and \(\mu_n(v, B)\) coincide when \(\mu_n(v, B)\) is finite. 
	For (2), \(\sqglymu{n}(v, B)=\infty\) means \(\mu_n(v, B) > \upperboundforloop +1 \), while  \(\sqglymu{n+1}(v, B) < \infty\) means \(\mu_n(v, B) \leq \upperboundforloop + 1\). 
	
In the following, we reach a contradiction by showing that there exists a \(\zug{v', B'}\) for which \(\mu_n(v', B') > \mu_{n-1}(v', B')\). 

Denote 
	\begin{align*}
		b_n &\coloneqq \arg\min_{b} \max\set{\ewin{n}(v, B, b), \elose{n}(v, B, b)}\\
		b_{n+1} &\coloneqq \arg\min_{b} \max\set{\ewin{n+1}(v, B, b), \elose{n+1}(v, B, b)}
	\end{align*}
Thus, we recall \(\mu_n(v, B) = \max\set{\ewin{n}(v, B, b_n), \elose{n}(v, B, b_n)}\), and \(\mu_{n+1}(v, B) = \linebreak \max\set{\ewin{n+1}(v, B, b_{n+1}), \elose{n+1}(v, B, b_{n+1})}\). 
It follows from here that
	\begin{align}\label{ineq: munwithmunminusone}
		\max\set{\ewin{n}(v, B, b_{n+1}), \elose{n}(v, B, b_{n+1})} &\geq \max\set{\ewin{n}(v, B, b_n), \elose{n}(v, B, b_n)} \nonumber\\
		&> \max\set{\ewin{n+1}(v, B, b_{n+1}), \elose{n+1}(v, B, b_{n+1})}
	\end{align}
	
	The first inequality is coming from the definition of \(b_n\), which is the bid where the minimum of \(\max\set{\ewin{n}(v, B, b), \elose{n}(v, B, b)}\) is attained.  The second inequality follows from our assumption \(\mu_n(v, B) > \mu_{n+1}(v, B)\). 
	
	Depending on where the maximum of the set \(\set{\ewin{n}(v, B, b_{n+1}), \elose{n}(v, B, b_{n+1})}\) is attained, we analyse two exhaustive cases in the following:
	\begin{itemize}
		\item \textbf{Case I:} \(\ewin{n}(v, B, b_{n+1}) \geq \elose{n}(v, B, b_{n+1})\)
		
		\smallskip
		
		We assume that \(v' \in \neighbor{v}\) is a vertex where \(\ewin{n+1}(v, B, b_{n+1})\) attains its minimum. 
		In other words, 
		\begin{align}\label{eq: monotonecase1}
			\ewin{n+1}(v, B, b_{n+1}) = \max\set{\mu_n(v', B \ominus  b_{n+1}) - \weightf(v, v'), 0}
		\end{align}
		On the other hand, with respect to \(\ewin{n}(v, B, b_{n+1})\), \(v'\) is just a vertex (where the minimum may or may not have attained)
		Thus, 
		\begin{align*}
			\max\set{\mu_{n-1}(v', B \ominus b_{n+1}) - \weightf(v, v'), 0} &\geq \ewin{n}(v, B, b_{n+1})\\
			&= \max\set{\ewin{n}(v, B, b_{n+1}), \elose{n}(v, B, b_{n+1})}\\
			&>\max\set{\ewin{n+1}(v, B, b_{n+1}), \elose{n+1}(v, B, b_{n+1})}\\
			&\geq \ewin{n+1}(v, B, b_{n+1})\\
			&= \max\set{\mu_n(v', B \ominus b_{n+1}) - \weightf(v, v'), 0}
		\end{align*}
		
		This implies the following:
		\begin{align*}
			\mu_{n-1}(v', B \ominus b_{n+1}) > \mu_{n}(v', B \ominus b_{n+1})
		\end{align*}
		which is a contradiction to the induction hypothesis \(\mu_{n-1} \leq \mu_n\). 
		
		\medskip

		\item \textbf{Case II:} \(\elose{n}(v, B, b_{n+1}) > \ewin{n}(v, B, b_{n+1})\).
		
		\smallskip
		
		Here, we assume that \(v' \in \neighbor{v}\) is a vertex where \(\elose{n}(v, B, b_{n+1})\) attains its maximum. 
		In other words, 
		\begin{align}\label{eq: monotonecase2}
			\elose{n}(v, B, b_{n+1}) = \max\set{\mu_{n-1}(v', B \oplus \trump(B, b_{n+1})) - \weightf(v, v'), 0}
		\end{align} 
		Thus, 
		\begin{align*}
			\max\set{\mu_{n-1}(v', B \oplus \trump(B, b_{n+1})) - \weightf(v, v'), 0} &= \elose{n}(v, B, b_{n+1})\\
			&= \max\set{\ewin{n}(v,  B, b_{n+1}), \elose{n}(v, B, b_{n+1})}\\
			&> \max\set{\ewin{n+1}(v, B, b_{n+1}), \elose{n+1}(v, B, b_{n+1})}\\
			&\geq \elose{n+1}(v, B, b_{n+1})\\
			&\geq \max\set{\mu_n(v', B \oplus \trump(B, b_{n+1})) - \weightf(v, v'), 0}
		\end{align*}
		This implies 
		\begin{align*}
			\mu_{n-1}(v', B \oplus \trump(B, b_{n+1})) > \mu_n(v', B \oplus \trump(B, b_{n+1}))
		\end{align*}
		which again contradicts the same induction hypothesis. 
	\end{itemize}
	
	Since, we arrive at a contradiction in both the exhaustive cases, we can conclude that there cannot be such a configuration \(\zug{v, B}\) for which \(\sqglymu{n}(v, B) > \sqglymu{n+1}(v,B)\). 
	In other words, \(\sqglymu{n} \leq \sqglymu{n+1}\), as required.
\end{proof}

\section{Proof of Lemma~\ref{lemma: conswinsinfiniteenergy}}\label{appn: conswinsinfiniteenergy}
\conswinsinfiniteenergy*

\begin{proof}
	When \(\mu(v, B) < \infty\), then $\mu \equiv \mu_n$, and the proof is as in Lem.~\ref{lem:conswinsfinite}. Assume $\mu(v, B) = \infty$.
	Let \(n\) such that \(\mu_n(v, B) > \upperboundforloop + 1\). By Lem.~\ref{lem:conswinsfinite}, for every energy \(e \leq |V|kW \leq  \mu_n(v, B) - 1\), \Cons \(e\)-wins in \(\G_n\), as well as in \(\G\). 
	
	The remaining case is \(e \geq \mu_n(v, B)\). 
	Recall that a configuration in $\G_n$ is $\zug{v', e', m, B'}$, where $\zug{v', B'} \in V \times [k]$ is a configuration in $\G$, the accumulated energy is $e'$, and $m$ is a ``counter'' that marks the remaining turns in the game. By Lem.~\ref{lem:conswinsfinite}, \Cons has a strategy $\tau_n$ that wins from $\zug{v, \mu_n(v, B) - 1, n, B}$ in \(\G_n\). 
	Thus, every play $\pi$ that is consistent with $\tau_n$ eventually reaches $\zug{v', -1, m, B'}$ with $m \geq 0$. Note that $\pi$ is finite and denote its length by $|\pi|$.
	Let $\rho= \path(\pi)$ be the path that is traversed in $\G$. Observe that $\sumofweights{\rho} < -\mu_n(v, B)$.
	
	We claim that \(|\pi| > |V| \cdot k\).  Indeed, at most $W$ units of energy can be consumed in each turn, thus consuming $e$ units takes at least $e/W > |V|\cdot k$ turns. Since the number of configuration in $\G$ is $|V| \cdot k$, it follows that a play $\pi$ that is consistent with $\tau_n$ must contain a cycle: there must be a prefix that ends in $\zug{v_2, e_2, m_2, B_2}$ and previously traverses a configuration $\zug{v_1, e_1, m_1, B_1}$ with $v_1 = v_2$ and $B_1 = B_2$. 
	
	Moreover, we claim that $\pi$ traverses a configuration cycle that accumulates negative energy. Let $\chi$ be a configuration cycle and write $\pi = \pi_1 \cdot \chi \cdot \pi_2$. If $\sumofweights{\chi} < 0$, we are done. Otherwise, consider the play $\pi' = \pi_1 \cdot \pi_2$ that is obtained by removing $\chi$. Note that $\sumofweights{\pi'} = \sumofweights{\pi} - \sumofweights{\chi} \leq -|V| \cdot k \cdot W + 0$, thus, as above, $|\pi'| > |V| \cdot k$ and $\pi'$ contains another configuration cycle, and we apply the same argument until a negative cycle is found. 
	
	We construct a strategy $\tau$ in $\G$ that simulates $\tau_n$. Recall that $\G$ starts from $\zug{v, B}$ with energy $e$. The simulation in $\G_n$ starts from $c_0 = \zug{v, \mu_n(v, B) - 1, n, B}$. We define $\tau$ to choose $\tau_n(c_0)$. Suppose that the next configuration in $\G$ is $\zug{v', B'}$, then we simulate \Pres in $\G_n$ so that the next configuration is $c_1 = \zug{v', \mu_n(v, B) - 1 + w(v,v'), n-1, B'}$ and define $\tau$ to choose $\tau_n(c_0, c_1)$. By continuing in this manner, the play $\pi'$ in $\G$ corresponds to a {\em simulated} play $\pi$ in $\G_n$; when $\pi'$ visits $\zug{v, B}$, then $\pi$ visits some $\zug{v, e, m, B}$. Moreover, $\pi$ is consistent with $\tau_n$. By the above, $\pi$ eventually closes a negative cycle. Write $\pi = c_0, \ldots, c_{i_1}, \ldots, c_{i_2}$ with $c_{i_j} =  \zug{v_{i_j}, e_{i_j}, m_{i_j}, B_{i_j}}$ for $j \in \set{1,2}$ such that $v_{i_1} = v_{i_2}$, $B_{i_1} = B_{i_2}$, and $\sumofweights{c_{i_1},\ldots, c_{i_2}} < 0$. We define $\tau$ to follow $\tau_n(\pref{i_1})$, that is we omit the negative cycle from the simulated play. 
	Suppose that the next configuration in $\G$ is $\zug{v', B'}$. Then, the simulated play in $\G_n$ is $\pref{i_1} \cdot \zug{v', e_{i_1} +w(v_{i_1}, v'), m_{i_1}-1, B'}$. We continue in a similar manner. The simulation $\pi$ is obtained from $\pi'$ by repeatedly taking the minimal prefix $i$ such that $\pi'^{\leq i}$ ends in a negative cycle, and omitting the cycle. We define $\tau(\pi') = \tau_n(\pi)$. 
	
	To conclude the proof, we claim that every play $\pi'$ that is consistent with $\tau$ in $\G$ is winning for \Cons. Let $m \in \Nat$. Let $\pref{\ell}$ denote the simulation in $\G_n$ of $\pi'^{\leq m}$, for $\ell \leq m$. Observe that $\pref{\ell}$ contains no negative cycles, thus it is not winning for \Cons, $\ell \leq n$, thus $\sumofweights{\pref{\ell}} \leq n \cdot W$. Choosing $m$ such that $\pi'^{\leq m}$ traverses at least $n \cdot W +e$ negative cycles, each of which has sum of weights at most $-1$, gives $e + \sumofweights{\pi'^{\leq m}} < 0$, as required. 
\end{proof}

\section{Proof of Thm.~\ref{thm:MP-energy}}\label{app: MP-energy}
\mpenergy*

\begin{proof}
	First, assume $\energy(v, B) = M \in \Nat$. \Max follows a \Pres $M$-winning strategy $\sigma$. Let $\pi$ be a consistent play. Since $\sigma$ is winning, for every $n \in \Nat$, we have $M + \sumofweights{\pref{n}} \geq 0$, thus $\meanpayoff(\pi) = \lim_{n \to \infty} \frac{1}{n} \sumofweights{\pref{n}} \geq \lim_{n \to \infty} \frac{-M}{n} = 0$. Second, assume $\energy(v, B) = \infty$. \Min follows the strategy $\tau$ that is constructed in Lem.~\ref{lemma: conswinsinfiniteenergy}. The proof of Lem.~\ref{lemma: conswinsinfiniteenergy} shows that every prefix of $\pi$ consists of negative cycles of length bounded by $|V| \cdot k$ with an additional acyclic path that accumulates energy at most $(|V| \cdot k -1) W$. It is not hard to show that this implies $\meanpayoff(\pi) < 0$.
\end{proof}

\section{Proof of Thm.~\ref{thm:preThreshAverage}}
\label{app:preThreshAverage}
\preThreshAverage*
			
			\begin{proof}
			By definition, \(\preThresh(v)\) is the minimum budget \(B\) for which \(\mu(v, B) < \infty\). 
			In other words, \(\preThresh(v)\) is the unique budget for vertex \(v\) such that \(\mu(v, \preThresh(v)) < \infty\) and \(\mu(v, \predb{\preThresh(v)}) = \infty\).
			In order to show that \(\preThresh\) satisfies the average property, in the following we define another function \(f: V \rightarrow \budgetset\) which maps  \(v\)  to roughly the average of \(\preThresh\) values of its neighbors, and show that \(f\) coincides with \(\preThresh\) for all vertex \(v\). 
			Formally, 
			\begin{align*}
				f(v) \coloneqq \floorvT{\sumTf{\preThresh}} + \eps 
			\end{align*} 
		where \(\vminus{\preThresh} = \arg\min_{v' \in \neighbor{v}} \preThresh(v')\), \(\vplus{\preThresh} = \arg\max_{v' \in \neighbor{v}} \preThresh(v')\), and 
		\begin{align}\label{eq: epsforthreshpr}
			\eps = 
			\begin{cases}
				0 &\text{~if~} \sumTf{\preThresh} \text{~is even and~} \preThresh(\vminus{\preThresh}) \in \Nat\\
				1 &\text{~if~} \sumTf{\preThresh} \text{~is odd and~} \preThresh(\vminus{\preThresh}) \in \Natstr\\
				* &\text{~otherwise}
			\end{cases}\tag{\text{eps}}
		\end{align}
		Below, we omit \(\preThresh\) from \(\vplus{\preThresh}\) and \(\vminus{\preThresh}\), and simply refer them as \(\vplus\) and \(\vminus\), respectively. 
		
		Thus, we proceed by showing the following two, in order to establish that \(f\) coincides with \(\preThresh\). 
		\begin{itemize}
			\item For all \(v \in V\), \(\mu(v, f(v)) < \infty\), and 
			\item \(\mu(v, \predb{f(v)}) = \infty\)
		\end{itemize} 
	
		\begin{claim}\label{clm: mufinite}
			For all \(v \in V\), \(\mu(v, f(v)) < \infty\)
		\end{claim}
 		
 		\begin{proof}
 		We define bids for $ f(v) $.  
 		Intuitively, this is exactly \(\optbid{\preThresh}\), except we cannot directly use that notation now, since it requires \(\preThresh\) to satisfy the average property (which we are proving right now). 
 		Thus, we define a candidate bid $ \optbidthresh(v) $:
 		\begin{align*}
 			\optbidthresh(v) = 
 			\begin{cases}
 				\floorvT{\diffvT{\preThresh}} &\text{~if \(\eps = 0\) or \(1\)}\\[10pt]
 				\predbInt{\floorvT{\diffvT{\preThresh}}} &\text{~if \(\eps = *\) and \(\preThresh(\vminus) \in \Natstro\)}\\[10pt]
 				\succbInt{\floorvT{\diffvT{\preThresh}}} &\text{~otherwise}
 			\end{cases}
 		\end{align*}
 		
 		We claim that when \Pres has a budget of \(f(v)\) at \(v\), she can indeed bid \(\optbidthresh{\preThresh}(v)\), i.e, \(f(v) \ominus \optbidthresh(v) \geq 0\). 
 		The feasibility of such a bids is similar to the argument for Lem.\ref{lem:Talgebra}. 
 	
 		We proceed to first show that \(\mu(v, f(v)) < \infty\) for all \(v \in V\), and that will be followed by the proof of \(\mu(v, \predb{f(v)}) = \infty\) to conclude that \(f\) coincides with \(\preThresh\). 
 		
 		\medskip
 		
 		In order to show that \(\mu(v, f(v)) = \mu(v, \floorvT{\sumvT{\preThresh}} + \eps)< \infty\) for all \(v \in V\), we delve into a case analysis of four cases, where each case correspond to a 
 		parity of \(\sumvT{\preThresh}\) and the advantage status of \(\preThresh(\vminus)\).
 		
 		\begin{itemize}
 			\item {\bf Case I:}\label{case: I} \(\sumvT{\preThresh}\) is even and \(\preThresh(\vminus)\) is in \(\Nat\).
 			
 			In this case, \(\eps = 0\) and \(\optbidthresh(v) = \fracvT{\diffT}\) by their respective definitions. 
 			Recall that
 			\begin{align*}
 				\ewin(v, B, b) &= \min_{v' \in \neighbor{v}}\max\set{\mu(v', B \ominus b) - \weightf(v, v'), 0}\\
 				\elose(v, B, b) &= \max_{v' \in \neighbor{v}}\max\set{\mu(v', B \oplus \trump(B, b)) - \weightf(v, v'), 0}, \text{~and}\\
 				\mu(v, B) &= \min_{b \leq B} \max \set{\ewin(v, B, b), \elose(v, B, b)}
 			\end{align*}
 			
 			Since \(\optbidthresh(v)\) is one possible bid for the budget \(\fracvT{\sumvT{\preThresh}}\), we have 
 			\begin{align*}
 				&\mu\left(v, \fracvT{\sumvT{\preThresh}}\right)\\
 				&\leq \max\left\{\ewin\left(v, \fracvT{\sumvT{\preThresh}}, \fracvT{\diffT{\preThresh}}\right), \right.\\
 				&\hspace*{5em} \left. \elose\left(v, \fracvT{\sumvT{\preThresh}}, \fracvT{\diffT{\preThresh}}\right)\right\}
 			\end{align*}
 			
 			We show that both \(\ewin\left(v, \fracvT{\sumvT{\preThresh}}, \fracvT{\diffT{\preThresh}}\right)\) and \linebreak\(\elose\left(v, \fracvT{\sumvT{\preThresh}}, \fracvT{\diffT{\preThresh}}\right)\) are finite, which implies \linebreak \(\mu\left(v, \fracvT{\sumvT{\preThresh}}\right)\) is finite itself, from the above inequality. 
 			 
 			 \smallskip 
 			 
 			First, 
 			\begin{align*}
 				&\ewin\left(v, \fracvT{\sumvT{\preThresh}}, \fracvT{\diffT{\preThresh}}\right) \\
 				&= \min_{v' \in \neighbor{v}} \max\left\{ \mu\left(v', \absolut{\preThresh(\vminus)}\right) - \weightf(v, v'), 0\right\}\\
 				&\leq \max\set{\mu(\vminus, \preThresh(\vminus)) - \weightf(v, \vminus), 0} &&\hspace*{-8em}[\text{Since \(\preThresh(\vminus) \in \Nat\)}]\\
 				&< \infty &&\hspace*{-8em}[\text{As \(\mu(\vminus, \preThresh(\vminus))\) is finite by definition}]
 			\end{align*}
 		
 		\smallskip
 		
 		On the other hand, consider an arbitrary vertex \(v' \in \neighbor{v}\), for which we have 
 		\begin{align*}
 			k^* \geq \absolut{\preThresh(\vplus)}^* \geq \preThresh(\vplus) \geq \preThresh(v')
 		\end{align*}
 		
 		\SSn{We need to address the case that what if \(\preThresh(\vplus)= k+1\). The argument is if indeed \(\preThresh(\vplus)= k+1\), then \Cons would not have enough budget to bid \(\trump{B, b}\), since otherwise \(B \oplus \trump(B, b)\) would have been \( \absolut{\preThresh(\vplus)}^* = (k+1)^*\). As a consequence, we do not need to consider showing \(\elose(v, B, b)\) in that case. This is a general observation, not restricted to this case. Thus, maybe it should appear when we define \(\ewin\), \(\elose\), \(\mustar\) etc? }

 		From the monotonicity of $ \mu $ we get
 		\begin{align*}
 			\mu(v', \absolut{\preThresh(\vplus)}^*) \leq \mu(v', \preThresh(v')) < \infty
 		\end{align*}
 		
 		Since, \(\mu(v', \absolut{\preThresh(\vplus)}^*)\) is finite for any arbitrary \(v' \in \neighbor{v}\), the finiteness would hold when we take the maximum over all the neighbors of \(v\). 
 		In other words, 
 		\begin{align*}
 			&\elose\left(v, \fracvT{\sumvT{\preThresh}}, \fracvT{\diffT{\preThresh}}\right)\\
 			&= \max_{v' \in \neighbor{v}} \max\set{\mu(v', \absolut{\preThresh(\vplus)}^* - \weightf{v, v'}), 0} < \infty 
 		\end{align*}

 		Therefore, \(\mu\left(v, \fracvT{\sumvT{\preThresh}}\right)\) is finite. 
 		
 		\medskip
 		
 		\item {\bf Case II:} \(\sumvT{\preThresh}\) is odd and \(\preThresh(\vminus)\) is in \(\Natstr\). 
 		
 		In this case, \(\eps = 1\) and \(\optbidthresh(v) = \floorvT{\diffT{\preThresh}}\) by their respective definitions. 
 		
 		Therefore, 
 		\begin{align*}
 			&\mu\left(v, \floorvT{\sumvT{\preThresh}} + 1\right)\\
 			&\leq \max\left\{\ewin\left(v, \floorvT{\sumvT{\preThresh}} + 1, \floorvT{\diffT{\preThresh}}\right), \right.\\
 			&\hspace*{5em} \left. \elose\left(v, \floorvT{\sumvT{\preThresh}} +1 , \floorvT{\diffT{\preThresh}}\right)\right\}\\
 			&= \max\left\{\min_{v' \in \neighbor{v}} \max \set{\mu(v', 	\absolut{\preThresh(\vminus)} + 1) - \weightf(v, v'), 0}\right.\\
 			&\hspace*{5em}\left. \max_{v' \in \neighbor{v}} \max \set{\mu(v', 	\absolut{\preThresh(\vplus)}^*) - \weightf(v, v'), 0}\right\}
 		\end{align*}
 	
 	We show that \(\min_{v' \in \neighbor{v}} \max \set{\mu(v', 	\absolut{\preThresh(\vminus)} + 1) - \weightf(v, v'), 0}\) and \(\max_{v' \in \neighbor{v}} \max \set{\mu(v', 	\absolut{\preThresh(\vplus)}^*) - \weightf(v, v'), 0}\) are finite, implying \(\mu\left(v, \floorvT{\sumvT{\preThresh}} + 1\right)\) be finite as well.

 	First, we have
 	\begin{align*}
 		&\min_{v' \in \neighbor{v}} \max \set{\mu(v', 	\absolut{\preThresh(\vminus)} + 1) - \weightf(v, v'), 0} \\
 		&\leq \max\set{\mu(\vminus, \absolut{\preThresh(\vminus)} + 1) - \weightf(v, \vminus), 0}\\
 		&\leq \max\set{\mu(\vminus, \preThresh(\vminus)) - \weightf{v, \vminus}, 0} &&[\text{From monotonicity of $ \mu $}]\\
 		&<\infty &&[\text{Since \(\mu(\vminus, \preThresh(\vminus))\) is finite}]
 	\end{align*}
 
 	Furthermore, the explanation why \(\max_{v' \in \neighbor{v}}\max\set{\mu(v', 	\absolut{\preThresh(\vplus)}^*) - \weightf(v, v'), 0}\) is finite is already given in Case I already, thus we conclude with \(\mu\left(v, \floorvT{\sumvT{\preThresh}}\right)\) is finite. 
 	
 	\medskip
 	
 	\item {\bf Case III:} \(\sumvT\) is even and \(\preThresh(\vminus)\) is in \(\Natstr\). 
 	
 	\smallskip
 	
 	In this case, \(\eps = *\) and \(\optbidthresh(v) = \predbInt{\fracvT{\diffT{\preThresh}}}\). 
 	As usual, we proceed as follows:
 	\begin{align*}
 		&\mu\left(v, \succInt{\fracvT{\sumvT{\preThresh}}}\right) \\
 		&\leq \max\left\{\ewin\left(v, \succInt{\fracvT{\sumvT{\preThresh}}}, \predbInt{\fracvT{\diffT{\preThresh}}}\right), \right.\\
 		&\hspace*{5em} \left. \elose\left(v, \succInt{\fracvT{\sumvT{\preThresh}}} , \predbInt{\fracvT{\diffT{\preThresh}}}\right)\right\}\\
 		&= \max\left\{\min_{v' \in \neighbor{v}} \max \set{\mu(v', 	\absolut{\preThresh(\vminus)} + 1) - \weightf(v, v'), 0}\right.\\
 		&\hspace*{5em}\left. \max_{v' \in \neighbor{v}} \max \set{\mu(v', 	\absolut{\preThresh(\vplus)}^*) - \weightf(v, v'), 0}\right\}
 	\end{align*}
 
 	Note that, both the constitutes inside \(\max\{\}\) coincide with that of Case II. 
 	Since, they are already shown finite, following the exact same argument here we conclude that \(\mu\left(v, \succInt{\fracvT{\sumvT{\preThresh}}}\right)\) is finite as well. 
 	
 	\smallskip
 	
 	\item {\bf Case IV:} \(\sumvT{\preThresh}\) is odd and \(\preThresh(\vminus)\) is in \(\Nat\). 
 	
 	\smallskip
 	
 	In this case, \(\eps = *\) and \(\optbidthresh(v) = \succbInt{\floorvT{\diffT{\preThresh}}}\). 
 	Furthermore, 
 	\begin{align*}
 		&\mu\left(v, \succbInt{\floorvT{\sumvT{\preThresh}}}\right)\\
 		&\leq \max\left\{\ewin\left(v, \succbInt{\floorvT{\sumvT{\preThresh}}}, \succbInt{\floorvT{\diffT{\preThresh}}}\right), \right.\\
 		&\hspace*{5em} \left. \elose\left(v, \succbInt{\floorvT{\sumvT{\preThresh}}} , \succbInt{\floorvT{\diffT{\preThresh}}}\right)\right\}\\
 		&= \max\left\{\min_{v' \in \neighbor{v}} \max \set{\mu(v', 	\preThresh(\vminus)) - \weightf(v, v'), 0}\right.\\
 		&\hspace*{5em}\left. \max_{v' \in \neighbor{v}} \max \set{\mu(v', 	\absolut{\preThresh(\vplus)}^*) - \weightf(v, v'), 0}\right\}
 	\end{align*}
 		
 		Here, both the constitutes inside \(\max\{\}\) coincide with that Case I, thus following the exact same argument for their finiteness, we can conclude that \(\mu\left(v, \succbInt{\floorvT{\sumvT{\preThresh}}}\right)\) is finite. 
 	
 		\end{itemize}
 	
 		Thus, we conclude the proof of \(\mu(v, f(v)) < \infty\) for all \(v \in V\). 
 	\end{proof}
 		
 		\bigskip
 		
 		\begin{claim}\label{clm: muinfinite}
 			\(\mu(v, \predb{f(v)}) = \infty\)
 		\end{claim}
 		
 		\begin{proof}
 		We will now show that \(\mu(v, \predb{f(v)}) < \infty\) by another similar case analysis in the following. 
 		We first describe the common structure that each of these cases will follow. 
 		
 		\smallskip
 		
 		For simplicity, let us fix \(B = \predb{\floorvT{\sumT} + \eps}\) for rest of this proof. 
 		In order to show that \(\mu(v, B) = \infty\), we show that for any bid \(b \leq B\), at least one of \(\ewin(v, B, b)\) and \(\elose(v, B, b)\) is \(\infty\). 
 		In each case, we distinguish between two scenarios and analyze them separately: (1) \(b < \optbidthresh(v)\), and (2) \(b \geq \optbidthresh(v)\). 
 		
 		\begin{itemize}
 			\item {\bf Case I:} \(\sumvT{\preThresh}\) is even and \(\preThresh(\vminus)\) is in \(\Nat\). 
 			
 			\smallskip 
 			
 			Recall that, in this case, \(\eps = 0\), \(B = \predb{\fracvT{\sumvT{\preThresh}}}\), and \(\optbidthresh(v)= \fracvT{\diffT{\preThresh}}\). 
 			
 			We first consider \(b < \optbidthresh\). 
 			It is easy to see that,
 			\begin{align*}\label{ineq: a}
 				\trump{B, b} \leq \succb{b} \leq \optbidthresh(v)\tag{a}
 			\end{align*}
 			That is, \(B + \trump{B, b} \leq B + \optbidthresh(v)\). 
 			We now proceed as follows:
 			\begin{align*}
 				&\max_{v' \in \neighbor{v}}\max\set{\mu(v', B \oplus \trump(B, b)) - \weightf{v, v'}, 0}\\
 				&\hspace*{1em}\geq \max\set{\mu(\vplus, B \oplus \trump(B, b)) - \weightf{v, \vplus}, 0}&&\hspace*{-15em}[\text{Since \(\vplus \in \neighbor{v}\)}]\\
 				&\hspace*{1em}\geq \max\set{\mu(\vplus, B \oplus \optbidthresh(v)) - \weightf{v, \vplus}, 0} &&\hspace*{-15em}[\text{Applying monotonicity of $ \mu $ on \eqref{ineq: a}}]\\
 				&\hspace*{1em}= \max\left\{ \mu\left(\vplus, \predb{\fracvT{\sumvT{\preThresh}}} \oplus \fracvT{\diffT{\preThresh}}\right) - \weightf{v, \vplus}, 0\right\}\\
 				&\hspace*{1em}=\max\left\{\mu\left(\vplus, \predbInt{\absolut{\preThresh(\vplus)}}\right) - \weightf{v, \vplus}, 0\right\} = \infty
 			\end{align*} 
 		
 			The very last equality is due to the fact that \(\preThresh(\vplus)\) is the minimum budget for which \(\mu(\vplus, \preThresh(\vplus)) < \infty\), for any budget (which includes \(\succInt{\absolut{\preThresh(\vplus)} -1}\)), it is \(\infty\). 
 			
 			\smallskip
 			
 			We now consider \(b \geq \optbidthresh(v)\). 
 			It implies \(B \ominus b \leq B \ominus \optbidthresh(v)\). 
 			Consider an arbitrary \(v' \in \neighbor{v}\), for which we have the following:
 			\begin{align*}
 				&\max\set{\mu(v', B \ominus b) - \weightf{v, v'}, 0}\\
 				&\hspace{2em}\geq \mu(v', B \ominus b) - \weightf{v, v'}\\
 				&\hspace*{2em}\geq \mu(v', B \ominus \optbidthresh(v)) - \weightf{v, v'} &&\hspace*{-10em}[\text{Using monotonicity of $ \mu $}]\\
 				&\hspace*{2em}= \mu\left(v', \predb{\fracvT{\sumvT{\preThresh}}} \ominus \fracvT{\diffT{\preThresh}}\right) - \weightf{v, v'}\\
 				&\hspace*{2em}= \mu\left(v', \predbInt{\absolut{\preThresh(v')}}\right) - \weightf{v, v'} = \infty
 			\end{align*}
 			
 			Since this holds for any arbitrary \(v' \in \neighbor{v}\), it holds when we take the minimum over \(\neighbor{v}\) as well. Thus, 
 			\begin{align*}
 				\min_{v' \in N(v)}\max\set{\mu\left(v', B \ominus b\right) - \weightf{v, v'}, 0} = \infty
 			\end{align*}
 			
 			Thus, we have shown for each \(b \leq B\), either one of \(\ewin(v, B, b)\) or \(\elose(v, B, b)\) is \(\infty\), implying \(\mu(v, B) = \mu(v, \predb{f(v)})\) be \(\infty\) as well, in this case. 
 			
 			\bigskip

 			 \item {\bf Case II:} \(\sumvT{\preThresh}\) is odd and \(\preThresh(\vminus)\) is in \(\Natstro\).
 			 
 			 \smallskip
 			 
 			 In this case, \(\eps = 1\), \(B = \succbInt{\floorvT{\sumT{\preThresh}}}\), and \(\optbidthresh(v) = \floorvT{\diffT{\preThresh}}\). 
 			 
 			 We first consider \(b \leq \optbidthresh(v)\). 
 			 We argue that for such \(b\), we have \(\trump(B, b) \leq \optbidthresh(v)\). 
 			 For \(b < \optbidthresh(v)\), we indeed have \(\trump(B, b) \leq \succb{b} \leq \optbidthresh(v)\). 
 			 On the other hand, for \(b = \optbidthresh(v)\), we note that \(B \in \Natstro\) but \(b \in \Nat\). 
 			 Therefore, by definition \(\trump(B, b) = b\).  
 			 Hence, for all \(b \leq \optbidthresh(v)\), we indeed have \(\trump(B, b)\leq \optbidthresh(v)\). 
 			  In other words, \(B + \trump(B, b) \leq B + \optbidthresh(v)\).

 			  We now proceed similar to how we did in Case I:
 			  
 			  \begin{align*}
 			  	&\ewin(v, B, b)\\ 
 			  	&= \max_{v' \in \neighbor{v}}\max\set{\mu(v', B \oplus \trump(B, b)) - \weightf{v, v'}, 0}\\
 			  	&\geq \max\set{\mu(\vplus, B \oplus \trump(B, b)) - \weightf{v, \vplus}, 0}&&\hspace*{-15em}[\text{Since \(\vplus \in \neighbor{v}\)}]\\
 			  	&\geq \max\set{\mu(\vplus, B \oplus \optbidthresh(v)) - \weightf{v, \vplus}, 0} &&\hspace*{-15em}[\text{Applying monotonicity of $ \mu $ on \eqref{ineq: a}}]\\
 			  	&= \max\left\{ \mu\left(\vplus, \succbInt{\floorvT{\sumT{\preThresh}}} \oplus \floorvT{\diffT{\preThresh}}\right) - \weightf{v, \vplus}, 0\right\}\\
 			  	&=\max\left\{\mu\left(\vplus, \predbInt{\absolut{\preThresh(\vplus)}}\right) - \weightf{v, \vplus}, 0\right\} = \infty
 			  \end{align*}
 		  
 		  	Next, we consider \(b > \optbidthresh(v)\). 
 		  	That is, 
 		  	\begin{align*}
 		  		b \geq \succb{\optbidthresh(v)}
 		  		\implies B \ominus b \leq B \ominus \left(\succb{\optbidthresh(v)}\right)
 		  	\end{align*} 
 	  		
 	  		Similar to Case I, here we pick an arbitrary \(v' \in \neighbor{v}\) for which we have 
 	  		\begin{align*}
 	  			&\max\set{\mu(v', B \ominus b) - \weightf{v, v'}, 0}\\
 	  			&\hspace{2em}\geq \mu(v', B \ominus b) - \weightf{v, v'}\\
 	  			&\hspace*{2em}\geq \mu(v', B \ominus \left(\succb{\optbidthresh(v)}\right)) - \weightf{v, v'} &&\hspace*{-10em}[\text{Using monotonicity of $ \mu $}]\\
 	  			&\hspace*{2em}= \mu\left(v', \succbInt{\floorvT{\sumT{\preThresh}}}\ominus \succbInt{\floorvT{\diffT{\preThresh}}}\right) - \weightf{v, v'}\\
 	  			&\hspace*{2em}= \mu\left(v', \absolut{\preThresh(\vminus)}\right) - \weightf{v, v'} 
 	  		\end{align*}
 		  	
 		  Since, \(\preThresh(\vminus) \in \Natstro\), we have \(\absolut{\preThresh(\vminus)} < \preThresh(v')\) for any \(v' \in \neighbor{v}\).
 		  In other words, \(\absolut{\preThresh(\vminus)} \leq \predb{\preThresh(v')}\) for any \(v' \in \neighbor{v}\). Thus, applying monotonicity of $ \mu $, we get
 		  \begin{align*}
 		  	\mu(v', \absolut{\preThresh(v)}) \geq \mu(v', \predb{\preThresh(v')})
 		  \end{align*}
		which implies, 
		\begin{align*}
			\max\set{\mu(v', B \ominus b) - \weightf(v, v'), 0} \geq \mu(v', \predb{\preThresh(v')}) - \weightf(v, v') = \infty
		\end{align*}
		Since, the above holds for any arbitrary \(v'\) of \(\neighbor{v}\), we can take a minimum over \(\neighbor{v}\), and conclude
		\begin{align*}
			\elose(v, B, b) = \min_{v' \in N(v)}\max\set{\mu(v', B \ominus b) - \weightf(v, v'), 0} = \infty
		\end{align*}
	
		Hence, in this case too, we have shown that for all \(b \leq B\), at least one of \(\ewin(v, B, b)\) and \(\elose(v, B, b)\) is \(\infty\), implying \(\mu(v, B) = \mu(v, \predb{f(v)})\) be \(\infty\) itself. 
		
		\bigskip
		
		\item {\bf Case III: } \(\sumvT{\preThresh}\) is even and \(\preThresh(\vminus)\) is in \(\Natstro\).
		
		\smallskip
		
		In that case, \(\eps = *\), \(B = \predb{f(v)} = \fracvT{\sumT{\preThresh}}\), and \(\optbidthresh(v) = \predbInt{\fracvT{\diffT{\preThresh}}}\). 
		
		It can be noted that here a bid cannot coincide with \(\optbidthresh(v)\), simply because \(B \in \Nat\) and \(\optbidthresh(v) \in \Natstro\). 
		In other words, we just need to take into account of \(b\)'s which are strictly greater or strictly less than 
		\(\optbidthresh(v)\).
		
		We start with \(b < \optbidthresh(v)\). 
		In that case, \(\trump(B, b) \leq \optbidthresh(v)\). 
		Hence, similar to the previous two cases, we proceed as follows:
		\begin{align*}
			&\ewin(v, B, b)\\ 
			&= \max_{v' \in \neighbor{v}}\max\set{\mu(v', B \oplus \trump(B, b)) - \weightf{v, v'}, 0}\\
			&\geq \max\set{\mu(\vplus, B \oplus \trump(B, b)) - \weightf{v, \vplus}, 0}&&\hspace*{-15em}[\text{Since \(\vplus \in \neighbor{v}\)}]\\
			&\geq \max\set{\mu(\vplus, B \oplus \optbidthresh(v)) - \weightf{v, \vplus}, 0} &&\hspace*{-15em}\\
			&= \max\left\{ \mu\left(\vplus, \fracvT{\sumT{\preThresh}} \oplus \predbInt{\fracvT{\diffT{\preThresh}}}\right) - \weightf{v, \vplus}, 0\right\}\\
			&=\max\left\{\mu\left(\vplus, \predbInt{\absolut{\preThresh(\vplus)}}\right) - \weightf{v, \vplus}, 0\right\} = \infty
		\end{align*} 
		
		On the other hand, when \(b > \optbidthresh(v)\), we have for every \(v' \in \neighbor{v}\):
		\begin{align*}
			&\max\set{\mu(v', B \ominus b) - \weightf{v, v'}, 0}\\
			&\hspace{2em}\geq \mu(v', B \ominus b) - \weightf{v, v'}\\
			&\hspace*{2em}\geq \mu(v', B \ominus \left(\succb{\optbidthresh(v)}\right)) - \weightf{v, v'} &&\hspace*{-10em}\\
			&\hspace*{2em}= \mu\left(v', \fracvT{\sumT{\preThresh}} \ominus \fracvT{\diffT{\preThresh}}\right) - \weightf{v, v'}\\
			&\hspace*{2em}= \mu\left(v', \absolut{\preThresh(\vminus)}\right) - \weightf{v, v'} 
		\end{align*}
		Similar to Case II, here we have \(\absolut{\preThresh(\vminus)} \leq \predb{\preThresh(v')}\) for all \(v' \in \neighbor{v}\) as well. 
		Thus, first applying monotonicity of $ \mu $, and then taking minimum over all \(\neighbor{v}\), we get
		\begin{align*}
			\elose(v, B, b) = \min_{v' \in \neighbor{v}}\max\set{\mu(v', B \ominus b) - \weightf(v, v'), 0} = \infty
		\end{align*}
	
		Therefore, in this case too, we have shown that for all \(b \leq B\), at least one of \(\ewin(v, B, b)\) and \(\elose(v, B, b)\) is \(\infty\), implying \(\mu(v, B) = \mu(v, \predb{f(v)})\) be \(\infty\) itself. 
		
		\bigskip
		
		\item {\bf Case IV:} \(\sumvT{\preThresh}\) is odd and \(\preThresh(\vminus)\) is in \(\Nat\).
		
		\smallskip
		
		In this case, \(\eps = *\), \(B = \predb{f(v)} = \floorvT{\sumT{\preThresh}}\), and \(\optbidthresh(v) = \succbInt{\floorvT{\diffT{\preThresh}}}\). 
		
		Similar to Case III, here too, since \(B \in \Nat\) and \(\optbidthresh(v) \in \Natstro\), a bid cannot coincide with \(\optbidthresh(v)\). 
		Thus, we only need to consider when \(b\) is strictly greater than or strictly lower than \(\optbidthresh(v)\). 
		
		We start with \(b < \optbidthresh(v)\). 
		In that case, \(\trump(B, b) \leq \optbidthresh(v)\). 
		Hence, similar to the previous two cases, we proceed as follows:
		\begin{align*}
			&\ewin(v, B, b)\\ 
			&= \max_{v' \in \neighbor{v}}\max\set{\mu(v', B \oplus \trump(B, b)) - \weightf{v, v'}, 0}\\
			&\geq \max\set{\mu(\vplus, B \oplus \trump(B, b)) - \weightf{v, \vplus}, 0}&&\hspace*{-15em}[\text{Since \(\vplus \in \neighbor{v}\)}]\\
			&\geq \max\set{\mu(\vplus, B \oplus \optbidthresh(v)) - \weightf{v, \vplus}, 0} &&\hspace*{-15em}\\
			&= \max\left\{ \mu\left(\vplus, \floorvT{\sumT{\preThresh}} \oplus \succbInt{\floorvT{\diffT{\preThresh}}}\right) - \weightf{v, \vplus}, 0\right\}\\
			&=\max\left\{\mu\left(\vplus, \predbInt{\absolut{\preThresh(\vplus)}}\right) - \weightf{v, \vplus}, 0\right\} = \infty
		\end{align*} 
		
		On the other hand, when \(b > \optbidthresh(v)\), we have for every \(v' \in \neighbor{v}\):
		\begin{align*}
			&\max\set{\mu(v', B \ominus b) - \weightf{v, v'}, 0}\\
			&\hspace{2em}\geq \mu(v', B \ominus b) - \weightf{v, v'}\\
			&\hspace*{2em}\geq \mu(v', B \ominus \left(\succb{\optbidthresh(v)}\right)) - \weightf{v, v'} &&\hspace*{-10em}[\text{Using monotonicity of $ \mu $}]\\
			&\hspace*{2em}= \mu\left(v', \floorvT{\sumT{\preThresh}} \ominus \left(\floorvT{\diffT{\preThresh}} + 1\right)\right) - \weightf{v, v'}\\
			&\hspace*{2em}= \mu\left(v', \absolut{\preThresh(\vminus)} - 1\right) - \weightf{v, v'} \\
			&\hspace*{2em}\geq \mu\left(v', \absolut{\preThresh(v')} - 1\right) - \weightf{v, v'} = \infty
		\end{align*}
	
		Thus, taking minimum over all \(\neighbor{v}\), we get
		\begin{align*}
			\elose(v, B, b) = \min_{v' \in \neighbor{v}}\max\set{\mu(v', B \ominus b) - \weightf(v, v'), 0} = \infty
		\end{align*}
		
		Therefore, in this case too, we have shown that for all \(b \leq B\), at least one of \(\ewin(v, B, b)\) and \(\elose(v, B, b)\) is \(\infty\), implying \(\mu(v, B) = \mu(v, \predb{f(v)})\) be \(\infty\) itself.
 		\end{itemize}
 		This concludes showing $ \mu(v, \predb{f(v)}) $ for all $ v $. 
 \end{proof}
 		
 		Thus, we now have both \(\mu(v, f(v)) < \infty\) and \(\mu(v, \predb{f(v)}) = \infty\). 
 		In other words, for all \(v \in V\), \(f(v)\) coincides with the characteristic property of \(\preThresh(v)\). 
 		Thus, 
 		\begin{align*}
 			\preThresh(v) = \floorvT{\sumT{\preThresh}} + \eps
 		\end{align*}
 		where \(\eps\) is as defined in \eqref{eq: epsforthreshpr}. 
 		Therefore, \(\preThresh\) satisfies the average property. 
 		\end{proof}
 		
\begin{observation}\label{obs: Talgebra}
 	Let \(T\) be a function that satisfies the average property, and \(v \in V\). Then 
 	\begin{itemize}
 		\item If \(T(\vminus) \in \Nat\), then \(T(v) \ominus b^T(v) = T(\vminus)\), and if \(T(\vminus) \in \Natstro\), then \(T(v) \ominus b^T(v) = \absolut{T(\vminus)} + 1\). On the other hand, for both cases, \(T(v) \oplus \left(\succb{b^T(v)}\right) = \succbInt{\absolut{T(\vplus)}}\). 
 		
 		\item \(\left(\succb{T(v)}\right) \ominus \left(\succb{b^T(v)}\right) = T(v) \ominus b^T(v)\), and \(\left(\succb{T(v)}\right) \oplus \left(b^T(v) \oplus 1\right) = \succbInt{\absolut{T(\vplus)}} +1 \)
 	\end{itemize}
 	
 \end{observation}
 		
\section{Proof of Lemma~\ref{lem:acts-coincide}}\label{app: acts-coincide}
\actscoincide*

\begin{proof}
%

		Since \(B \in \set{\preThresh(v), \succb{\preThresh(v)}}\), we have \(\energy(v, B) < \infty\).
		Moreover, \(\energy(v, B) = \ebid(v, B, b')\), where \(b' = \arg\min_{b \leq B} \ebid(v, B, b)\). It follows that we must have both \(\ewin(v,B, b') < \infty\) and
		\(\elose(v, B, b') < \infty\). 
		
		We show that for any bid \(b\) which is either strictly less than \(\optbid{\preThresh}(v, B)\) or strictly greater than \(\succb{\optbid{\preThresh}(v, B)}\), \(\ebid(v, B, b) = \infty\). 
		Thus, \(\ebid(v, B, b)\) cannot attain its minimum for those \(b\)'s. 
		Finally, we show that \(\ebid(v, B, \optbid{\preThresh}(v, B)) \leq \ebid(v, B, \succb{\optbid{\preThresh}(v, B)})\), thus proving \(\ebid(v, B, b)\) attains its minimum at \(b = \optbid{\preThresh}(v, B)\).

		\begin{itemize}
			\item First, we consider \(b < \optbid{\preThresh}(v, B)\). 
			
			There are two possibilities: either \(b < \predb{\optbid{\preThresh}(v, B)}\) or \(b=  \predb{\optbid{\preThresh}(v, B)}\). 
			We claim that, either way \(\trump(B, b) \leq \predb{\optbid{\preThresh}(v, B)}\). 
			For \(b < \predb{\optbid{\preThresh}(v, B)}\), it is trivial. 
			On the other hand, since \(B\) and \(\optbid{\preThresh}(v, B)\)'s tie breaking status coincide, when \(b = \predb{\optbid{\preThresh}(v, B)}\), it is the case when \Pres has the advantage but does not use it. 
			Thus, \(\trump(B, b) = b = \predb{\optbid{\preThresh}(v, B)}\) in this case as well. 
			Thus, 
			\begin{align*}
				&\trump(B, b) \leq \predb{\optbid{\preThresh}(v, B)}\\
				\implies &B \oplus \trump(B, b) \leq B \oplus \left(\predb{\optbid{\preThresh}(v, B)}\right) \leq \predb{\preThresh(\vplus)} \\
				&\hspace*{2em} [\text{~Since \(T(v) \oplus \left(\succb{b^T(v)}\right) = \succbInt{\absolut{T(\vplus)}}\) from Obs.~\ref{obs: Talgebra}}]
			\end{align*}
			Using monotonicity of \(\energy\) with respect to budgets (Corr.~\ref{corr: energyinherits}), for all \(u \in \neighbor{v}\), we have 
			\begin{align*}
				\energy(u,B\oplus \trump(B, b)) \geq \energy(u, \predb{\preThresh(\vplus)})
			\end{align*}
			In particular, 
			\begin{align*}
				\energy(\vplus,B \oplus \trump(B, b)) \geq \energy(\vplus, \predb{\preThresh(\vplus)}) = \infty
			\end{align*}
			Since, \(\preThresh(\vplus)\) is the minimum budget \(B'\) for which \(\energy(\vplus, B') < \infty\). 
			
			Therefore, taking maximum over all \(u \in \neighbor{v}\), we get
			\begin{align*}
				\elose(v,B, b) &= \max_{u \in \neighbor{v}} \max\set{\energy(u,B \oplus \trump(B, b)) - \weightf(v, u), 0} \\
				&\geq \max\set{\energy(\vplus,B \oplus \trump(B, b)) - \weightf(v, \vplus), 0}\\
				&\geq \energy(\vplus,B \oplus \trump(B, b)) - \weightf(v, \vplus) = \infty
			\end{align*}
			
			Hence, \(\ebid(v, B, b) \geq \elose(v, B, b) = \infty\) when \(b <\optbid{\preThresh}(v, B)\).
			
			\bigskip

			\item Second, we consider \(b > \succb{\optbid{\preThresh}(v, B)}\). 
			
			In other words, \(b \geq \optbid{\preThresh}(v, B) \oplus 1\). 
			Then we have 
			\begin{align*}
				B \ominus b \leq B \ominus \left(\optbid{\preThresh}(v, B) \oplus 1\right) &\leq \predb{\preThresh(\vminus)}\\
				&\hspace*{-2em}\text{~[Since \(T(v) \ominus b^T(v) \leq \absolut{T(\vminus)} +1\) from Obs~\ref{obs: Talgebra}]}
			\end{align*}
		From Corr~\ref{corr: energyinherits}, for all \(u \in \neighbor{v}\), we have
			\begin{align*}
				\energy(u, B \ominus b) &\geq \energy(u, \predb{\preThresh(\vminus)}) \geq \energy(u, \predb{\preThresh(u)}) = \infty
			\end{align*}
			The last inequality is derived from the fact that \(\preThresh(u)\) is the minimum budget \(B\) for which \(\energy(u, B) < \infty\). 
			Since this holds for any \(u \in \neighbor{v}\), it should also be satisfied for the minimum. Thus, 
			\begin{align*}
				\ewin(v, B, b) &= \min_{u \in \neighbor{v}} \max\set{\energy(u, B \ominus b) - \weightf(v, u), 0}\\
				&\geq \min_{u \in \neighbor{v}} \energy(u, B \ominus b) - \weightf(v, u) = \infty
			\end{align*}
			Hence \(\ebid(v, B, b) \geq \ewin(v, B, b) = \infty\) as well when \(b > \succb{b^{\preThresh}(v)}\). 
		\end{itemize}
		
		\bigskip
		
		Finally, we will show that \(\ebid(v, B, b)\) attains its minimum at \(b = \optbid{\preThresh}(v, B)\) by explicitly showing that 
		\begin{align*}
			\ewin(v, B, \optbid{\preThresh}(v, B)) &\leq 
			\ewin(v, \preThresh(v), \succb{\optbid{\preThresh}(v, B)}), \text{~and}\\
			\elose(v, B, \optbid{\preThresh}(v, B)) &=
			\elose(v, B, \succb{\optbid{\preThresh}(v, B)})
		\end{align*}
		
		\begin{itemize}
			\item First, from Corr~\ref{corr: energyinherits}, we get for any \(u \in \neighbor{v}\)
			\begin{align*}
				\energy(u, B \ominus \optbid{\preThresh}(v, B)) \leq \energy(u, B \ominus \succb{\optbid{\preThresh}(v, B)})
			\end{align*}
			Thus, subtracting \(\weightf(v, u)\) from both sides, then taking maximum with \(0\), and finally taking the minimum over \(\neighbor{v}\), we arrive at
			\begin{align*}
				\ewin(v, B, \optbid{\preThresh}(v, B)) \leq \ewin(v, B, \succb{\optbid{\preThresh}(v, B)})
			\end{align*}
			
			\medskip
			
			\item Second, recall that the advantage status of \(B\) and \(\optbid{\preThresh}(v, B)\) always coincide. 
			Thus, \(\succb{\optbid{\preThresh}(v, B)}\) is a feasible bid for budget \(B\) only if \(B \in \Natstro\). 
			Consequently, we have \(\succb{\optbid{\preThresh}(v, B)} \in \Nat\), making \(\trump(B, \succb{\optbid{\preThresh}(v, B)}) = \succb{\optbid{\preThresh}(v, B)}\). 
			Therefore for any \(u \in \neighbor{v}\):
			\begin{align*}
				\energy(u, B \oplus \trump(B, \optbid{\preThresh}(v,B))) 
				&= \energy(u, B \oplus \succb{\optbid{\preThresh}(v, B)})\\
				&= \energy(u, B \oplus \trump(B, \succb{\optbid{\preThresh}(v, B)}))\\
			\end{align*}
			Henceforth, first subtracting \(\weightf(v, u)\) from both sides, then taking the maximum with \(0\), and finally by taking the maximum over all \(\neighbor{v}\), we arrive at
			\begin{align*}
				\elose(v, B, \optbid{\preThresh}(v, B)) = \elose(v, B, \succb{\optbid{\preThresh}(v, B)})
			\end{align*}
		\end{itemize}
		
		This concludes the argument that indeed \(\ebid(v, B, b)\) attains its minimum at \(b = \optbid{\preThresh}(v, B)\) where \(B \in \set{\preThresh(v), \succb{\preThresh(v)}}\). 

	We proceed to the choice of vertex. 
	We use \(b' = \optbid{\preThresh}(v, B)\) for \(\zug{v, B}\) to establish that \(\sVI{\sigma}\) must choose a vertex from \(\Allowd{\preThresh}(v)\) at \(\zug{v, B}\). 
	Without loss of generality, we assume that \(B = \preThresh(v)\) since \(\preThresh(v) \ominus b^{\preThresh}(v) = (\succb{\preThresh(v)}) \ominus (\succb{b^{\preThresh}(v)})\). 
	Assume for the sake of contradiction that \(v' \not\in \Allowd{\preThresh}(v)\). 
	That means, if \(\preThresh(v) \in \Nat\), then \(\preThresh(v') > \preThresh(\vminus)\), and otherwise \(\preThresh(v') > \succb{\preThresh(\vminus)}\). 
	We argue for the case \(\preThresh(v) \in \Nat\), and the other case is dual. 
	Since \(\preThresh(v')\) is the minimum budget \(B'\) for which \(\energy(v', B')\) is finite, we have from the assumption \(\preThresh(v') > \preThresh(\vminus)\) that \(\energy(v', \preThresh(\vminus)) = \infty\). 
	Since \(\preThresh(v) \ominus \optbid{\preThresh}(v) = \preThresh(\vminus)\), \Cons can let \Pres win the current bid, making the subsequent configuration \(\zug{v', \preThresh(\vminus)}\), which is  \(e\)-winning for \Cons for any finite energy \(e\), which contradicts that \(\sVI{\sigma}\) is \(\energy(v, \preThresh(v))\)-winning for \Pres from \(\zug{v, \preThresh(v)}\). Thus, \(v'\) must be from \(\Allowd{\preThresh}(v)\). 
\end{proof}

\section{Proof of Claim~\ref{clm: singletransitioninvariant}}
\label{app:singletransitioninvariant}

\singletransitioninvariant*
			
			\begin{proof}

			Recall that Lemma.~\ref{lem:acts-coincide} implies that 
			when \(B \in \set{\preThresh(v), \succb{\preThresh(v)}}\), \(\presagn\) guarantees that no matter what \Cons responses at \(\zug{v, B}\), it maintains the following invariant of energy: 
			\begin{itemize}
				\item If \Pres wins by bidding $ \optbid{\conThresh}(v, B) $, and proceeds to a $u \in \Allowd{\preThresh}(v)  $, then 
				\begin{align*}
					\energy(v, \relativebud(v, B)) + \weightf(v, u) \geq \energy(v, \relativebud(v, B) \ominus \optbid{\preThresh}(v, B))
				\end{align*}
				
				\item If \Cons wins by bidding at least $ \succb{\optbid{\preThresh}(v, B)} $, then for every $ v' \in \neighbor{v} $, 
				\begin{align*}
					\energy(v, \relativebud(v, B)) + \weightf(v, v') \geq \energy(v', \relativebud(v, B) \oplus (\succb{\optbid{\preThresh}(v, B)})) 
				\end{align*}
			\end{itemize}

		However, when we try to extend to $ i^{th} $ transition of $ \pi $, \(\relativebud(v_{m+1}, B_{m+1})\) may not coincide with either of \(\relativebud(v_m, B_m) \ominus \optbid{\preThresh}(v_m, B_m)\) and \(\relativebud(v_m, B_m) \oplus \succb{\optbid{\preThresh}(v_m, B_m)}\). 
		
			\medskip
			
			We will argue by a case analysis that, in fact, in all but one cases they do coincide (thus satisfy the energy invariant from there), and the case where they may not coincide is the only case where the equality \(\spare{m} = \spare{m+1}\) does not hold. 
			This proves that whenever \(\spare{m} = \spare{m+1}\), we get the energy invariant as mentioned. 
			
			\medskip
			
			We first assume \(\spare{m} = l < k\). 
			In other words, \(B_m\) equals to either  \(\preThresh(v) \oplus l\) or \(\preThresh(v) \oplus l^*\). 
			We consider the following exhaustive case analysis. 
			\begin{itemize}
				\item First suppose \(\relativebud(v_m, B_m) = \preThresh(v_m)\), and \Pres wins the bid at \(\zug{v_m, B_m}\). 
				Thus,
				\begin{align*}
					B_{m+1} = B_m \ominus \optbid{\preThresh}{v_m} &= \left(\preThresh(v) \oplus l\right) \ominus \optbid{\preThresh}{v_m}\\
					&= \left(\preThresh(v_{m}) \ominus \optbid{\preThresh}{v_m}\right) \oplus l\\
					&= 
					\begin{cases}
						\preThresh(v_{m+1}) \oplus l &\text{~if \(\preThresh(v_{m+1}) \in \Nat\)}\\
						\absolut{\preThresh(v_{m+1})} \oplus 1 \oplus l &\text{~otherwise}
					\end{cases}\\
				&\hspace*{10em}\text{~[Using Lem.~\ref{lem:Talgebra}]}
				\end{align*} 
				Either way, \(\relativebud(v_{m+1}, B_{m+1})\) coincides with \(\preThresh(v_m) \ominus \optbid{\preThresh}(v_m)\). 
				Hence, the invariant holds.  		

				\smallskip
				
				\item Second, suppose \(\relativebud(v_m, B_m) = \preThresh(v_m)\) but \Cons wins the bid at \(\zug{v_m, B_m}\). 
				Recall that, \Cons bids at least \(\succb{\optbid{\preThresh}{v_m}}\), and chooses a vertex \(u \in \neighbor{v_m}\). 
				In the worst case (from \Pres's perspective), \(v_{m+1}\) can be such that \(\preThresh(v_{m+1}) = \preThresh(\vplus{v_m})\). 
				In this case, 
				\begin{align*}
					B_{m+1} \geq B_m \oplus \succb{\optbid{\preThresh}{v_m}} &= \left(\preThresh(v_m) \oplus l\right) \oplus \succb{\optbid{\preThresh}{v_m}} \\
					&= \succbInt{\absolut{\preThresh(v_{m+1})}} \oplus l
				\end{align*}
				If \(B_{m+1} > \succbInt{\absolut{\preThresh(v_{m+1})}} \oplus l\), we get \(\spare{m+1} > \spare{m}\), and we have nothing further to show.
				Similar is the case when \Cons chooses \(v_{m+1}\) with \(\preThresh(v_{m+1}) < \preThresh(\vplus{v_m})\).  
				Otherwise, \(\relativebud(v_{m+1}, B_{m+1})\) exactly coincides with \(\preThresh(v_m) \oplus \succb{\optbid{\preThresh}{v_m}}\), and thus the energy invariant holds. 
				
				\smallskip
				
				\item Third, supose \(\relativebud(v_m, B_m) = \succb{\preThresh(v_m)}\), and \Pres wins at \(\zug{v_m, B_m}\). 
				Thus,
				\begin{align*}
					B_{m+1} = B_m \ominus \succb{\optbid{\preThresh}} &= \left(\succb{\preThresh(v_m)} \oplus l\right) \ominus \left(\succb{\optbid{\preThresh}}\right)\\
					&= \left(\preThresh(v_m) \ominus \optbid{\preThresh}{v_m}\right) \oplus l\\
					&= 
					\begin{cases}
						\preThresh(v_{m+1}) \oplus l &\text{~if \(\preThresh(v_{m+1}) \in \Nat\)}\\
						\absolut{\preThresh(v_{m+1})} \oplus 1 \oplus l &\text{~otherwise}
					\end{cases}\\
					&\hspace*{10em}\text{~[Using Lem.~\ref{lem:Talgebra}]}
				\end{align*}
				Similar to the first case, here too, either way \(\relativebud(v_{m+1}, B_{m+1})\) coincides with \(\left(\preThresh(v_m) \ominus \optbid{\preThresh}{v_m}\right)\). 
				
				\item Finally, suppose \(\relativebud(v_m, B_m) = \succb{\preThresh(v_m)}\), and \Cons wins at \(\zug{v_m, B_m}\). 
				Moreover, suppose \Cons chooses the vertex with maximum threshold budget among its neighbors, i.e, \(\preThresh(v_{m+1}) = \preThresh(\vplus{v_m})\). 
				Since in this case, \Pres bids \(\succb{\optbid{\preThresh}{v_m}}\), we get 
				\begin{align*}
					B_{m+1} \geq B_m \oplus \optbid{\preThresh}{v_m} \oplus 1 &= \left(\succb{\preThresh(v_m)} \oplus l\right) \oplus \optbid{\preThresh}{v_m} \oplus 1\\
					&= \succbInt{\absolut{\preThresh(\vplus{v_m})}} \oplus 1 \oplus  l\\
					&= \succbInt{\absolut{\preThresh(v_{m+1})}} \oplus 1 \oplus l
				\end{align*}
				Thus, \(\spare{m+1} = \absolut{B_{m+1} \ominus \preThresh(v_{m+1})} \geq l+1\), violating the hypothesis itself. 
				Therefore, we do not need to show the energy invariant in this case. 
			\end{itemize}
		In all the above cases, we observed either the required energy invariant holds, 
		or \(\spare{m+1}\) becomes strictly greater than \(\spare{m}\). 
		\end{proof}

\section{Proof of Thm.~\ref{thm:sagn-is-winning}}
\label{app:thm:sagn-is-winning}
			We proceed to establish the consequences of Claim~\ref{clm: singletransitioninvariant} in the following:
			\begin{itemize}
				\item First, we show by induction that for all \(m \geq N\), for all \(n \geq 0\) the following holds:
				\begin{align*}
					\energy(v_m, \relativebud{v_m, B_m}) + \sum_{i = m}^{m+n} \weightf(v_i, v_{i+1}) \geq 0
				\end{align*}
				Note that, the induction is only on \(n\). 
				Thus, each step of the induction comprises the universal quantification over \(m\). 
				
				The base case is \(n = 0\). 
				It requires us to argue that for all \(m \geq N\), \(\energy(v_m, \relativebud{v_m, B_m}) \geq 0\). 
				This is true because \(\energy\) is a function that maps each configuration to a non-negative integer by construction. 
				
				We assume that the statement is true for some fixed \(n\), and will show for \(n+1\). 
				To be precise, the induction hypothesis says: for all \(m \geq N\) and for a fixed \(n\), \(\energy(v_m, \relativebud{v_m, B_m}) + \sum_{i = m}^{m+n} \weightf(v_i, v_{i+1}) \geq 0\). 
				In the inductive step, we consider an arbitrary \(m \geq N\), and \(n+1\). 
				There, we have
				\begin{align*}
					&\energy(v_m, \relativebud{v_m, B_m}) + \sum_{i = m}^{m+n+1} \weightf(v_i, v_{i+1})\\
					&\hspace*{1em}= \energy(v_m, \relativebud{v_m, B_m}) + \weightf(v_m, v_{m+1}) + \sum_{i = m+1}^{m+n+1} \weightf(v_i, v_{i+1})\\
					&\hspace*{1em}\geq \energy(v_{m+1}, \relativebud{v_{m+1}, B_{m+1}}) + \sum_{i = m+1}^{m+n+1} \weightf(v_i, v_{i+1})\\
					&\hspace*{10em} \text{~[By applying Claim~\ref{clm: singletransitioninvariant} since for all \(m \geq N, \spare{m} = r\)]~}\\
					&\hspace*{1em}=  \energy(v_{m'}, \relativebud{v_{m'}, B_{m'}}) + \sum_{i = m'}^{m'+n} \weightf(v_i, v_{i+1}) \hspace*{1em}\text{~[Setting \(m' = m+1 \geq N\)]}\\
					&\hspace{1em}\geq 0  \hspace*{19.6em}\text{~[From induction hypothesis]}
				\end{align*}
				
				In particular, the statement holds true for \(m = N\), which implies that if the energy at \(\zug{v_N, B_N}\) is at least \(\energy(v_N, \relativebud(v_N, B_N))\), then for the rest of the \(\pi\), it remains non-negative. 
				
				\item Second, we consider the finite prefix \(\finprefix\) of length \(N+1\) of \(\pi\), i.e, \(\finprefix = \zug{v_0, B_0}, \ldots \zug{v_N, B_N}\). 
				We call an infix \(h\) of \(\finprefix\) a \emph{patch} if  \(\spare\) value is fixed in \(\rho\), and is maximal. 
				Since \(\spare\) is monotonically increasing, we observe that \(\finprefix\) is a sequence of such monotonically \(\spare\)-value increasing patches, until it reaches \(\zug{v_N, B_N}\). 
				Consider such a patch \(h = \zug{v_j, B_j}, \ldots \zug{v_{j+p}, B_{j+p}}\). 
				We show by induction that for all \(0 \leq n \leq p-1\):
				\begin{align*}
					\energy(v_j, \relativebud(v_j, B_j)) + \sum_{i = j}^{j+n}\weightf(v_i, v_{i+1}) \geq \energy(v_{j+n+1}, \relativebud(v_{j+n+1}, B_{j+n+1})) 
				\end{align*}
				The base case is \(n = 0\), and is directly derived from Claim~\ref{clm: singletransitioninvariant}: \(\energy(v_j, \relativebud(v_j, B_j)) + \weightf(v_j, v_{j+1}) \geq \energy(v_{j+1}, \relativebud(v_{j+1}, B_{j+1}))\) since \(\spare{j} = \spare{j+1}\). 
				We assume that the statement is true for some fixed \(n < p-1\), and will show it for \(n+1\). 
				We start with 
				\begin{align*}
					&\energy(v_j, \relativebud(v_j, B_j)) + \sum_{i = j}^{j+n+1}\weightf(v_i, v_{i+1})\\
					&\hspace{1em}=  \energy(v_j, \relativebud(v_j, B_j)) + \sum_{i = j}^{j+n}\weightf(v_i, v_{i+1}) + \weightf(v_{j+n+1}, v_{j+n+2})\\
					&\hspace{1em}\geq \energy(v_{j+n+1}, \relativebud(v_{j+n+1}, B_{j+n+1})) + \weightf(v_{j+n+1}, v_{j+n+2}) \text{~[From induction hypothesis]}\\
					&\hspace{1em}\geq \energy(v_{j+n+2}, \relativebud(v_{j+n+2}, B_{j+n+2})) \text{~[Directly from Claim~\ref{clm: singletransitioninvariant}]}
				\end{align*} 
				Since \(\energy\) only assigns non-negative value by construction, it implies that if the energy at the beginning of such a patch is at least \(\energy(v_j , \relativebud(v_j, B_j))\), it stays non-negative throughout the patch. 
				
				There are at most \(k\) such patches. 
				Let us denote \(E_1 = k \cdot \max_{v \in V}\energy(v, \preThresh(v))\), and note that \( \max_{v \in V}\energy(v, \preThresh(v)) \geq \energy(v_n, \relativebud(v_n, B_n))\) for any \(n \geq 0\). 
				Therefore, having initial energy at least \(E_1\) ensures that at the beginning of each such patch \(h\), the accumulated energy remains at least \(\energy(v_j, \relativebud(v_j, B_j))\), which, in turn, will keep the accumulated energy non-negative throughout the patch. 
				Moreover, there are at most \(k\) many transition edges between patches, which can consume at most \(E_2 = k \cdot \maxweight\) energy. 
				Since, the initial energy of \(\pi\) is \(M  \geq E_1 + E_2 + \max_{v \in V}\energy(v, \preThresh(v))\), when it eventually reaches \(\zug{v_N, B_N}\), the accumulated energy remains at least \(\max_{v \in V}\energy(v, \preThresh(v)) \geq \energy(v_N, \relativebud(v_N, B_N))\). 
			\end{itemize}
			 
\stam{
			 This concludes the argument that \(M\) is a sufficient initial energy which keeps the accumulated energy throughout \(\pi\) non-negative. 
			 Since, \(\pi\) is an arbitrary play consistent with \(\presagn\), this shows that \Pres wins the energy game from any configuration \(\zug{v, B}\) (with \(B \geq \preThresh(v)\)) by playing a budget agnostic strategy \(\presagn\).  
}


%

\section{Proof of Lemma~\ref{lemma: sufficientnumber}}\label{app:sufficientnumber}
\sufficientnumber*

\begin{proof}
	Assume towards contradiction that, there exists \(e\) such that for every integer \(m\), \(\energy_{m}(v, k^* \ominus \conThresh(v)) \leq e\). 
	In other words, for any integer \(m\), to keep the energy non-negative for \(m\) steps from \(\zug{v, B}\), \Pres needs at most \(e\) as the initial energy. 
	Note that, this does not immediately mean \Pres has a strategy to keep the energy level non-negative forever from \(\zug{v, B}\), even when the initial energy is \(e\).
	
	On the other hand, for any finite initial energy \(M\), \Cons has a strategy from \(\zug{v, B}\) to eventually make the energy level negative. 
	First, we compute existence of a finite upper bound on the number of steps \Cons needs in order to make the energy level negative if the initial energy at most \(M\). 
	Suppose the upper bound is \(n_M\). 
	Therefore, by definition \(\energy_{n_M}(v, B) > M\).
	If we substitute \(M = e\), \(n_M\) contradicts to our assumption that for all integers \(m\), \(\energy_{m} < e\). 
	In the following, we show such an upper bound \(n_M\) exists for a given finite energy \(M\). That suffices to prove the claim. 
	
	We consider an integer \(n\) such that \(\energy_{n}(v, B) > \upperboundforloop + 1\). 
	Such an \(n\) exists simply because \(\energy(v, B) = \infty\). 
	Since, \(\energy_m(v, B) \leq mW\) for any integer \(m\), we have \(n > |V|k\). 
	Recall from the proof of Lemma~\ref{lemma: conswinsinfiniteenergy}, we get a winning strategy \(\sVI{\tau}\) in \(\G\) from using \(\sVIn{\tau}{n}\) (referred to as $ \tau_n $ there) in a particular way.    
	Consider any play \(\pi\) that is consistent with \(\sVI{\tau}\). 
	We also observed in the same proof that \(\pi\) consists of only negative cycles apart from at most \(n-1\) steps.
	Each (negative) cycle is of at most \(|V|k\) lengths, and the energy level decreases by at least \(1\) in each negative cycle. 
	Thus, the energy level in \(\pi\), which was initiated with \(M\) by assumption, becomes negative in at most \((n-1) + M\upperboundforloop\) steps. 
	Hence, an upper bound on the number of steps, \(n_M\), exists for any finite energy \(M\).
\end{proof}


\section{Proof of Lemma~\ref{lemma: properties}}\label{app: properties}
\properties*

\begin{proof}
	
	We establish the three properties one by one in the following:
	
	\begin{itemize}
		
		\item 
		First of all, since $ \sVIn{\tau}{n} $ is a winning strategy in a reachability bidding game, it is memoryless~\cite{DP10}.
		Since \Cons wants to reach a vertex with negative energy in $ \R_n $, $ \tup{x, B'} $ is at least as favourable configuration for \Cons as $ (\tup{v, \energy_{m}', m}, B') $, if not more. 
		Moreover, if any \Cons-action that is available at $ (\tup{v, \energy_{m}', m}, B') $ is also avaialable at $ \tup{x, B'} $ for \Cons. 
		Therefore, any action that she plays at $ (\tup{v, \energy_{m}', m}, B')  $, she can play at $ \tup{x, B'} $ as well, and still win. 
		
		\item 
		Second, both the inequalities ensure that the energy level stays below the energy threshold that \Cons can consume from that configuration.
		Either of the inequality is violated by \(\sVIn{\tau}{N}\) implies \Pres has a response at \(\zug{x, \thresh{x}{\R_N}}\), which would make the next configuration \((m-1)\)-winning for \Pres, contradicting the assumption that \(\sVIn{\tau}{N}\) is a \Cons winning strategy.  
		
		Formally, if \(e' = e + \weightf(v, u) \geq \energy_{m-1}'(u, B' \ominus b)\) holds, then \Pres' response at \(\zug{x, B'}\) should be \(\zug{0, u}\). 
		As a result, the subsequent configuration would be \(\zug{\tup{u, e', m-1}, B \ominus b}\), which is \((m-1)\)-winning for \Pres. 
		
		Similarly, if \(e'' = e + \weightf(v, v') \geq \energy_{m-1}(v', B' \oplus \left(\succb{b}\right))\) for some \(v' \in \neighbor{v}\), then \Pres's winning response at \(\zug{x, B'}\) would be \(\zug{\succb{b}, v'}\). 
		As a result, \Pres wins the current bidding, and
		the subsequent configuration would become \(\zug{\tup{v', e'', m-1}, B' \oplus \left(\succb{b}\right)}\), which is again \((m-1)\)-winning for \Pres.
		
		Existence of such  \Pres' winning response to \(\sVIn{\tau}{N}\) contradicts the assumption that \(\sVIn{\tau}{N}\) is a winning strategy for \Cons, thus the stated inequalities must hold. 
		
		\item Finally, we consider \(B' = \conThresh(v)\), the other case is dual. 
		Suppose \(\sVIn{\tau}{n}(\tup{v, e, n}, \conThresh(v)) = \zug{b, u}\), where $ e = \energy_{n}'(v, \conThresh(v)) $. 
		
		We argue for the bids. 
		We show contradiction to each of the following cases: (1) \(b\) is \(\succb{b^{\conThresh}(v)}\), (2) \(b\) is strictly greater than $ \succb{b^{\conThresh}(v)} $, and (3) finally \(b\) is strictly less than \(\predb{b^{\conThresh}(v)}\). 
		
		First, since \(\sVIn{\tau}{n}\) is a winning strategy from reachability bidding game \(\R_N\), its proposed bid always agree with the budget, thus \(b\) cannot be \(\succb{b^{\conThresh}(v)}\). 
		
		Second, assume for the sake of contradiction, \(b \geq b^{\conThresh}(v) \oplus 1\). 
		We consider the scenario where \Pres lets \Cons win the current bid. 
		Then, at the subsequent configuration \Cons' bid will be 
		\begin{align*}
			\conThresh(v) \ominus \left(\optbid{\conThresh}(v) \oplus 1\right) = \predb{\conThresh(u)}
		\end{align*}
		Thus, applying the second iten of the lemma, we get 
		\begin{align*}
			\energy_{n-1}'(u, \predb{\conThresh(u)}) 
			&\geq \energy_{n}'(v, \conThresh(v)) + \weightf(v, u)\\
			&= \energy_{n}(v, \predb{\preThresh(v)}) -1 + \weightf(v, u)\\
			&\geq \upperboundforloop + 1 - 1+ \weightf(v, v') = \upperboundforloop + \weightf(v, u)
		\end{align*}
		
		Therefore, 
		\begin{align}\label{ineq: lower1}
			\energy_{n-1}(u, \preThresh(u)) > \upperboundforloop + \weightf(v, u)
		\end{align}
		
		On the other hand, from \(\zug{v, \preThresh(v)}\), \Pres \(\energy_{n}(v, \preThresh(v))\)-wins in \(\G_n\) by playing \(\sagn\) that plays \(\zug{\optbid{\preThresh}(v), u}\), where \(u \in \Allowd{\preThresh}(v)\). 
		Consider the scenario when \Cons wins the current bidding. 
		Since, \Cons can choose any vertex, in particular we consider the scenario when he chooses \(u\). 
		Thus, 
		\begin{align*}
			\energy_{n}(v, \preThresh(v)) + \weightf(v, u) \geq \energy_{n-1}(u, \preThresh(v) \oplus  \succb{\optbid{\preThresh}(v)})
		\end{align*}
		Since \(\preThresh(v) \oplus \succb{\optbid{\preThresh}(v)} \geq \preThresh(u)\), we have 
		\begin{align}\label{ineq: upper1}
			\energy_{n-1}(u, \preThresh(u)) &\leq \energy_{n}(v, \preThresh(v)) + \weightf(v, u)\nonumber\\
			&< \upperboundforloop + 1 + \weightf(v, u)
		\end{align}
		
		Combining \eqref{ineq: lower1} and \eqref{ineq: upper1}, we get 
		\begin{align*}
			\upperboundforloop + \weightf(v, u) < \energy_{n-1}(u, \preThresh(u)) < \upperboundforloop + 1+ \weightf(v, u)
		\end{align*}
		This is not possible since \(\energy\) is a integer valued function. 
		Thus, \(b\) cannot be strictly greater than \(\succb{\optbid{\conThresh}(v)}\). 
		
		Third, assume towards contradiction, \(b < \predb{b^{\conThresh}(v)}\). 
		\Pres should bid \(\predb{b^{\conThresh}(v)}\) to win the bidding and choose some \(v' \in  \Allowd{\preThresh}(v)\). 
		Subsequently, \Cons' budget would be 
		\begin{align*}
			\conThresh(v) \oplus \left(\predb{b^{\conThresh}(v)}\right) \leq \predb{\conThresh(\vplus)}
		\end{align*}
		Recall that \(v' \in \Allowd{\preThresh}(v)\) is actually a vertex \(\vplus{\conThresh}\). 
		Thus, as before, we have 
		\begin{align}\label{ineq: lower2}
			\energy_{n-1}(v', \preThresh(v')) > \upperboundforloop + \weightf(v, u)
		\end{align}
		
		On the other hand, from \(\zug{v, \preThresh(v)}\), \Pres \(\energy_{n}(v, \preThresh(v))\)-wins in \(\G_n\). 
		We consider the scenario when \Pres wins the current bidding by \(\sagn\), which bids \(b^{\preThresh}(v)\) and chooses a \(v' \in \Allowd{\preThresh}(v)\). 
		Thus, 
		\begin{align*}
			\energy_{n}(v, \preThresh(v)) + \weightf(v, v') \geq \energy_{n-1}(v', \preThresh(v) \ominus b^{\preThresh}(v))
		\end{align*}
		Since, \(\preThresh(v) \ominus b^{\preThresh}(v) \geq \preThresh(v')\), we have 
		\begin{align}\label{ineq: upper2}
			\energy_{n-1}(v', \preThresh(v')) &\leq \energy_{n}(v, \preThresh(v)) + \weightf(v, v')\nonumber\\
			&< \upperboundforloop+ \weightf(v, v')
		\end{align}
		Combining ~\eqref{ineq: lower2} and \eqref{ineq: upper2}, we get 
		\begin{align*}
			\upperboundforloop + \weightf(v, v') < \energy_{n-1}(v', \preThresh(v')) < \upperboundforloop + 1 + \weightf(v, v')
		\end{align*}
		
		Again a contradiction! Thus, \(b\) cannot be strictly less than \(b^{\conThresh}(v)\) either. 
		
		\medskip
		
		Finally, we argue for the choice of vertex. 
		We derive similar contradiction, now by using the fact that \Cons' bid is indeed \(b^{\conThresh}(v)\) at \(\zug{v, \conThresh(v)}\). 
		Assume towards contradiction, \(u \not \in \Allowd{\conThresh}{v}\). 
		We consider the scenario where \Pres lets \Cons win the bidding, thus \Cons' budget becomes 
		\begin{align*}
			\conThresh(v) \ominus b^{\conThresh}(v) = 
			\begin{cases}
				\conThresh(\vminus) &\text{\(\conThresh(\vminus)\) is in \(\Nat\)}\\
				\absolut{\conThresh(\vminus)} +1 &\text{otherwise}
			\end{cases}
		\end{align*}
		
		From Def.~\ref{def: allowedvertices}, it turns out \(\conThresh(u) > \conThresh(v) \ominus b^{\conThresh}(v)\). 
		Thus, arguing similar to bids, we get 
		\begin{align*}
			\upperboundforloop + \weightf(v, u) < \energy_{n-1}(u, \preThresh(u)) < \upperboundforloop + 1 + \weightf(v, u)
		\end{align*}
		
		Thus, it implies that $ u \in \Allowd{\conThresh}(v) $. 
	\end{itemize}

\end{proof}


\section{Proof of Lem.~\ref{lemma: tau-winning-aux}}\label{app:  tau-winning-aux}
\tauwinningaux*
			
				\begin{proof}
				
				Consider a play $ \pi' = \zug{v_0, B_0'}, \zug{v_1, B_1'} \ldots$ consistent with $\tagn'$, and the initial energy $ e_0 \geq 2\upperboundforloop$.
				Denote the energy level at $ i^{th} $ step as $ e_i $. 
				If the energy level \(e_i\) is already $\leq \upperboundforloop $, we are done.
				Otherwise, we assume that \(e_i > \upperboundforloop\).  
				Therefore, by Lem.\ref{lemma: properties}, $\tagn'$ bids $ \optbid{\conThresh}(v_i, B_i') $, and choose $ u_i \in \Allowd{\conThresh}(v_i) $ for every $ i $. 
				 By the average property of $ \conThresh $, we have $ B_i' \geq \conThresh(v_i) $ for all $ i $. 
				Like in the proof of Thm.~\ref{thm:sagn-is-winning}, we define a function $ \triangledown_{\pi'}: \Nat \rightarrow [k] $ that intuitively assign \Cons spare change. 
				Formally, $ \spare{i}{\pi'} = \absolut{B_i' \ominus \conThresh(v_i)} $. 
				$ \spare{i}{\pi'} $ is monotonically non-decreasing, and bounded above by $ k $, thus it eventually stabilises; let $ m \geq 0 $ and $ r \in \Nat $ such that $ \spare(i) = r $ for $ i \geq m $.  We now establish an energy invariant that $\tagn'$ maintains at each transition where the spare change is fixed.

				\begin{restatable}{claim}{conssingletransition}\label{clm: conssingletransition}
				For all \(i \geq m\), and $ n = P(e_i) $, we have $ \energy'_{n}(v_i, \relativebud(v_i, B_i)) +  \weightf(v_i, v_{i+1}) \leq \energy'_{n-1}(v_{i+1}, \relativebud'(v_{i+1}, B_{i+1}))$. 
				\end{restatable}
			
				\begin{proof}[Proof of the Claim]
						First of all, combining the second and third item of Lem.~\ref{lemma: properties}, we get the following energy invariant: 
						\begin{itemize}
							\item If \Cons wins by bidding $ \optbid{\conThresh}(v, B') $, and proceeds to a $u \in \Allowd{\conThresh}(v)  $, then 
							\begin{align*}
								\energy_n'(v, \relativebud'(v, B')) + \weightf(v, u) \leq \energy_{n-1}'(v, \relativebud'(v, B') \ominus \optbid{\conThresh}(v, B))
							\end{align*}
							
							\item If \Pres wins by bidding at least $ \succb{\optbid{\conThresh}(v, B')} $, then for every $ v' \in \neighbor{v} $, 
							\begin{align*}
								\energy_{n}'(v, \relativebud'(v, B')) + \weightf(v, v') \leq \energy_{n-1}'(v', \relativebud(v, B') \oplus (\succb{\optbid{\conThresh}(v, B')})) 
							\end{align*}
						\end{itemize}
						However, \(\relativebud(v_{i+1}, B'_{i+1})\) may not coincide with either of \(\relativebud(v_i, B'_i) \ominus \optbid{\conThresh}(v_i, B'_i)\) and \(\relativebud(v_i, B'_i) \oplus \succb{\optbid{\conThresh}(v_i, B'_i)}\), thus we do not immediately get to establish the claim. 
						
						\medskip
						
						We argue by a case analysis that, in fact, in all but one cases they do coincide, and the case where they may not coincide is the only case where the equality \(\spare{i}{\pi'} = \spare{i+1}{\pi'}\) does not hold, contradicting the hypothesis of the claim. 
						The case analysis uses Lem.~\ref{lem:Talgebra}, and exact same as the one in App~\ref{app:singletransitioninvariant}, substituting $ \preThresh $ by $ \conThresh $. 
				\end{proof}

				We suppose the energy level at step $ m $ in $ \pi' $ is $ e_m > \upperboundforloop$. 
				We apply Claim.~\ref{clm: conssingletransition}, for every step $ i \geq m$, until the energy level becomes $ \leq \upperboundforloop$, or $\pi' $ traverses $ P(e_m) $ more steps, whichever is sooner. For the latter, 
				if $ \pi' $ traverses all $ P(e_m) $ steps, the energy level would become less than or equal to $ \energy_{0}'(v_{m+N(e_m)}, \relativebud'(v_{m+P(e_m)}, B_{m+P(e_m)}))$, which is $ -1 $. 
			\end{proof}

 \section{Proof that a finite upper-bound exists for $ \tagn' $}\label{app:maximumnumber}
 
 \begin{restatable}{claim}{maximumnumber}\label{clm: maximumnumber}
 	There exists a finite upper bound $ \psi(e, v) $ on the number of steps by which from configuration $ \zug{v, \conThresh(v)} $, $ \tagn' $ can reduce an initial energy $ e_0 \geq 2\upperboundforloop $ to an energy level less than or equal to $ \upperboundforloop$. 
 \end{restatable}
	
\begin{proof}
	We consider the play $ \pi = \zug{v_0, B_0'}, \zug{v_1, B_1'} \ldots$ consistent with $ \tagn' $. 
	Note that, $ \pi $ can be partitioned into "patches" such that $ \triangledown_{\pi} $ is stable within each patch and changes between patches. 
	Thus, we can apply Clm.~\ref{clm: conssingletransition} for a single patch. 
	We first upper bound the length of each patch.
	We claim that if the energy level at the beginning of $ m^{th} $ patch is $ e_m $, then the patch cannot be more than of length $ P(e_m) $, except for the last one, for $ 0 \leq m \leq k $. 
	This is because  if $ \pi' $ traverses a whole patch of $ P(e_m) $ steps, the energy invariant of Clm~\ref{clm: conssingletransition} makes the energy $ -1 $, which is already below $ \upperboundforloop $, implying this can happen only at the last patch. 
	There are at most $ k-1 $ transitions in $ \pi $, which are outside of any patch. 
	Moreover, the energy level at the beginning of each patch is also bounded; namely $ e_1 $ is less than or equal to $ e_0 + \maxweight \cdot P(e_0) $, $ e_2 $ is less than or equal to $ e_1 + \maxweight\cdot P(e_1) $, and so on. 
	This comes from the fact that one edge can be of at most energy $ W $.
	We define $ \psi(e_0, v) = \sum_{0 \leq m}^{k} P(e_m)$. 
	Since $ P(e) $ for any finite energy $ e $ is finite, $ \psi(e, v) $ provides the maximum number of steps by which $ \tagn' $ makes the energy level $ \leq \upperboundforloop $, from initial energy $ \geq 2\upperboundforloop $. 
\end{proof}

\section{Construction of $G_{T,\G}$}
\label{app:turn-based-construction}
Given an energy bidding game $\G$ and $T$ that satisfies the average property, we construct a turn-based energy game $G_{T, \G} = \zug{V_{\Pres}, V_{\Cons}, E, \weightf'}$ as follows:
\begin{itemize}
\item \Pres vertices are $V_{\Pres} = \set{\zug{v, B}: v \in V \text{ and } B \in \set{T(v), \succb{T(v)}, \top}}$, that is we make three copies of each vertex in $\G$.
\item \Cons vertices are $V_{\Cons} = \set{\zug{v', c}: v' \in V \text{ and } c \in V_\Pres}$.
\item The edges are as follows.
\begin{itemize}
\item Outgoing edges from \Pres vertices are as follows. Every $\zug{v,\top}$ is a winning sink for \Pres; its only outgoing edge is a self loop with weight $0$. For every $\zug{v, B} \in V_\Pres$ with $B \neq \top$ and $v' \in \Allowd{T}(v)$, we have an edge from $\zug{v,B}$ to $\zug{v', \zug{v,B}}$. 
\item Outgoing edges from \Cons vertices are as follows. Consider $\zug{u, \zug{v,B}}$, which intuitively corresponds to \Pres bidding $\zug{\optbid{T}(v, B), v'}$. 
 \Cons edges correspond to a choice between (1)~letting \Pres win the bidding by proceeding to $\zug{v', B \ominus \optbid{T}(v, B)}$ or (2)~win the bidding and choosing the successor vertex $u$, formally, $\zug{u, \tilde{B}} \in N(\zug{v', \zug{v, B}})$, for $u \in N(v)$ and $\tilde{B} = B \oplus \succb{\optbid{T}(v, B)}$ if $\tilde{B} \in \set{T(u), \succb{T(u)}}$, otherwise $\tilde{B} = \top$. 
 \end{itemize}
 \item We define the weight function \(\weightf'\) in $G_{T,\G}$. 
Weights on outgoing edges from \Cons vertices simulate the weight in $\G$, formally \(\weightf'(\zug{v', \zug{v,B}}, \zug{u, \tilde{B}}) = \weightf(v, u)\). Edges from $\zug{v, B}$ to $\zug{v', \zug{v, B}}$ have weight $0$. Thus, a play in $G_{T, \G}$ that does not end in a sink, corresponds to a play in $\G$ that traverses the same sequence of weights. 
\end{itemize}

\section{Proof of Lemma~\ref{lemma: soundness}}\label{app:soundness}
\soundness*

\begin{proof}
	Suppose \(f\) is a memoryless winning strategy by which \Pres \(\energy(w)\)-wins from any \(w\) of \(G_{T, \G}\). 
	We describe a winning strategy \(f'\) in \(\G\). 
	Suppose that \(G\) starts from a configuration \(c_0 = \zug{v, T(v)}\). 
	We initiate \(G_{T, \G}\) from \Pres vertex \(w_0 = \zug{v, T(v)}\). 
	Suppose that \Cons plays according to a strategy \(g'\) in \(\G\). 
	We simulate \(g'\) by a \Cons strategy \(g\) in \(G_{T, \G}\) such that when a configuration \(c = \zug{u, B_1}\) is reached in \(\G\), the vertex in \(G_{T, \G}\) would be \(w = \zug{u, B}\) where \(B \in \set{T(u), \succb{T(u)}}\). 
	We describe how we simulate \(g'\) with a \(g\), and \(f'\) simulates \(f\). 
	Suppose that \(\G\) is in a configuration \(c = \zug{v, B_1}\), and \(G_{T, \G}\) is at a vertex \(w = \zug{v, B}\). 
	We define \(f'\) to agree with \(f_T\), and bid \(b^T(v)\). 
	Moreover, if \(f\) chooses \(\zug{v', c}\) from \(w\) in \(G_{T, \G}\), we let \(f'\)  choose \(v'\) as the next vertex. 
	If \(g'\) looses the bidding, then we define \(g\) to proceed to \(w_1 = \zug{v', B \ominus \optbid{T}(v, B)}\), and we define \(f'\) to match the move of \(f\) from \(w_1\). 
	If \(g'\) wins the bidding, and chooses \(u\) as the next vertex, we define \(g\) to proceed to \(\zug{u, B'}\) if \(B' = B \oplus \succb{b^T(v)} \leq \succb{T(u)}\). 
	Otherwise we define \(g\) to proceed to \(\zug{u, \top}\). 
	
	Let \(\pi\) be a play in \(G_{T, \G}\) that results from \(f\) and \(g\), and \(\pi'\) is a play in \(\G\) that results from \(f'\) and \(g'\). 
	Since, \(f\) is a \Pres winning strategy, she \(\energy(w)\)-wins in \(\pi\). 
	We distinguish two cases. 
	In the first case, \(\pi\) does not reach a sink vertex in \(\G_{T, \G}\). 
	In that case, \(\energy(w) + \sum_{e \in \pi} \weightf'(e) = \energy(w) + \sum_{e \in \pi'} \weightf(e') \geq 0\), thus \Pres \(\energy(w)\)-wins in \(\pi'\) as well. 
	In the second case, \(\pi\) reaches a sink vertex \(\zug{u, \top}\) in \(G_{T, \G}\).  
	Let \(\zug{u, B_1}\) is the corresponding configuration in \(\G\). 
	Note that, \(B_1\) is strictly greater than \(\succb{T(u)}\). 
	Let \(B' < B_1\) is a budget whose tie-breaking status agrees with \(B_1\), and \(\zug{u, B'} \in V_1\). 
	Intuitively, this means a "spare change" of \(B_1 - B'\) gets added in \Pres budget, and plays as if her budget is \(B'\) by restarting the game from configuration \(\zug{u, B'}\). 
	Since the total budget is bounded above by \(k\), and restarting \(f\) in this manner strictly increases the spare change, it can happen at most \(k\)-times. 
	Moreover, every such restart consumes at most \(\max_{w \in G_{T, \G}} \energy(w)\). 
	Thus the initial energy \(M\) can facilitate \(k\)-start restart, and still \Pres would be left with enough energy to win the game like the first scenario. 
	In other words, \Pres \(M\)-wins from \(\zug{v, T(v)}\) by \(f'\). 
\end{proof}
	
	\section{Proof of Lemma~\ref{lemma: completeness}} \label{app: completeness}
	\completeness*
	
	\begin{proof}
		Assume towards contradiction that \(T \equiv \preThresh\) but there exists a vertex \(w\) such that \(\energy(w) = \infty\). 
		That is, for any finite energy \(M\), there exists a \Cons strategy \(g'\) by which he \(M\)-wins from \(w\) in \(G_{T, \G}\). 
		Without loss of generality, we can assume \(w\) is a \Pres vertex. 
		This is because  if there is a \Cons \(\zug{v', c}\), energy of which is \(\infty\), it must have a successor \Pres vertex whose energy is \(\infty\). 
		Therefore, we assume \(w = \zug{v, B}\) of \(G_{T, \G}\). 
		
		Since \(B \geq T(v)\), and \(T \equiv \preThresh\), we have \(\energy(v, B) < \infty\). 
		Thus, there exists a finite \(M\) and a budget-agnostic strategy \(f\) by which \Pres \(M\)-wins from \(\zug{v, B}\).
		To be specific, \(M = (k+1) \cdot (\maxweight + \max_{v \in V} \energy(v, T(v)))\) (appears in the proof of Lemma~\ref{thm:sagn-is-winning}).
		Moreover, \(f\) bids \(\optbid{T}(v, B)\) and chooses the next vertex from \(\Allowd{T}(v)\). 
		We will reach a contradiction by constructing a \Cons strategy \(g\) based on \(g'\), that counters \(f\), thus showing that \(f\) is not winning. 
		Recall that a winning \Pres strategy can be thought of a strategy that, in each turn, reveals \Pres' action first, and allows \Cons to respond.
		
		We construct a \Pres strategy \(f'\) in \(G_{T, \G}\), based on 
		the winning strategy \(f\). 
		Then, however \(g'\) responds to \(f'\), we simulate that in constructing \(g\) against \(f\). 
		The construction on both side is straightforward. 
		We begin from \Pres vertex \(w = \zug{v, B}\) of \(G_{T, \G}\), and the corresponding configuration configuration \(c = \zug{v, B}\) of \(G\). 
		Since \(f\)'s action is \(\zug{b^T(v), u}\), where \(u \in \Allowd{T}(v)\), at \(c\), we define \(f'\) to choose \(\zug{u, w}\). 
		If intuitively \(g'\) chooses to let \Pres win the bidding, we define \(g\) to bid \(0\), and choose any neighbouring vertex of \(v\). 
		On the other hand, if \(g'\) chooses a vertex \(\zug{v', B \oplus \succb{b^T(v)}}\), then we define \(g\) to play the action \(\zug{\succb{b^T(v), v'}}\), and so on. 
		
		Suppose \(\pi\) be the play that results from \(f\) and \(g\) in \(\G\), and \(\pi'\) be the play that results from \(f'\) and \(g'\) in \(\G_{T, \G}\). 
		If \(\pi'\) does not reach a sink vertex, then for every finite prefix \(h\) of \(\pi\), there is a corresponding finite prefix \(h'\) of \(\pi'\), and vice-versa, such that \( \sum_{e \in h} \weightf(e) = \sum_{e' \in h'} \weightf'(e')\). 
		Since, there must exists a finite prefix \(h'\) for which \(M + \sum_{e' \in h} \weightf'(h') < 0\) but for all finite prefix \(h\), \(M + \sum_{e \in h} \weightf(e) \geq 0\), we get a contradiction! 
		Finally, if \(\pi'\) reaches a sink vertex in \(G_{T, \G'}\), then \Pres \(E\)-wins in \(\pi'\) for some finite energy \(E\), thus showing a play consistent with \Cons winning strategy \(g'\) from \(w\) but finite-energy winning for \Pres. 
		Contradiction to the assumption that \(\energy(w) = \infty\).
		
		Therefore, \(\energy(w) < \infty\) for all vertex \(w\) if \(T \equiv \preThresh\). 
	\end{proof}
\end{document}